\documentclass[sigplan,screen]{acmart}

\setcopyright{none}
\renewcommand\footnotetextcopyrightpermission[1]{} 
\usepackage{array}
\newcolumntype{L}[1]{>{\raggedright\let\newline\\\arraybackslash\hspace{0pt}}m{#1}}
\newcolumntype{C}[1]{>{\centering\let\newline\\\arraybackslash\hspace{0pt}}m{#1}}
\newcolumntype{R}[1]{>{\raggedleft\let\newline\\\arraybackslash\hspace{0pt}}m{#1}}

\usepackage{booktabs}   
\usepackage{subcaption} 

\usepackage{enumerate}
\usepackage{tabu}

\usepackage{comment}
\usepackage{multirow}
\usepackage{graphicx}

\usepackage{color}
\usepackage[T1]{fontenc}
\usepackage{textcomp}
\usepackage{booktabs}

\usepackage{tikz,pgffor}
\usetikzlibrary{arrows}
\usetikzlibrary{shapes}
\usetikzlibrary{calc}
\usetikzlibrary{automata}
\usetikzlibrary{positioning}
\usepackage{mathtools}

\tikzstyle{proglabel}=[shape=circle,draw,inner sep=0pt,minimum size=5mm]
\tikzstyle{tran}=[draw,->,>=stealth, rounded corners]

\usepackage{amsmath}
\usepackage{amssymb}
\usepackage{amsthm}

\usepackage{stmaryrd}
\usepackage{url}                  
\usepackage{paralist}
\usepackage{listings}
\lstdefinelanguage{prog}
{
morekeywords={if, then, else, fi, while, do, od, true, false, and, or, skip, prob, tick},
sensitive = false
}

\usepackage{array}
\newcolumntype{C}[1]{>{\centering\arraybackslash\hspace{0pt}}p{#1}}
\usepackage{makecell}

\newcommand{\Rset}{\mathbb{R}}
\newcommand{\Nset}{\mathbb{N}}

\newcommand{\pvars}{V_\mathrm{p}}
\newcommand{\rvars}{V_{\mathrm{r}}}

\newcommand{\expv}{\mathbb{E}}
\newcommand{\pv}{\mathbf{v}}
\newcommand{\rv}{\mathbf{u}}

\newcommand{\loc}{\ell}
\newcommand{\locs}{\mathit{L}}
\newcommand{\blocs}{\mathit{L}_{\mathrm{b}}}
\newcommand{\alocs}{\mathit{L}_{\mathrm{a}}}
\newcommand{\plocs}{\mathit{L}_{\mathrm{p}}}
\newcommand{\tlocs}{\mathit{L}_{\mathrm{t}}}

\newcommand{\Dlocs}{\mathit{L}_{\mathrm{nd}}}
\newcommand{\transitions}{{\rightarrow}}

\newcommand{\lin}{\loc_\mathrm{in}}
\newcommand{\lout}{\loc_\mathrm{out}}
\newcommand{\val}[1]{\mbox{\sl Val}_{#1}}

\newcommand{\probm}{\mathbb{P}}
\newcommand{\condexpv}[2]{{\expv}{\left({#1}{\mid}{#2}\right)}}

\newcommand{\sat}[1]{\langle #1 \rangle}
\newcommand{\monoid}{\mbox{\sl Monoid}}

\newcommand{\pre}{\mathrm{pre}}

\newcommand{\supval}{\mbox{\sl supval}}
\newcommand{\handelmanformat}{(\dagger)}

\newcommand{\pspace}{(\Omega,\mathcal{F},\probm)}
\newcommand{\setR}{\mathbb{R}}
\newcommand{\setN}{\mathbb{N}}

\newcommand{\run}{\{ (\loc_n, \pv_n) \}_{n=0}^\infty}

\theoremstyle{acmdefinition}
\newtheorem{thm}{Theorem}

\newtheorem{remark}[thm]{Remark}

\newcommand{\rd}[1]{{#1}}

\usepackage{paralist} 

\usepackage{graphicx}

\begin{document}

\title[Cost Analysis of Nondeterministic Probabilistic Programs]{Cost Analysis of Nondeterministic\\Probabilistic Programs}         
\titlenote{A conference version of this paper appears in~\cite{confver}.}             


\author{Peixin Wang}
\affiliation{
  \institution{Shanghai Jiao Tong University}            
}
\affiliation{
	\department{Shanghai Key Laboratory of Trustworthy Computing}              
	\institution{East China Normal University}            
	\city{Shanghai}
	\country{China}                    
}
\email{wangpeixin@sjtu.edu.cn}          

\author{Hongfei Fu}
\authornote{Corresponding Author}          
\affiliation{
	\institution{Shanghai Jiao Tong University}            
}
\affiliation{
	\department{Shanghai Key Laboratory of Trustworthy Computing}              
	\institution{East China Normal University}            
	\city{Shanghai}
	\country{China}                    
}
\email{fuhf@cs.sjtu.edu.cn}

\author{Amir Kafshdar Goharshady}
\authornote{Recipient of a DOC Fellowship of the Austrian Academy of Sciences (\"{O}AW)}          
\affiliation{
  \institution{IST Austria}           
  \city{Klosterneuburg}
  \country{Austria}                   
}
\email{ amir.goharshady@ist.ac.at }         

\author{Krishnendu Chatterjee}
\affiliation{
	\institution{IST Austria}           
	\city{Klosterneuburg}
	\country{Austria}                   
}
\email{ krishnendu.chatterjee@ist.ac.at }

\author{Xudong Qin}
\affiliation{
	\institution{East China Normal University}           
	\city{Shanghai}
	\country{China}                   
}
\email{ 52174501019@stu.ecnu.edu.cn}

\author{Wenjun Shi}
\affiliation{
	\institution{East China Normal University}           
	\city{Shanghai}
	\country{China}                   
}
\email{51174500041@stu.ecnu.edu.cn }         

\renewcommand{\shortauthors}{P. Wang, H. Fu, A.K. Goharshady, K. Chatterjee, X. Qin and W. Shi}

\begin{abstract}
We consider the problem of expected cost analysis over nondeterministic probabilistic programs,
which aims at automated methods for analyzing the resource-usage of such programs.
Previous approaches for this problem could only handle nonnegative bounded costs.
However, in many scenarios, such as queuing networks or analysis of cryptocurrency protocols,
both positive and negative costs are necessary and the costs are unbounded as well.

In this work, we present a sound and efficient approach to obtain polynomial bounds on the
expected accumulated cost of nondeterministic probabilistic programs.
Our approach can handle (a)~general positive and negative costs with bounded updates in
variables; and (b)~nonnegative costs with general updates to variables.
We show that several natural examples which could not be
handled by previous approaches are captured in our framework.

Moreover, our approach leads to an efficient polynomial-time algorithm, while no
previous approach for cost analysis of probabilistic programs could guarantee polynomial runtime.
Finally, we show the effectiveness of our approach using experimental results on a variety of programs for which we efficiently synthesize tight resource-usage bounds.
\end{abstract}


\maketitle

\renewcommand{\paragraph}[1]{\smallskip\noindent\textbf{\textit{#1}}}

\section{Introduction}

In this work, we consider expected cost analysis of nondeterminisitic
probabilistic programs, and present a sound and efficient approach
for a large class of such programs.
We start with the description of probabilistic programs
and the cost analysis problem, and then present our contributions.

\paragraph{Probabilistic programs.}
Extending classical imperative programs with randomization, i.e.~generation of random values according to a predefined probability distribution, leads to the class of probabilistic programs~\cite{gordon2014probabilistic}.
Probabilistic programs are shown to be powerful models for a wide variety of applications, such as analysis of stochastic network protocols~\cite{netkat,netkat2,netkat3}, machine learning applications~\cite{roy2008stochastic,gordon2013model,scibior2015practical,claret2013bayesian}, and robot planning~\cite{thrun2000probabilistic,thrun2002probabilistic}, to name a few. There are also many probabilistic programming languages (such as Church~\cite{goodman2008church}, Anglican~\cite{anglican} and WebPPL~\cite{dippl}) and automated analysis of such programs is an active research area in formal methods and programming languages (see~\cite{SriramCAV,pmaf,pldi18,AgrawalC018,ChatterjeeFNH16,ChatterjeeFG16,EsparzaGK12,KaminskiKMO16,kaminski2018hardness}).

\paragraph{Nondeterministic programs.}
Besides probability, another important modeling concept in programming languages is nondeterminism.
A classic example is abstraction: for efficient static analysis of large programs, it is often infeasible
to keep track of all variables. Abstraction ignores some variables and replaces them with worst-case behavior, which is modeled by nondeterminism~\cite{cousot1977abstract}.

\paragraph{Termination and cost analysis.}
The most basic liveness question for probabilistic programs is {\em termination}.
The basic qualitative questions for termination of probabilistic programs, such as, whether the program
terminates with probability~1 or whether the expected termination time is bounded, have been widely studied~\cite{kaminski2018hardness,ChatterjeeFG16,KaminskiKMO16,ChatterjeeFNH16}.
However, in program analysis, the more general quantitative task of obtaining precise bounds on resource-usage is a challenging problem that is of significant interest for the following reasons:
(a)~in applications such as hard real-time systems, guarantees of worst-case behavior are required;
and (b)~the bounds are useful in early detection of egregious performance problems in large code bases.
Works such as~\cite{SPEED1,SPEED2,Hoffman1,Hoffman2} provide excellent motivation for the study of automated methods to obtain worst-case
bounds for resource-usage of nonprobabilistic programs.
The same motivation applies to the class of probabilistic programs as well.
Thus, the problem we consider is as follows: given a probabilistic program with costs associated to each execution step,
compute bounds on its expected accumulated cost until its termination. \rd{Note that several probabilistic programming languages have observe statements and conditioning operators for limiting the set of valid executions. In this work, we do not consider conditioning and instead focus on computing the expected accumulated cost over \emph{all} executions. See Remark~\ref{rem:conditioning}.}

\paragraph{Previous approaches.}
While there is a large body of work for qualitative termination analysis problems
(see Section~\ref{sec:rel} for details), the cost analysis problem has only been considered recently.
The most relevant previous work for cost analysis is that of Ngo, Carbonneaux and Hoffmann~\cite{pldi18}, which considers the stepwise
costs to be nonnegative and bounded.
While several interesting classes of programs satisfy the above restrictions, there are
many natural and important classes of examples that cannot be modeled in this framework.
For example, in the analysis of cryptocurrency protocols, such as mining, there are both energy costs (positive costs) and solution rewards (negative costs).
Similarly, in the analysis of queuing networks, the cost is proportional to the length of the queues, which might be unbounded.
For concrete motivating examples see Section~\ref{sec:motivation}.

\paragraph{Our contribution.}
In this work, we present a novel approach for synthesis of polynomial bounds on the expected accumulated cost of nondeterministic probabilistic programs.
\begin{compactenum}

\item Our sound framework can handle the following cases:
(a)~general positive and negative costs, with bounded updates to the variables at every step of the execution; and
(b)~nonnegative costs with general updates (i.e.~unbounded costs and unbounded updates to the variables).
In the first case, our approach obtains both upper and lower bounds, whereas in the second case we only obtain upper bounds.
In contrast, previous approaches only provide upper bounds for bounded nonnegative costs.
A key technical novelty of our approach is an extension of the classical Optional Stopping Theorem (OST) for martingales.

\item We present a sound algorithmic approach for the synthesis of polynomial bounds. Our algorithm runs in polynomial time and \rd{only relies on standard tools such as linear programming and linear invariant generation as prerequisites}. Note that no previous approach provides polynomial runtime guarantee for synthesis of such bounds for nondeterministic probabilistic programs.
Our synthesis approach is based on application of results from semi-algebraic geometry.

\item Finally, we present experimental results on a variety of programs, which are motivated from applications such as
cryptocurrency protocols, stochastic linear recurrences, and queuing networks, and show that our approach can efficiently obtain
tight polynomial resource-usage bounds.

\end{compactenum}
We start with preliminaries (Section~\ref{sec:pre}) and then present a set of motivating examples (Section~\ref{sec:motivation}).
Then, we provide an overview of the main technical ideas of our approach in Section~\ref{sec:novelty}. The following sections each present technical details of one of the steps of our approach.

\section{Preliminaries} \label{sec:pre}

In this section, we define some necessary notions from probability theory and probabilistic programs. We also formally define the expected accumulated cost of a program.

\subsection{Martingales}

We start by reviewing some notions from probability theory. We consider a probability space $\pspace$ where $\Omega$ is the sample space, $\mathcal{F}$ is the set of events and $\probm: \mathcal{F} \rightarrow [0, 1]$ is the probability measure.

\paragraph{Random variables.} A \emph{random variable} is an $\mathcal{F}$-measurable function $X: \Omega \rightarrow \setR \cup \{+\infty,-\infty\}$, i.e.~a function satisfying the condition that for all $d \in \setR \cup \{ +\infty, -\infty \}$, the set of all points in the sample space with an $X$ value of less than $d$ belongs to $\mathcal{F}$.

\paragraph{Expectation.} The \emph{expected value} of a random variable $X$, denoted by $\expv(X)$, is the Lebesgue integral of $X$ wrt $\probm$. See~\cite{williams1991probability} for the formal definition of Lebesgue integration. If the range of $X$ is a countable set $A$, then $\expv(X) = \sum_{\omega \in A} \omega \cdot \probm(X = \omega)$.

\paragraph{Filtrations and stopping times.} A \emph{filtration} of the probability space $\pspace$ is an infinite sequence $\{ \mathcal{F}_n \}_{n=0}^{\infty}$ such that for every $n$, the triple $(\Omega, \mathcal{F}_n, \probm)$ is a probability space and $\mathcal{F}_n \subseteq \mathcal{F}_{n+1} \subseteq \mathcal{F}$. A \emph{stopping time} wrt $\{ \mathcal{F}_n \}_{n=0}^{\infty}$ is a random variable $U: \Omega \rightarrow \setN \cup \{0, \infty\}$ such that for every $n \geq 0$, the event $U \leq n$ is in $\mathcal{F}_n$. Intuitively, $U$ is interpreted as the time at which the stochastic process shows a desired behavior.

\paragraph{Discrete-time stochastic processes.} A \emph{discrete-time stochastic process} is a sequence $\Gamma = \{X_n\}_{n=0}^\infty$ of random variables in $\pspace$. The process $\Gamma$ is \emph{adapted} to a filtration $\{ \mathcal{F}_n \}_{n=0}^{\infty}$, if for all $n \geq 0$, $X_n$ is a random variable in $(\Omega, \mathcal{F}_n, \probm)$.

\paragraph{Martingales.} A discrete-time stochastic process $\Gamma=\{X_n\}_{n=0}^\infty$ adapted to a filtration $\{\mathcal{F}_n\}_{n=0}^\infty$ is a \emph{martingale} (resp. \emph{supermartingale}, \emph{submartingale})
if for all $n \geq 0$, $\expv(|X_n|)<\infty$ and it holds almost surely (i.e., with probability $1$) that
$\condexpv{X_{n+1}}{\mathcal{F}_n}=X_n$ (\mbox{resp. } $\condexpv{X_{n+1}}{\mathcal{F}_n}\le X_n$, $\condexpv{X_{n+1}}{\mathcal{F}_n}\ge X_n$).
See~\cite{williams1991probability} for details.

Intuitively, a martingale is a discrete-time stochastic process, in which at any time $n$, the expected value $\condexpv{X_{n+1}}{\mathcal{F}_n}$ in the next step, given all previous values, is equal to the current value $X_n$. In a supermartingale, this expected value is less than or equal to the current value and a submartingale is defined conversely.
Applying martingales for termination analysis is a well-studied technique~\cite{SriramCAV,ChatterjeeFG16,ChatterjeeNZ2017}.

\subsection{Nondeterministic Probabilistic Programs}

We now fix the syntax and semantics of the nondeterministic probabilistic programs we consider in this work.

\paragraph{Syntax.} 
Our nondeterministic probabilistic programs are imperative programs with the usual conditional and loop structures (i.e.~\textbf{if} and \textbf{while}), as well as the following new structures: (a)~probabilistic branching statements of the form ``\textbf{if prob}$(p) \dots$'' that lead to the \textbf{then} part with probability $p$ and to the \textbf{else} part with probability $1-p$, (b)~nondeterministic branching statements of the form ``\textbf{if $\star$} \dots'' that nondeterministically lead to either the \textbf{then} part or the \textbf{else} part, and (c)~statements of the form \textbf{tick}$(q)$ whose execution triggers a cost of $q$. Moreover, the variables in
our programs can either be program variables, which act in the usual way, or sampling variables, whose values are randomly sampled from predefined probability distributions each time they are accessed in the program.

\begin{remark} \label{rem:conditioning}
	\rd{
		
		We remark two points about the probabilistic programs considered in this work:
		
		\begin{compactitem}
		
		\item \emph{Probability Distributions.} We do not limit our sampling variables to any specific type of distributions. Our approach supports any predefined distribution, including but not limited to Bernoulli and binomial distributions, which are used in our examples. 
		
		\item \emph{Conditioning.}
		 We do not consider conditioning and observe statements in this work. Note that as shown in~\cite{DBLP:journals/toplas/OlmedoGJKKM18}, in many cases of probabilistic programs conditioning can be removed. Moreover, our focus is on computing the expected accumulated cost of a program over \emph{all} executions, while conditioning is used to limit the set of valid executions.
		
		\end{compactitem}

}
\end{remark}

Formally, nondeterministic probabilistic programs are generated by the grammar in Figure~\ref{fig:syntax}.
In this grammar $\langle\mathit{pvar}\rangle$ (resp. $\langle\mathit{rvar}\rangle$) expressions range over program (resp. sampling) variables. For brevity, we omit the $\textbf{else}$ part of the conditional statements if it
contains only a single $\textbf{skip}$.
See Appendix~\ref{app:syntax} for more details about the syntax.

 An example program is given in Figure~\ref{fig:example} (left). Note that the complete specification of the program should also include distributions from which the sampling variables are sampled.

\begin{figure}
	\footnotesize{
	\begin{align*}
	\langle \mathit{stmt}\rangle &::= \mbox{`\textbf{skip}'}  \mid \langle\mathit{pvar}\rangle \,\mbox{`$:=$'}\, \langle\mathit{expr} \rangle\\
	& \mid \mbox{`\textbf{if}'} \, \langle\mathit{bexpr}\rangle\,\mbox{`\textbf{then}'} \, \langle \mathit{stmt}\rangle \, \mbox{`\textbf{else}'} \, \langle \mathit{stmt}\rangle \,\mbox{`\textbf{fi}'}
	\\
	& \mid \mbox{`\textbf{if}'} \, \mbox{`\textbf{prob}' `('$p$`)'}\,\mbox{`\textbf{then}'} \, \langle \mathit{stmt}\rangle \, \mbox{`\textbf{else}'} \, \langle \mathit{stmt}\rangle \,\mbox{`\textbf{fi}'}\\
	& \mid \mbox{`\textbf{if}'} \, \mbox{`$\star$'}\,\mbox{`\textbf{then}'} \, \langle \mathit{stmt}\rangle \, \mbox{`\textbf{else}'} \, \langle \mathit{stmt}\rangle \,\mbox{`\textbf{fi}'}\\
	&\mid  \mbox{`\textbf{while}'}\, \langle\mathit{bexpr}\rangle \, \text{`\textbf{do}'} \, \langle \mathit{stmt}\rangle \, \text{`\textbf{od}'}
	\\
	& \mid \mbox{`\textbf{tick}'`$($'}\langle\mathit{pexpr}\rangle\mbox{`$)$'}
	\mid \langle\mathit{stmt}\rangle \, \text{`;'} \, \langle \mathit{stmt}\rangle
	\\
	\langle\mathit{literal} \rangle &::= \langle\mathit{pexpr} \rangle\, \mbox{`$\leq$'} \,\langle\mathit{pexpr} \rangle \mid \langle\mathit{pexpr} \rangle\, \mbox{`$\geq$'} \,\langle\mathit{pexpr} \rangle
	\\
	\langle \mathit{bexpr}\rangle &::=  \langle \mathit{literal} \rangle \mid \text{`}\neg\text{'} \langle\mathit{bexpr}\rangle\\
	&\mid \langle \mathit{bexpr} \rangle \, \mbox{`\textbf{or}'} \, \langle\mathit{bexpr}\rangle
	\mid \langle \mathit{bexpr} \rangle \, \mbox{`\textbf{and}'} \, \langle\mathit{bexpr}\rangle \\
	\langle\mathit{pexpr} \rangle &::= \langle \mathit{constant} \rangle
	\mid \langle\mathit{pvar}\rangle
	\mid \langle \mathit{pexpr} \rangle \,\text{`$*$'} \,
	\langle\mathit{pexpr}\rangle
	\\
	&\mid \langle\mathit{pexpr} \rangle\, \text{`$+$'}
	\,\langle\mathit{pexpr} \rangle \mid \langle\mathit{pexpr} \rangle\,
	\text{`$-$'} \,\langle\mathit{pexpr} \rangle   \\
	\langle\mathit{expr} \rangle &::= \langle \mathit{constant} \rangle
	\mid \langle\mathit{pvar}\rangle \mid \langle\mathit{rvar}\rangle
	\mid \langle \mathit{expr} \rangle \,\text{`$*$'} \,
	\langle\mathit{expr}\rangle
	\\
	&\mid \langle\mathit{expr} \rangle\, \text{`$+$'}
	\,\langle\mathit{expr} \rangle \mid \langle\mathit{expr} \rangle\,
	\text{`$-$'} \,\langle\mathit{pexpr} \rangle
	\end{align*}
}
	\caption{Syntax of nondeterministic probabilistic programs.}
	\label{fig:syntax}
\end{figure}

\paragraph{Labels.} We refer to the status of the program counter as a \emph{label}, and assign labels $\lin{}$ and $\lout{}$ to the start and  end of the program, respectively. Our label types are as follows:
\begin{compactitem}
	\item An \emph{assignment} label corresponds to an assignment statement indicated by $:=$ or $\textbf{skip}$. After its execution, the value of the expression on its right hand side is stored in the variable on its left hand side and control flows to the next statement. A $\textbf{skip}$ assignment does not change the value of any variable.
	\item A \emph{branching} label corresponds to a conditional statement, i.e.~either an ``\textbf{if} $\phi$ \dots'' or a ``\textbf{while} $\phi$ \dots'', where $\phi$ is a condition on program variables, and the next statement to be executed depends on $\phi$.
	\item A \emph{probabilistic} label corresponds to an ``\textbf{if}~\textbf{prob}($p$) \dots'' with $p\in [0,1]$, and leads to the \textbf{then} branch with probability $p$ and the \textbf{else} branch with probability $1-p$.
	\item A \emph{nondeterministic} label corresponds to a nondeterministic branching statement indicated by ``\textbf{if} $\star$ \dots'', and is nondeterministically followed by either the \textbf{then} branch or the \textbf{else} branch.
	\item A \emph{tick} label corresponds to a statement \textbf{tick}($q$) that triggers a cost of $q$, and leads to the next label. Note that $q$ is an arithmetic expression, serving as the step-wise \emph{cost function}, and can depend on the values of program variables.
\end{compactitem}

\paragraph{Valuations.} Given a set $V$ of variables, a valuation over $V$ is a function $\pv: V \rightarrow \setR$ that assigns a value to each variable. We denote the set of all valuations on $V$ by $\val{V}$.

\paragraph{Control flow graphs (CFGs)~\cite{allen1970control}.} We define control flow graphs of our programs in the usual way, i.e.~a CFG contains one vertex for each label and an edge connects a label $\loc_i$ to another label $\loc_j$, if $\loc_j$ can possibly be executed right after $\loc_i$ by the rules above. Formally, a CFG is a tuple
\begin{equation}\label{eq:cfg}
\left(\pvars{},\rvars,\locs{},\transitions{}\right)
\end{equation}
where:
\begin{compactitem}
	\item $\pvars{}$ and $\rvars$ are finite sets of \emph{program variables} and \emph{sampling (randomized) variables}, respectively;
	\item $\locs{}$ is a finite set of \emph{labels} partitioned into (i) the set $\alocs{}$ of \emph{assignment} labels, (ii) the set $\blocs{}$ of \emph{branching} labels, (iii) the set $\plocs{}$ of \emph{probabilistic} labels, (iv) the set $\Dlocs{}$ of \emph{nondeterministic} labels, (v) the set $\tlocs{}$ of \emph{tick} labels, and (vi) a special terminal label $\lout{}$ corresponding to the end of the program. Note that the start label $\lin$ corresponds to the first statement of the program and is therefore covered in cases (i)--(v).
	
	\item $\transitions{}$ is a transition relation whose every member is a triple of the form $(\loc,\alpha,\loc')$ where $\loc$ is the source and $\loc'$ is the target of the transition, and $\alpha$ is the rule that must be obeyed when the execution goes from $\loc$ to $\loc'$. The rule $\alpha$ is either an \emph{update function} $F_\ell:\val{\pvars{}}\times\val{\rvars}\rightarrow \val{\pvars{}}$
	if $\loc\in\alocs{}$, which maps values of program and sampling variables before the assignment to the values of program variables after the assignment, or a condition $\phi$ over $\pvars{}$ if $\loc\in\blocs{}$,
	or a real number $p\in [0,1]$ if $\loc\in\plocs{}$,
	or $\star$ if $\loc\in\Dlocs{}$, or a \emph{cost function} $R_{\loc}:\val{\pvars{}}\rightarrow \Rset$ if $\loc\in\tlocs{}$. In the last case, the cost function $R_\loc$ is specified by the arithmetic expression $q$ in \textbf{tick}($q$) and maps the values of program variables to the cost of the tick operation.
\end{compactitem}

\begin{example} \label{ex:ex1}
	Figure~\ref{fig:example} provides an example program and its CFG. We assume that the probability distributions for the random variables $r$ and $r'$ are $(1, -1):(1/4, 3/4)$ and $(1, -1):(2/3, 1/3)$ respectively. In this program, the value of the variable $x$ is incremented by the sampling variable $r$, whose value is $1$ with probability $1/4$ and $-1$ with probability $3/4$. Then, the variable $y$ is assigned a random value sampled from the variable $r'$, that is $1$ with probability $2/3$ and $-1$ with probability $1/3$. The \textbf{tick} command then incurs a cost of $x \cdot y$, i.e.~$x * y$ is used as the cost function. 

		\begin{figure}
	\begin{minipage}{4.1cm}
		\lstset{language=prog}
		\lstset{linewidth=4.0cm}
		\begin{lstlisting}[basicstyle=\small,mathescape]
$1$: while $x\ge 1$ do
$2$:    $x:=x+r;$
$3$:    $y:=r';$
$4$:    tick$(x * y)$
   od
$5$:
		\end{lstlisting}
	\end{minipage}
	\begin{minipage}{3.1cm}
		\begin{flushright}
			\includegraphics[scale=0.45]{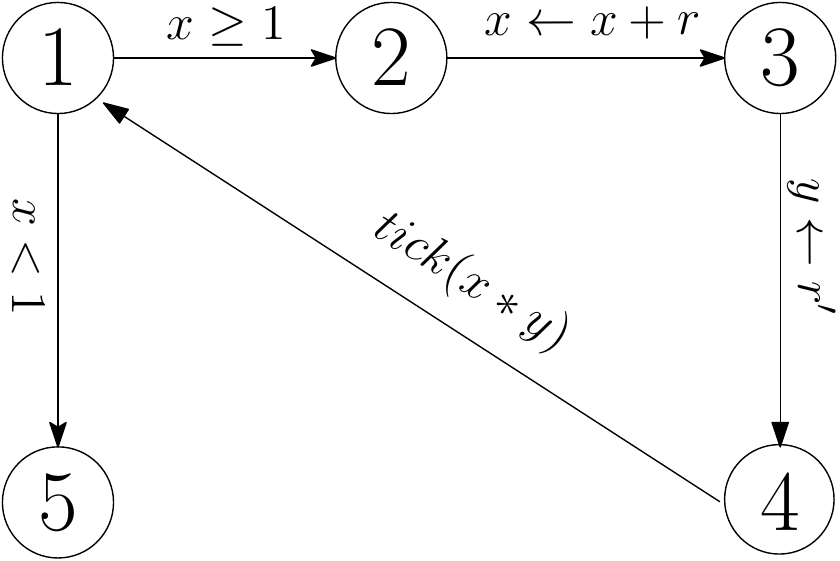}
		\end{flushright}
	\end{minipage}
	\caption{An example program with its labels (left), and its CFG (right). We have $\lin = 1$ and $\lout = 5$. }
	\label{fig:example}
\end{figure}
\end{example}
	
\paragraph{Runs and schedulers.} A \emph{run} of a program is an infinite sequence $\{ (\loc_n, \pv_n) \}_{n=0}^\infty$ of labels $\loc_n$ and valuations $\pv_n$ to program variables that respects the rules of the CFG. A \emph{scheduler} is a policy that chooses the next step, based on the history of the program, when the program reaches a nondeterministic choice. For more formal semantics see Appendix~\ref{app:semantic}.

\paragraph{Termination time~\cite{HolgerPOPL}.} The \emph{termination time} is a random variable $T$ defined on program runs as $T( \{ (\loc_n, \pv_n) \}_{n=0}^\infty ) := \min \{n ~\vert~ \loc_n = \lout \}$. We define $\min \emptyset := \infty$. Note that $T$ is a stopping time on program runs. Intuitively, the termination time of a run is the number of steps it takes for the run to reach the termination label $\lout$ or $\infty$ if it never terminates.

\paragraph{Types of termination~\cite{kaminski2018hardness,HolgerPOPL,ChatterjeeFG16}.} A program is said to \emph{almost surely terminate} if it terminates with probability 1 using any scheduler. Similarly, a program is \emph{finitely terminating} if it has finite expected termination time over all schedulers. Finally, a program has the \emph{concentration property} or \emph{concentratedly terminates} if there exist positive constants $a$ and $b$ such that for sufficiently large $n$, we have $\probm(T > n) \leq a \cdot e^{-b \cdot n}$ for all schedulers, i.e.~if the probability that the program takes $n$ steps or more decreases exponentially as $n$ grows.

Termination analysis of probabilistic programs is a widely-studied topic. For automated approaches, see~\cite{ChatterjeeFG16,AgrawalC018,SriramCAV,mciver2017new}.

\subsection{Expected Accumulated Cost}

The main notion we use in cost analysis of nondeterministic probabilistic programs is the expected accumulated cost until program termination. This concept naturally models the total cost of execution of a program in the average case. We now formalize this notion.

\paragraph{Cost of a run.} We define the random variable $C_m$ as the cost at the $m$-th step in a run, which is equal to a cost function $R_\ell$ if the $m$-th executed statement is a tick statement and is zero otherwise, i.e.~given a run $\rho = \run$, we define:
$$
C_m(\rho) := \left\{\begin{matrix*}[l]
R_{\loc_m} (\pv_m) & \text{if } \loc_m \in \tlocs\\

0 & \text{otherwise}

\end{matrix*}\right.
$$
Moreover, we define the random variable $C_\infty$ as the total cost of all steps, i.e. $C_\infty(\rho) := \sum_{m=0}^\infty C_m (\rho)$. Note that when the program terminates, the run remains in the state $\lout$ and does not trigger any costs. Hence, $C_\infty$ represents the total accumulated cost until termination. Given a scheduler $\sigma$ and an initial valuation $\pv$ to program variables, we define $\expv_{\pv}^\sigma(C_\infty)$ as the expected value of the random variable $C_\infty$ over all runs that start with $(\lin, \pv)$ and use $\sigma$ for making choices at nondeterministic points.


\begin{definition}[Expected Accumulated Cost]
	Given an initial valuation $\pv$ to program variables, the \emph{maximum expected accumulated cost}, $\supval(\pv)$, is defined as $\sup_\sigma \expv_{\pv}^\sigma(C_\infty)$,
	where $\sigma$ ranges over all possible schedulers.
\end{definition}

Intuitively, $\supval(\pv)$ is the maximum expected total cost of the program until termination, i.e.~assuming a scheduler that resolves nondeterminism to maximize the total accumulated cost.
In this work, we focus on automated approaches to find polynomial bounds for $\supval(\pv)$.

\section{Motivating Examples} \label{sec:motivation}

In this section, we present several motivating examples for the expected cost
analysis of nondeterministic probabilistic programs.
Previous general approaches for probabilistic programs, such as~\cite{pldi18},
require the following restrictions: (a)~stepwise costs are nonnegative; and (b)~stepwise costs are bounded.
We present natural examples which do not satisfy the above restrictions.
Our examples are as follows:
\begin{compactenum}
\item In Section~\ref{sec:mining}, we present an example of Bitcoin mining, where the costs are both positive and negative, but bounded.
Then in Section~\ref{sec:pool}, we present an example of Bitcoin pool mining, where the costs
are both positive and negative, as well as unbounded, but the updates to the variables at each program execution step are bounded.
\item
In Section~\ref{sec:queue}, we present an example of queuing networks which also has unbounded costs but bounded updates to the variables.

\item In Section~\ref{sec:problin}, we present an example of stochastic linear recurrences,
where the costs are nonnegative but unbounded, and the updates to the variable values
are also unbounded.
\end{compactenum}

\subsection{Bitcoin Mining} \label{sec:mining}

Popular decentralized cryptocurrencies, such as Bitcoin and Ethereum, rely on proof-of-work Blockchain protocols to ensure a consensus about ownership of funds and validity of transactions~\cite{nakamoto2008bitcoin,vogelstellerethereum}. In these protocols, a subset of the nodes of the cryptocurrency network, called \emph{miners}, repeatedly try to solve a computational puzzle. In Bitcoin, the puzzle is to invert a hash function, i.e.~to find a nonce value $v$, such that the SHA256 hash of the state of the Blockchain and the nonce $v$ becomes less than a predefined value~\cite{nakamoto2008bitcoin}. The first miner to find such a nonce is rewarded by a fixed number of bitcoins. If several miners find correct nonces at almost the same time, which happens with very low probability, only one of them will be rewarded and the solutions found by other miners will get discarded~\cite{baliga2017understanding}.

Given the one-way property of hash functions, one strategy for a miner is to constantly try randomly-generated nonces until one of them leads to the desired hash value. Therefore, a miner's chance of getting the next reward is proportional to her computational power. Bitcoin mining uses considerable electricity and is therefore very costly~\cite{de2018bitcoin}.

Bitcoin mining can be modeled by the nondeterministic probabilistic program given in Figure~\ref{fig:mining}. In this program, a miner starts with an initial balance of $x$ and mines as long as he has some money left for the electricity costs. At each step, he generates and checks a series of random nonces. This leads to a cost of $\alpha$ for electricity. With probability $p$, one of the generated nonces solves the puzzle. When this happens, with probability $p'$ the current miner is the only one who has solved the puzzle and receives a reward of $\beta$ units. However, with probability $1 - p'$, other miners have also solved the same puzzle in roughly the same time. In this case, whether the miner receives his reward or not is decided by nondeterminism. \rd{Since we are modeling the total cost from the point-of-view of the miner, getting a reward has negative cost while paying for electricity has positive cost.} The values of parameters $\alpha, \beta, p,$ and $p'$ can be found experimentally in the real world. Basically, $\alpha$ is the cost of electricity for the miner, which depends on location, $\beta$ is the reward for solving the puzzle, which depends on the Bitcoin exchange rate, and $p$ and $p'$ depend on the total computational power of the Bitcoin network, which can be estimated at any time~\cite{hashrate}. In the sequel, we assume $\alpha=1,\beta=5000,p=0.0005,p'=0.99$.

\begin{figure}
	\lstset{language=prog}
	\lstset{linewidth=5.1cm}
	\begin{lstlisting}[basicstyle=\small,mathescape]
while $x \ge \alpha$ do
	$x := x - \alpha$; tick$(\alpha)$;
	if prob($p$) then
	   if prob($p'$) then tick$(-\beta)$
	   else if $\star$ then tick$(-\beta)$
	 fi fi fi od
	\end{lstlisting}
	\caption{Bitcoin mining}
	\label{fig:mining}
\end{figure}
\begin{remark}
Note that in the example of Figure~\ref{fig:mining}, the costs are both positive ($\mathbf{tick}(\alpha)$) and
negative ($\mathbf{tick}(-\beta)$), but bounded by the constants $|\alpha|$ and $|\beta|$. Also all updates to the program variable $x$ are bounded by $|\alpha|$.
\end{remark}

\subsection{Bitcoin Pool Mining} \label{sec:pool}
As mentioned earlier, a miner's chance of solving the puzzle in Bitcoin is proportional to her computational power. Given that the overall computational power of the Bitcoin network is enormous, there is a great deal of variance in miners' revenues, e.g.~a miner might not find a solution for several months or even years, and then suddenly find one and earn a huge reward. To decrease the variance in their revenues, miners often collaborate in \emph{mining pools}~\cite{rosenfeld2011analysis}.

A mining pool is created by a manager who guarantees a steady income for all participating miners. This income is proportional to the miner's computational power. Any miner can join the pool and assign its computational power to solving puzzles for the pool, instead of for himself, i.e.~when a puzzle is solved by a miner participating in a pool, the rewards are paid to the pool manager~\cite{ChatterjeeErgodic}. Pools charge participation fees, so in the long term, the expected income of a participating miner is less than what he is expected to earn by mining on his own.

A pool can be modeled by the probabilistic program in Figure~\ref{fig:pool}. The manager starts the pool with $y$ identical miners\footnote{This assumption does not affect the generality of our modeling. If the miners have different computational powers, a more powerful miner can be modeled as a union of several less powerful miners.}. At each time step, the manager has to pay each miner a fixed amount $\alpha$. Miners perform the mining as in Figure~\ref{fig:mining}. Note that their mining revenue now belongs to the pool manager. Finally, at each time step, a small stochastic change happens in the number of miners, i.e.~a miner might choose to leave the pool or a new miner might join the pool. The probability of such changes can also be estimated experimentally.
In our example, we have that the number of miners increases by one with probability $0.4$, decrease by one with probability $0.5$, and does not change with probability $0.1$ ($y := y + (-1, 0, 1):(0.5, 0.1, 0.4)$).

\begin{figure}
	\lstset{language=prog}
	\lstset{linewidth=5.1cm}
	\begin{lstlisting}[basicstyle=\small,mathescape]
while $y \ge 1$ do
  tick$(\alpha * y)$; $i := 1$;
  while $i \leq y$ do
    if prob($p$) then
      if prob($p'$) then tick$(-\beta)$
      else if $\star$ then tick$(-\beta)$
      fi fi fi; $i := i+1$ od;
      $y := y + (-1, 0, 1):(0.5, 0.1, 0.4)$ od
	\end{lstlisting}
	\caption{Bitcoin pool mining}
	\label{fig:pool}
\end{figure}

\begin{remark}
In contrast to Figure~\ref{fig:mining} where the costs are bounded,
in Figure~\ref{fig:pool}, they are not bounded ($\mathbf{tick}(\alpha * y)$). Moreover, they are both positive ($\mathbf{tick}(\alpha * y)$) and negative ($\mathbf{tick}(-\beta)$).
However, changes to the program variables $i$ and $y$ are bounded.
\end{remark}

\subsection{Queuing Networks} \label{sect:Queing-network} \label{sec:queue}

A well-studied structure for modeling parallel systems is the \emph{Fork and Join} (FJ) queuing network~\cite{alomari2014efficient}. An FJ network consists of $K$ processors, each with its own dedicated queue (Figure~\ref{fig:FJ}). When a job arrives, the network probabilistically divides (\emph{forks}) it into one or more parts and assigns each part to one of the processors by adding it to the respective queue. Each processor processes the jobs in its queue on a first-in-first-out basis. When all of the parts of a job are processed, the results are \emph{joined} and the job is completed. The \emph{processing time} of a job is the amount of time it takes from its arrival until its completion.

FJ networks have been used to model and analyze the efficiency of a wide variety of parallel systems~\cite{alomari2014efficient}, such as web service applications~\cite{menasce2004response}, complex network intrusion detection systems~\cite{alomari2012autonomic}, MapReduce frameworks~\cite{dean2008mapreduce}, programs running on multi-core processors~\cite{hill2008amdahl}, and health care applications such as diagnosing patients based on test results from several laboratories~\cite{almomen2012design}.

An FJ network can be modeled as a probabilistic program. For example, the program in Figure~\ref{fig:FJcode} models a network with $K=2$ processors that accepts jobs for $n$ time units. At each unit of time, one unit of work is processed from each queue, and there is a fixed probability $0.02$ that a new job arrives. The network then probabilistically decides to assign the job to the first processor (with probability $0.2$) or the second processor (with probability $0.4$) or to divide it among them (with probability $0.4$). We assume that all jobs are identical and for processor~$1$ it takes $3$ time units to process a job, while processor~$2$ only takes $2$ time units. If the job is divided among them, processor~$1$ takes $2$ units to finish its part and processor~$2$ takes $1$ time unit. The variables $l_1$ and $l_2$ model the length of the queues for each processor, and the program cost models the total processing time of the jobs.

Note that the processing time is computed from the point-of-view of the jobs and does not model the actual time spent on each job by the processors, instead it is defined as the amount of time from the moment the job enters the network, until the moment it is completed. Hence, the processing time can be computed as soon as the job is assigned to the processors and is equal to the length of the longest queue.

\begin{figure}
	\includegraphics[scale=0.65]{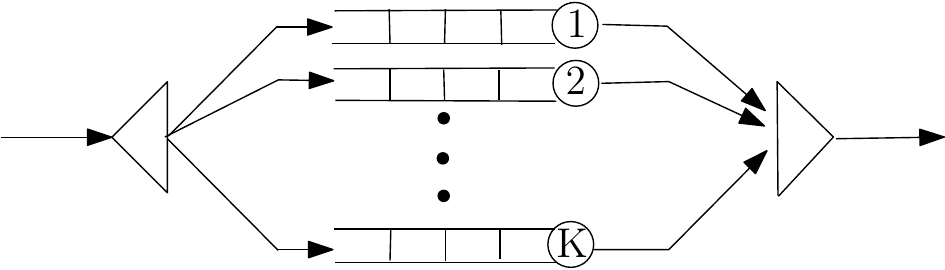}
	\caption{A Fork and Join network with $K$ processors}
	\label{fig:FJ}
\end{figure}

\begin{figure}
	\lstset{language=prog}
	\lstset{linewidth=5.1cm}
	\begin{lstlisting}[basicstyle=\small,mathescape]
$l_1 := 0; \quad l_2 := 0; \quad i:= 1;$
while $i \leq n$ do
  if $l_1 \ge 1$ then $l_1 := l_1 - 1$ fi;
  if $l_2\ge 1$ then $l_2 := l_2 - 1$ fi;
  if prob ($0.02$) then
    if prob($0.2$) then
      $l_1 := l_1 + 3$
    else if prob($0.5$) then
      $l_2 := l_2 + 2$
    else
      $l_1 := l_1 + 2;$ $l_2 := l_2 + 1$
    fi fi;
    if $l_1 \geq l_2$ then tick$(l_1)$ else tick$(l_2)$ fi
  fi; $i := i + 1$ od
	\end{lstlisting}
	\caption{A FJ-network Example with $K=2$ Processors}
	\label{fig:FJcode}
\end{figure}

\begin{remark}
	In the example of Figure~\ref{fig:FJcode}, note that the costs, i.e.~$\mathbf{tick}(l_1)$ and $\mathbf{tick}(l_2)$, depend on the length of the queues and are therefore unbounded. However, all updates in program variables are bounded, i.e.~a queue size is increased by at most $3$ at each step of the program. The maximal update appears in the assignment $l_1 := l_1 + 3$.
\end{remark}

\subsection{Stochastic Linear Recurrences} \label{sec:problin}

Linear recurrences are systems that consist of a finite set $\textbf{x}$ of variables, together with a finite set $ \textbf{a}_1, \textbf{a}_2, \ldots, \textbf{a}_m$ of linear update rules. At each step of the system's execution, one of the rules is chosen and applied to the variables. Formally, if there are $n$ variables, then we consider $\textbf{x}$ and each of the $\textbf{a}_i$'s to be a vector of length $n$, and applying the rule $\textbf{a}_i$ corresponds to the assignment $\textbf{x} := \textbf{a}_i \cdot \textbf{x}$. This process continues as long as a condition $\phi$ is satisfied.
Linear recurrences are well-studied and appear in many contexts, e.g.~to model linear dynamical systems, in theoretical biology, and in statistical physics (see~\cite{joel1,joel2,joelSURVEY}). A classical example is
the so-called species fight in ecology.

A natural extension of linear recurrences is to consider stochastic linear recurrences, where at each step the rule to be applied is chosen probabilistically. Moreover, the cost of the process at each step is a linear combination $\textbf{c} \cdot \textbf{x}$ of the variables. Hence, a general stochastic linear recurrence is a program in the form shown in Figure~\ref{fig:linrec}.

We present a concrete instantiation of such a program in the context of species fight.
Consider a fight between two types of species, $a$ and $b$, where there are a finite number of
each type in the initial population.
The types compete and might also prey upon each other.
The fitness of the types depends on the environment, which evolves stochastically. For example, the environment may represent the temperature,
and a type might have an advantage over the other type in warm/cold environment.
The cost we model is the amount of resources consumed by the population. Hence, it is a linear combination of the population of each type
(i.e.~each individual consumes some resources at each time step).

Figure~\ref{fig:species} provides an explicit example,
in which with probability $1/2$, the environment becomes hospitable to $a$, which leads to an increase in its population, and assuming that $a$ preys on $b$, this leads to a decrease in the population of $b$.
On the other hand, the environment might become hostile to $a$, which leads to an increase in $b$'s population.
Moreover, each individual of either type $a$ or $b$ consumes 1 unit of resource per time unit. We also assume that a population of less than $5$ is unsustainable and leads to extinction.

\begin{figure}
	\lstset{language=prog}
	\lstset{linewidth=5.1cm}
	\begin{lstlisting}[basicstyle=\small,mathescape]
while $\phi$ do
    if prob($p_1$) then
    	$\textbf{x}$ := $\textbf{a}_1 \cdot \textbf{x}$
    else if prob($p_2$) then
    	$\textbf{x}$ := $\textbf{a}_2 \cdot \textbf{x}$
	    $\vdots$
    else if prob($p_m$) then
    	$\textbf{x}$ := $\textbf{a}_m \cdot \textbf{x}$
    fi $\ldots$ fi;
    tick$(\textbf{c} \cdot \textbf{x})$
od
	\end{lstlisting}
	\caption{A general stochastic linear recurrence}
	\label{fig:linrec}
\end{figure}

\begin{figure}
	\lstset{language=prog}
	\lstset{linewidth=5.1cm}
	\begin{lstlisting}[basicstyle=\small,mathescape]
while $a \ge 5$ and $b \ge 5$ do
    tick$(a+b)$;
    if prob($0.5$) then $b := 0.9*b$; $a := 1.1*a$
    else $b := 1.1*b$; $a := 0.9*a$ fi
od
	\end{lstlisting}
	\caption{A species fight example}
	\label{fig:species}
\end{figure}

\begin{remark}
	Note that in Figure~\ref{fig:species}, there are unbounded costs ($\mathbf{tick}(a+b)$) and unbounded updates to the variables (e.g.~$a := 1.1*a$). However, the costs are always nonnegative.
\end{remark}

\section{Main Ideas and Novelty} \label{sec:novelty}

In this work, our main contribution is an automated approach for obtaining polynomial bounds on the expected accumulated cost of nondeterministic probabilistic programs.
In this section, we present an outline of our main ideas, and a discussion on their novelty in comparison with previous approaches.
The key contributions are organized as follows:
(a)~mathematical foundations; (b)~soundness of the approach; and (c)~computational results.

\subsection{Mathematical Foundations}

The previous approach of~\cite{pldi18} can only handle nonnegative bounded costs.
Their main technique is to consider {\em potential functions} and probabilistic
extensions of weakest precondition, which relies on monotonicity. This is the key reason why the costs must be nonnegative.
Instead, our approach is based on martingales, and can hence handle both positive and
negative costs.

\paragraph{Extension of OST.}
A standard result in the analysis of martingales is the Optional Stopping Theorem (OST), which provides a set of conditions on a (super)martingale $\{ X_n\}_{n=0}^\infty$ that are sufficient to ensure bounds on its expected value at a stopping time.
A requirement of the OST is the so-called bounded difference condition, i.e.~that there should exist a constant number $c$, such that the stepwise difference $\vert X_{n+1} - X_n \vert$ is always less than $c$. In program cost analysis, this condition translates to the requirement that the stepwise cost function at each program point must be bounded by a constant.
It is well-known that the bounded difference condition in OST is an essential prerequisite,
and thus application of classical OST can only handle bounded costs.

We present an extension of the OST that provides certain new conditions for handling differences $\vert X_{n+1} - X_n \vert$ that are not bounded by a constant, but instead by a polynomial on the step number $n$. Hence, our extended OST can be applied to programs such as the motivating examples in Sections~\ref{sec:mining},~\ref{sec:pool} and~\ref{sec:queue}.
The details of the OST extension are presented in Section~\ref{sec:ost}.

\subsection{Soundness of the Approach}

For a sound approach to compute polynomial bounds on expected accumulated cost, we present the following results
(details in Section~\ref{sec:approach}):
\begin{compactenum}
\item We define the notions of \emph{polynomial upper cost supermartingale (PUCS)} and \emph{polynomial lower cost submartingale (PLCS)} for upper and lower bounds of the expected accumulated cost over probabilistic programs, respectively (see Section~\ref{sec:defs}).

\item For the case where the costs can be both positive and negative (bounded or unbounded), but the variable updates  are bounded,
we use our extended OST to establish that PUCS's and PLCS's provide a sound approach to obtain upper and lower
bounds on the expected accumulated cost (see Section~\ref{sec:soundness}).

\item For costs that are nonnegative (even with unbounded updates), we show that PUCS's
provide a sound approach to obtain upper bounds on the expected accumulated cost (see Section~\ref{sec:soundness2}).
The key mathematical result we use here is the Monotone Convergence theorem. We do not need OST in this case.

\end{compactenum}

\subsection{Computational Results}

By our definition of PUCS/PLCS, a candidate polynomial $h$ is a PUCS/PLCS for a given program, if it satisfies a number of polynomial inequalities, which can be obtained from the CFG of the program. Hence, we reduce the problem of synthesis of a PUCS/PLCS to solving a system of polynomial inequalities. Such systems can be solved using quantifier elimination, which is computationally expensive. Instead, we present the alternative sound method of using a Positivstellensatz, i.e.~a theorem in real semi-algebraic geometry that characterizes positive polynomials over a semi-algebraic set. In particular, we use Handelman's Theorem to show that given a nondeterministic probabilistic program, a PUCS/PLCS can be synthesized by solving a linear programming instance of polynomial size (wrt the size of the input program and invariant). Hence, our sound approach for obtaining polynomial bounds on the expected accumulated cost of a program runs in polynomial time.
The details are presented in Section~\ref{sec:computational}.

\subsection{Novelty}

\rd{
The main novelties of our approach are as follows:
\begin{compactenum}
\item \emph{Positive and Negative Costs.} In contrast to previous approaches, such as~\cite{pldi18}, that can only handle positive costs, our approach can handle both positive and negative costs. In particular, approaches that are based on weakest pre-expectation require the one-step pre-expectation of the cost to be non-negative (due to monotonicity conditions). This requirement is enforced by disallowing negative costs. In contrast, our approach can even handle cases where the one-step pre-expectation is negative, e.g.~see lines~4--5 in Figure~\ref{fig:mining} (Section~\ref{sec:mining}) and lines~5--6 in Figure~\ref{fig:pool} (Section~\ref{sec:pool}) where the cost is always negative. As shown by these examples, many real-world scenarios contain both costs and rewards (negative costs). We provide the first approach that can handle such scenarios.

\item \emph{Variable-dependent Costs.} Previous approaches, such as~\cite{pldi18}, require the one-step costs to be bounded constants. A major novelty of our approach is that it can handle unbounded \emph{variable-dependent} costs. This allows us to consider real-world examples, such as the ubiquitous Queuing Networks of Section~\ref{sec:queue}, in which the cost depends on the length of a queue.

\item \emph{Upper and Lower Bounds.} While previous approaches, such as~\cite{pldi18}, could only present sound upper bounds with positive bounded
costs, our approach for positive and negative costs, with the restriction of bounded updates
to the variables, can provide both upper and lower bounds on the expected accumulated costs.
Thus, for the examples of Sections~\ref{sec:mining}, \ref{sec:pool} and~\ref{sec:queue}, we obtain
both upper and lower bounds. This is the first approach that is able to provide lower bounds for expected accumulated cost.

\item \emph{Efficiency.} We present a \emph{provably polynomial-time} computational approach for obtaining bounds on the expected accumulated costs. The previous approach in~\cite{pldi18} has exponential dependence on size of the program.

\item \emph{Compositionality.} Our approach directly leads to a compositional approach. For a program $P$, we can annotate $P$ with $\{G\}P\{H\}$ meaning that $G$ and $H$ are functions of a PUCS/PLCS at the start and end program counter. Then, by applying the conditions of PUCS/PLCS to the program syntax, one can directly establish a proof system for proving the triple $\{G\}P\{H\}$.

\end{compactenum}
}

\subsection{Limitations}

We now discuss some limitations of our approach.
\begin{compactenum}
\item As in previous approaches, such as~\cite{pldi18,ijcai18}, we need to assume that the input program terminates.

\item For programs with both positive and negative costs, we handle either bounded updates to variables or
bounded costs. The most general case, with both unbounded costs and unbounded updates,
remains open.

\item For unbounded updates to variables, we consider nonnegative costs, and present only upper bounds, and
not lower bounds. However, note that our approach is the first one to present any lower bounds for
cost analysis of probabilistic programs (with bounded updates to variables),
and no previous approach can obtain lower bounds in any case.

\item \rd{
	Our approach assumes that linear invariants at every point of the program are given as part of the input (see Section~\ref{sec:defs}). Note that linear invariant generation is a classical problem with several efficient tools (e.g.~\cite{sting1}).
}

\end{compactenum}

\section{The Extension of the OST}\label{sect:ost} \label{sec:ost}

The Optional Stopping Theorem (OST) states that, given a martingale (resp. supermartingale), if its step-wise difference $X_n - X_{n+1}$ is bounded, then its expected value at a stopping time
is equal to (resp. no greater than) its initial value.

\begin{theorem}[Optional Stopping Theorem (OST)~\cite{williams1991probability,doob1971martingale}]
Consider any stopping time $U$ wrt a filtration $\{\mathcal{F}_n\}_{n=0}^\infty$ and any martingale (resp. supermartingale) $\{X_n\}_{n=0}^\infty$ adapted to $\{\mathcal{F}_n\}_{n=0}^\infty$ and let $Y=X_U$.
Then the following condition is sufficient to ensure that $\expv\left(|Y|\right)<\infty$ and $\expv\left(Y\right) = \expv(X_0)$ (resp. $\expv\left(Y\right)\le\expv(X_0)$):
\begin{itemize}
	\item There exists an $M \in [0, \infty)$ such that for all $n \geq 0$, $|X_{n+1}-X_n|\le M$ almost surely.
\end{itemize}
\end{theorem}

It is well-known that the stepwise bounded difference condition (i.e. $|X_{n+1}-X_n|\le M$) is an essential prerequisite~\cite{williams1991probability}.
Below we present our extension of OST to unbounded differences.

\begin{theorem}[The Extended OST]\label{eost} \label{thm:eost}
Consider any stopping time $U$ wrt a filtration $\{\mathcal{F}_n\}_{n=0}^\infty$ and any martingale (resp. supermartingale) $\{X_n\}_{n=0}^\infty$ adapted to $\{\mathcal{F}_n\}_{n=0}^\infty$ and let $Y = X_U$.
Then the following condition is sufficient to ensure that $\expv\left(|Y|\right)<\infty$ and $\expv\left(Y\right)=\expv(X_0)$ (resp. $\expv\left(Y\right)\le\expv(X_0)$):
\begin{compactitem}
\item There exist real numbers $M, c_1, c_2, d > 0$ such that (i) for sufficiently large $n \in \setN$, it holds that $\probm(U>n) \leq c_1 \cdot e^{-c_2 \cdot n}$ and (ii) for all $n \in \setN$, $\vert X_{n+1} - X_n \vert \leq M \cdot n^d$ almost surely.
\end{compactitem}
\end{theorem}

\paragraph{Intuition and proof idea.} 
We extend the OST so that the stepwise difference $|X_{n+1}-X_n|$ need not be bounded by a constant, but instead by a polynomial in terms of the step counter $n$.
However, we require that the stopping time $U$ satisfies the concentration condition that specifies an exponential decrease in $\probm(U>n)$.
We present a rigorous proof that uses Monotone and Dominated Convergence Theorems along with the
concentration bounds and polynomial differences to establish the above result. 
For technical details see Appendix~\ref{app:OST}.

\rd{
\begin{remark} We note several points about our extended OST:
\begin{compactenum}
\item It can handle unbounded difference. 
\item While the original proof for OST (\cite{williams1991probability}) is restricted to bounded-difference, our proof for the extended OST (Appendix~\ref{app:OST}) uses a novel condition, i.e.~that the stopping time has exponentially-decreasing probabilities, to handle unbounded difference.
\item We show that the new condition in our extended OST formulation corresponds to program termination. Thus, our extension provides a sound method for cost analysis of terminating programs.
\end{compactenum}
\end{remark}
}

\section{Polynomial Cost Martingales}\label{sec:approach}

In this section, we introduce the notion of polynomial cost martingales, which serve as the main tool for reducing the cost analysis problem over nondeterministic probabilistic programs to the analysis of a stochastic process.

\subsection{Definitions}\label{sec:defs}

Below, we fix a probabilistic program and its CFG of form~(1).
In order to apply our extended OST for cost analysis of the program, it should first be translated into a discrete-time stochastic process.
This is achieved using the concept of \emph{cost martingales}.
To define cost martingales, we first need the notions of invariants and pre-expectation.

\begin{definition}[Invariants and linear invariants]
	
	Given a program, its set $L$ of labels, and an initial valuation $\pv^*$ to program variables $\pvars$, an \emph{invariant} is a function $I: L \rightarrow \rd{\mathcal{P}}({\val{\pvars}})$ that assigns a set $I(\loc)$ of valuations over $\pvars$ to every label $\loc$, such that for all configurations $(\loc, \pv)$ that are reachable from the initial configuration $(\lin, \pv^*)$ by a run of the program, it holds that $\pv \in I(\loc)$. \rd{The invariant $I$ is linear if every $I(\loc)$ is a finite union of polyhedra.}
\end{definition}
\paragraph{Intuition.} An invariant $I$ is an over-approximation of the reachable valuations at each label of the program. \rd{An invariant is called linear if it can be represented by a finite number of linear inequalities.}

\begin{remark}\rd{In the sequel, we compute polynomial bounds that are applicable to \emph{every} initial valuation that satisfies the linear invariants. To obtain concrete bounds, we fix a single initial valuation $\pv^*$ and choose polynomial bounds that are as tight as possible wrt $\pv^*$. Nevertheless, these polynomial bounds are valid upper/lower bounds for all other valid initial valuations, too.}
\end{remark}

\begin{example}
	Figure~\ref{fig:ex2} (top), shows the same program as in Example~\ref{ex:ex1}, together with linear invariants for each label of the program. The invariants are enclosed in square brackets.
\end{example}

\begin{definition}[Pre-expectation]\label{def:pre-exp1}
Consider any function $h:\locs{}\times \val{\pvars{}}\rightarrow\Rset$.
We define its \emph{pre-expectation} as the function $\pre_{h}:\locs{}\times\val{\pvars{}}\rightarrow\Rset$ by:
\begin{compactitem}
\item $\pre_h(\loc,\pv):=h(\loc,\pv)$ if $\loc=\lout{}$ is the terminal label;
\item $\pre_{h}(\loc,\pv):=\expv_{\rv}[h(\loc', F_\loc(\pv,\rv))]$ if $\loc\in\alocs{}$ is an assignment label with the update function $F_\loc$, and the next label is $\loc'$.  Note that in the expectation $\expv_{\rv}[h(\loc', F_\loc(\pv,\rv))]$, the values of $\loc'$ and $\pv$ are treated as constants and $\rv$ observes the probability distributions specified for the sampling variables;
\item $\pre_{h}(\loc,\pv):= \mathbf{1}_{\pv\models\phi}\cdot h(\loc_{1},\pv) + \mathbf{1}_{\pv\not\models\phi}\cdot h(\loc_{2},\pv)$ if $\loc\in\blocs{}$ is a branching label and $\loc_{1},\loc_2$ are the labels for the \textbf{true}-branch and the \textbf{false}-branch, respectively. The indicator $\mathbf{1}_{\pv\models\phi}$ is equal to $1$ when $\pv$ satisfies $\phi$ and $0$ otherwise. Conversely, $\mathbf{1}_{\pv\not\models\phi}$ is $1$ when $\pv$ does not satisfy $\phi$ and $0$ when it does;
\item $\pre_{h}(\loc,\pv):= \sum_{(\loc,p,\loc')\in\transitions{}} p \cdot h(\loc',\pv)$ if $\loc\in\plocs{}$ is a probabilistic label;
\item $\pre_{h}(\loc,\pv):= R_{\loc}(\pv) + h(\loc',\pv)$ if $\loc\in\tlocs{}$ is a tick label with the cost function $R_{\loc}$ and the successor label $\loc'$;
\item $\pre_{h}(\loc,\pv):= \max_{(\loc,\star,\loc')\in\transitions{}}h(\loc',\pv)$ if $\loc\in\Dlocs{}$ is a nondeterministic label.
\end{compactitem}
\end{definition}

\paragraph{Intuition.} The pre-expectation $\pre_{h}(\loc,\pv)$ is the cost of the current step plus the expected value of $h$
in the next step of the program execution, i.e.~the step after the configuration $(\loc,\pv)$. In this expectation, $\loc$ and $\pv$ are treated as constants.
For example, the pre-expectation at a probabilistic branching label is the averaged sum over the values of $h$ at all possible successor labels.

\begin{example}\label{ex:simple}
In Figure~\ref{fig:ex2} (top) we consider the same program as in Example~\ref{ex:ex1}. Recall that the probability distributions used for sampling variables $r$ and $r'$ are $(1, -1):(1/4, 3/4)$ and $(1,-1):(2/3,1/3)$, respectively.
The table in Figure~\ref{fig:ex2} (bottom) provides an example function $h$ and the corresponding pre-expectation $\pre_h$. The gray part shows the steps in computing the function $\pre_h$ and the black part is the final result\footnote{The reason for choosing this particular $h$ will be clarified by Example~\ref{ex:complex2}.}.

\begin{figure}
	\lstset{language=prog}
	\lstset{linewidth=6.0cm}
\begin{center}
	\begin{minipage}{7cm}
\begin{lstlisting}[basicstyle=\small,mathescape]
$1$: $[x\ge 0]$               while $x\ge 1$ do
$2$: $[x\ge 1]$                 $x:=x+r$;
$3$: $[x\ge 0]$                 $y:=r'$;
$4$: $[x\ge 0 \wedge -1 \le y \le 1]$       tick$(x * y)$ od
$5$: $[0 \le x \le 1]$
\end{lstlisting}	\end{minipage}
\end{center}	
\begin{minipage}{\columnwidth}
\begin{center}

\footnotesize{
\begin{tabular}{|m{2mm}|m{2cm}|C{5cm}|}
\hline
$n$&$h(\loc_n,x,y)$&$pre_h(\loc_n,x,y)$\\
\hline
$1$ &$\frac{1}{3} x^2+\frac{1}{3} x$& \makecell{\smaller\textcolor{gray}{$\mathbf{1}_{x \geq 1} \cdot h(\loc_2, x, y) + \mathbf{1}_{x<1} \cdot h(\loc_5, x, y)$ =}\\ $\mathbf{1}_{x \geq 1}\cdot(\frac{1}{3}x^2+\frac{1}{3}x)+\mathbf{1}_{x < 1}\cdot 0 $}    \\
\hline
$2$ &$\frac{1}{3} x^2+\frac{1}{3} x$& \makecell{\smaller\textcolor{gray}{$\frac{1}{4} h(\loc_3, x+1, y) + \frac{3}{4} h(\loc_3, x-1, y)$ =} \\ $\frac{1}{3}x^2+\frac{1}{3} x$} \\
\hline
$3$ &$\frac{1}{3} x^2+\frac{2}{3} x$ & \makecell{\smaller\textcolor{gray}{$\frac{2}{3} h(\loc_4, x, 1) + \frac{1}{3} h(\loc_4, x, -1)$ =} \\ $\frac{1}{3} x^2+\frac{2}{3} x$}  \\
\hline
$4$ &$\frac{1}{3} x^2+xy+\frac{1}{3} x$ & \makecell{\smaller\textcolor{gray}{$x \cdot y + h(\loc_1, x, y)$ =} \\ $\frac{1}{3} x^2+xy+\frac{1}{3} x$ }         \\
\hline
$5$ &$0$& \makecell{\smaller\textcolor{gray}{$h(\loc_5, x, y)$ =} $0$}\\
\hline
\end{tabular}
}
\end{center}
\end{minipage}

\caption{A program together with an example function $h$ and the corresponding pre-expectation function $\pre_h$.}
\label{fig:ex2}

\end{figure}

\end{example}

We now define the central notion of \emph{cost martingales}.
For algorithmic purposes, we only consider \emph{polynomial} cost martingales in this work. We start with the notion of PUCS which is meant to serve as an upperbound for the expected accumulated cost of a program.

\begin{definition}[Polynomial Upper Cost Supermartingales]\label{def:pupf}
A \emph{polynomial upper cost supermartingale (PUCS)} of degree $d$ wrt a given linear invariant $I$ is a function $h:\locs{}\times\val{\pvars{}} \rightarrow\Rset$
that satisfies the following conditions:
\begin{compactitem}
\item[\emph{(C1)}] for each label $\loc$, $h(\loc)$ is a polynomial of degree at most $d$ over program variables;
\item[\emph{(C2)}] for all valuations $\pv\in\val{\pvars{}}$, we have
$h(\lout,\pv)=0$;

\item[\emph{(C3)}] for all non-terminal labels $\loc\in\locs\setminus\{\lout\}$ and reachable valuations $\pv \in I(\loc)$, we have $\pre_{h}(\loc,\pv)\leq h(\loc,\pv)$.

\end{compactitem}
\end{definition}

\paragraph{Intuition.} Informally, (C1) specifies that the PUCS should be polynomial at each label, (C2) says that the value of the PUCS at the terminal label $\lout$ should always be zero, and (C3) specifies that at all reachable configurations $(\loc,\pv)$, the pre-expectation is no more than the value of the PUCS itself.

Note that if $h$ is polynomial in program variables, then $\pre_h(\loc,-)$ is also polynomial
if $\loc$ is an assignment, probabilistic branching or tick label.
For example, in the case of assignment labels, $\expv_{\rv}[h(\loc', F_\loc(\pv,\rv))]$
is polynomial in $\pv$ if both $h$ and $F_\loc$ are polynomial.

\begin{example}\label{ex:complex2} \label{ex:pucs1}
By Definition~\ref{def:pupf}, the function $h$ given in Example~\ref{ex:simple} is a PUCS. For every label $\loc$ of the program, $h(\loc, -)$ is a polynomial of degree at most $2$, so $h$ satisfies condition (C1). It is straightforward to verify, using the table in Figure~\ref{fig:ex2} (bottom), that $h$ satisfies (C2) and (C3) as well.

\end{example}

We now define the counterpart of PUCS for lower bound.

\begin{definition}[Polynomial Lower Cost Submartingales]\label{de:plcs}
A \emph{polynomial lower cost submartingale (PLCS)} wrt a linear invariant $I$ is a function $h:\locs{}\times\val{\pvars{}}\rightarrow\Rset$
that satisfies (C1) and (C2) above, and the additional condition (C3') below (instead of (C3)):
\begin{compactitem}
\item[\emph{(C3')}] for all non-terminal labels $\loc\neq\lout{}$ and reachable valuations $\pv \in  I(\loc)$, we have $\pre_h(\loc,\pv)\ge h(\loc,\pv)$;
\end{compactitem}
\end{definition}
Intuitively, a PUCS requires the pre-expectation $\pre_h$ to be no more than $h$ itself, while a PLCS requires the converse, i.e.~that $\pre_h$ should be no less than $h$.

\begin{example} \label{ex:alsoplcs}
	As shown in Example~\ref{ex:pucs1}, the function $h$ given in Example~\ref{ex:simple} (Figure~\ref{fig:ex2}) satisfies (C1) and (C2). Using the table in Figure~\ref{fig:ex2}, one can verify that $h$ satisfies (C3') as well. Hence, $h$ is a PLCS.
\end{example}

In the following sections, we prove that PUCS's and PLCS's are sound methods for obtaining upper and lower bounds on the expected accumulated cost of a program.

\subsection{General Unbounded Costs and Bounded Updates}\label{sec:soundness}

In this section, we consider nondeterministic probabilistic programs with general unbounded costs, i.e.~both positive and negative costs, and bounded updates to the program variables. Using our extension of the OST (Theorem~\ref{thm:eost}), we show that PUCS's and PLCS's are sound for deriving upper and lower bounds for the expected accumulated cost.

Recall that the extended OST has two prerequisites. One is that, for sufficiently large $n$, the stopping time $U$ should have exponentially decreasing probability of nontermination, i.e.~$\probm(U > n) \leq c_1 \cdot e^{-c_2 \cdot n}$. The other is that the stepwise difference $\vert X_{n+1} - X_n \vert$ should be bounded by a polynomial on the number $n$ of steps. We first describe how these conditions affect the type of programs that can be considered, and then provide our formal soundness theorems.

The first prerequisite is equivalent to the assumption that the program
has the concentration property.
To ensure the first prerequisite, we apply the existing approach of difference-bounded ranking-supermartingale maps~\cite{ChatterjeeFNH16,ChatterjeeFG16}.
We ensure the second prerequisite by assuming the bounded update condition, i.e.~that every assignment to each program variable changes the value of the variable by a bounded amount. We first formalize the concept of bounded update and then argue why it is sufficient to ensure the second prerequisite.

\begin{definition}[Bounded Update]\label{bu}
A  program $P$ with invariant $I$ has the \emph{bounded update} property over its program variables, if there exists a constant $M>0$ such that
for every assignment label $\loc$ with update function $F_\loc$, we have
$\forall \pv\in I(\loc)\,~\forall \rv\,~\forall x\in \pvars\,~~ |F_\loc(\pv,\rv)(x)-\pv(x)|\le M$\enskip.
\end{definition}

\paragraph{The reason for assuming bounded update.} A consequence of the bounded update condition is that at the $n$-th execution step of any run of the program, the absolute value of any program variable $x$ is bounded by $M\cdot n + x_0$, where
$M$ is the constant bound in the definition above and $x_0$ is the initial value of the variable $x$. Hence, for large enough $n$, the absolute value of any variable $x$ is bounded by $(M+1) \cdot n$.
Therefore, given a PUCS $h$ of degree $d$, one can verify that
the step-wise difference of $h$ is bounded by a polynomial on the number $n$ of steps. More concretely, $h$ is a degree-$d$ polynomial over variables that are bounded by $(M+1) \cdot n$, so $h$ is bounded by $M'\cdot n^d$ for some constant $M'> 0$ .
Thus, the bounded update condition is sufficient to fulfill the second prerequisite of our extended OST.

Based on the discussion above, we have the following soundness theorems:

\begin{theorem}[Soundness of PUCS]\label{thm:lmf1}
Consider a nondeterministic probabilistic program $P$, with a linear invariant $I$ and a PUCS $h$. If $P$ satisfies the concentration property and the bounded update property, then $\supval(\pv)\le h(\lin,\pv)$ for all initial valuations $\pv\in I(\lin)$.
\end{theorem}

\begin{proof}[Proof Sketch]
We define the stochastic process $\{X_n\}_{n=0}^\infty$ as $X_n := h(\overline{\loc}_n,\overline{\pv}_n)$, where $\overline{\loc}_n$ is the random variable representing
the label at the $n$-th step of a program run, and $\overline{\pv}_n$ is a vector of random variables consisting of components $\overline{\pv}_n(x)$
which represent values of program variables $x$ at the $n$-th step.
Furthermore, we construct the stochastic process $\{Y_n\}_{n=0}^\infty$ such that $Y_n=X_n+\sum_{k=0}^{n-1} C_k$. Recall that $C_k$ is the cost of the $k$-th step of the run and $C_\infty = \sum_{k=0}^\infty C_k$.
We consider the termination time $T$ of $P$ and prove that $\{Y_n\}_{n=0}^\infty$ satisfies the prerequisites of our extended OST (Theorem~\ref{eost}). This proof depends on the assumption that $P$ has concentration and bounded update properties.
Then by applying Theorem~\ref{eost}, we have that $\expv(Y_T)\le \expv(Y_0)$.
Since $Y_T=X_T+\sum_{k=0}^{T} C_k= C_\infty$, we obtain the desired result.  For a more detailed proof, see Appendix~\ref{app:PUPFs}.
\end{proof}

\begin{example}\label{ex:complex3}
Given that the $h$ in Example~\ref{ex:simple} is a PUCS, we can conclude that for all initial values $x_0$ and $y_0$, we have $\supval(x_0,y_0)\le h(\loc_1,x_0,y_0)=\frac{1}{3}x_0^2+\frac{1}{3}x_0$.
\end{example}

We showed that PUCS's are sound upper bounds for the expected accumulated cost of a program. The following theorem provides a similar result for PLCS's and lower bounds.

\begin{theorem}[Soundness of PLCS, Proof in Appendix~\ref{app:PLPFs}]\label{thm:llmf1}
	
	Consider a nondeterministic probabilistic program $P$, with a linear invariant $I$ and a PLCS $h$. If $P$ satisfies the concentration property and the bounded update property, then $\supval(\pv)\ge h(\lin,\pv)$ for all initial valuations $\pv\in  I(\lin)$.
\end{theorem}

\begin{example}\label{ex:complex5}
Given that the $h$ in Example~\ref{ex:simple} is a PLCS, we can conclude that for all initial values $x_0$ and $y_0$, we have $\supval(x_0,y_0)\ge h(\loc_1,x_0,y_0)=\frac{1}{3}x_0^2+\frac{1}{3}x_0$.
\end{example}

\begin{remark}
Putting together the results from Examples~\ref{ex:complex3} and~\ref{ex:complex5}, we conclude that the expected accumulated cost of Example~\ref{ex:simple} is precisely $\frac{1}{3}x_0^2+\frac{1}{3}x_0$.
\end{remark}

\begin{remark}
	The motivating examples in Sections~\ref{sec:mining}, \ref{sec:pool} and \ref{sec:queue}, i.e. Bitcoin mining, Bitcoin pool mining and FJ queuing networks, have potentially unbounded costs that can be both positive and negative. Moreover, they satisfy the bounded update property. Therefore, using PUCS's and PLCS's leads to sound bounds on their expected accumulated costs.
\end{remark}

\subsection{Unbounded Nonnegative Costs and General Updates}\label{sec:soundness2}

In this section, we consider programs with unbounded \emph{nonnegative} costs, and show that a PUCS is a sound upper bound for their expected accumulated cost.
This result holds for programs with arbitrary unbounded updates to the variables.

Our main tool is the well-known Monotone Convergence Theorem (MCT)~\cite{williams1991probability}, which states that if $X$ is a random variable and $\{X_n\}_{n=0}^\infty$ is a non-decreasing discrete-time stochastic process such that $\lim_{n \rightarrow \infty} X_n = X$ almost surely, then $\lim_{n \rightarrow \infty} \expv(X_n) = \expv(X)$.

As in the previous case, the first step is to translate the program to a stochastic process. However, in contrast with the previous case, in this case we only consider \emph{nonnegative} PUCS's. This is because all costs are assumed to be nonnegative. We present the following soundness result:

\begin{theorem}[Soundness of nonnegative PUCS]\label{thm:lmf2}
	Consider a nondeterministic probabilistic program $P$, with a linear invariant $I$ and a nonnegative PUCS $h$. If all the step-wise costs in $P$ are always nonnegative, then $\supval(\pv)\le h(\lin,\pv)$ for all initial valuations $\pv\in I(\lin)$.
\end{theorem}

\begin{proof}[Proof Sketch]
We define the stochastic process $\{X_n\}_{n=0}^\infty$ as in Theorem~\ref{thm:llmf1}, i.e.~$X_n := h(\overline{\loc}_n,\overline{\pv}_n)$.
 By definition, for all~$n$, we have $\expv(X_{n+1}) + \expv(C_n) \leq \expv(X_n)$, hence by induction, we get $\expv(X_{n+1}) + \sum_{m=0}^n \expv(C_m) \leq \expv(X_0)$. Given that $h$ is nonnegative, $\expv(X_{n+1})\ge 0$, so $\sum_{m=0}^n \expv(C_m) \leq \expv(X_0)$.
By applying the MCT, we obtain $\expv(C_\infty) = \expv(\lim_{n \rightarrow \infty} \sum_{m=0}^n C_m) = \lim_{n \rightarrow \infty} \sum_{m=0}^n \expv(C_m) \le \expv(X_0)$, which is the desired result.
For a more detailed proof, see Appendix~\ref{app:PUPFs2}.
\end{proof}

\begin{remark}
	The motivating example in Section~\ref{sec:problin}, i.e. the species fight stochastic linear recurrence, has unbounded nonnegative costs. Hence, nonnegative PUCS's lead to sound upper bounds on its expected accumulated cost.
\end{remark}

\rd{
\begin{remark} We remark two points about general updates:
	\begin{compactitem}
		\item As in Theorem~\ref{thm:lmf2}, our approach for general updates requires the PUCS to be nonnegative. However, note that our approach for bounded updates (Section~\ref{sec:soundness}), does not have this requirement.
		\item Since we are considering real-valued variables, unbounded updates cannot always be replaced with bounded updates and loops. (e.g.~consider $x := 3.1415 \cdot x.$)
	\end{compactitem}
\end{remark}
}

\section{Algorithmic Approach} \label{sec:computational}


In this section, we provide automated algorithms that, given a program $P$, an initial valuation $\pv^*$, a linear invariant $I$ and a constant  $d$, synthesize a PUCS/PLCS of degree $d$. For brevity, we only describe our algorithm for PUCS synthesis. A PLCS can be synthesized in the same manner. Our algorithms run in polynomial time and reduce the problem of PUCS/PLCS synthesis to a linear programming instance by applying Handelman's theorem.

In order to present Handelman's theorem, we need a few basic definitions. Let $X$ be a finite set of variables and $\Gamma \subseteq \setR[X]$ a finite set of linear functions (degree-$1$ polynomials) over $X$. We define $\sat{\Gamma} \subseteq \val{X}$ as the set of all valuations $v$ to the variables in $X$ that satisfy $g_i(v) \geq 0$ for all $g_i \in \Gamma$. We also define the monoid set of $\Gamma$ as
$$
\monoid(\Gamma) :=\left\{\prod\limits_{i=1}^t g_i \ \vert~~t\in \setN \cup \{0\} ~\wedge~ \ g_1,\ldots,g_t \in\Gamma\right\}.
$$
By definition, it is obvious that if $g \in \monoid(\Gamma)$, then for every $v \in \sat{\Gamma}$, we have $g(v) \geq 0$. Handelman's theorem characterizes every polynomial $g$ that is positive over $\sat{\Gamma}$.

\begin{theorem}[\bf Handelman's Theorem~\cite{handelman1988representing}]\label{thm:Handelman}
	Let $g\in\setR[X]$ and $g(\mathbf{x}) > 0$ for all $\mathbf{x}\in\langle \Gamma\rangle$. If $\langle \Gamma\rangle$ is compact, then
	\[
	 g= \sum\limits_{k=1}^s c_k \cdot f_k \  \  \ \quad \quad \handelmanformat
	\]
	for some $s\in\Nset$, $c_1,\ldots,c_s>0$ and $f_1,\ldots,f_s\in Monoid(\Gamma)$.
\end{theorem}

Intuitively, Handelman's theorem asserts that every polynomial $g$ that is positive over $\sat{\Gamma}$ must be a positive linear combination of polynomials in $\monoid(\Gamma)$. This means that in order to synthesize a polynomial that is positive over $\sat{\Gamma}$ we can limit our attention to polynomials of the form $\handelmanformat$. When using Handelman's theorem in our algorithm, we fix a constant $K$ and only consider those elements of $Monoid(\Gamma)$ that are obtained by $K$ multiplicands or less.

We now have all the required tools to describe our algorithm for synthesizing a PUCS.

\smallskip\paragraph{PUCS Synthesis Algorithm.} The algorithm has four steps:
\begin{compactenum}[(1)]
	\item \emph{Creating a Template for $h$.} Let $X = \pvars$ be the set of program variables. According to (C1), we aim to synthesize a PUCS $h$, such that for each label $\loc_i$ of the program, $h(\loc_i)$ is a polynomial of degree at most $d$ over $X$. Let $M_d(X) = \{\bar{f}_1, \bar{f}_2, \ldots, \bar{f}_r\}$ be the set of all monomials of degree at most $d$ over the variables $X$. Then, $h(\loc_i)$ has to be of the form  $\sum_{j=1}^r a_{ij} \cdot \bar{f}_j$ for some unknown real values $a_{ij}$. We call this expression a template for $h(\loc_i)$. Note that by condition (C2) the template for $h(\lout)$ is simply $h(\lout) = 0$. The algorithm computes these templates at every label $\loc_i$, treating the $a_{ij}$'s as unknown variables.
	
	\item \emph{Computing Pre-expectation.} The algorithm symbolically computes a template for $\pre_{h}$ using Definition~\ref{def:pre-exp1} and the template obtained for $h$ in step (1). This template will also contain $a_{ij}$'s as unknown variables.
	
	\item \emph{Pattern Extraction.} The algorithm then processes condition (C3) by symbolically computing polynomials 
$g = h(\loc_i) - \pre_h(\loc_i)$ for every label $\loc_i$. Then, as in Handelman's theorem, it rewrites each $g$ on the left-hand-side of the equations above in the form $\handelmanformat$, using the linear invariant $I(\loc_i)$ as the set $\Gamma$ of linear functions.
The nonnegativity of $h$ is handled in a similar way.
This effectively translates (C3) and the nonnegativity into a system $S$ of linear equalities over the $a_{ij}$'s and the new nonnegative unknown variables $c_k$ resulting from equation $\handelmanformat$.
	\item \emph{Solution via Linear Programming.} The algorithm calls an LP-solver to find a solution of $S$ that optimizes $h(\lin, \pv^*)$.
\end{compactenum}

If the algorithm is successful, i.e.~if the obtained system of linear equalities is feasible, then the solution to the LP contains values for the unknowns $a_{ij}$ and hence, we get the coefficients of the PUCS $h$. Note that we are optimizing for $h(\lin, \pv^*)$, so the obtained PUCS is the one that produces the best polynomial upper bound for the expected accumulated cost of $P$ with initial valuation $\pv^*$. We use the same algorithm for PLCS synthesis, except that we replace (C3) with (C3').

\begin{theorem}
	The algorithm above has polynomial runtime and synthesizes sound upper and lower bounds for the expected accumulated cost of the given program $P$.
\end{theorem}

\begin{proof}
	
	Step (1) ensures that (C1), (C2) are satisfied, while step (3) forces the polynomials $h$ and $g$'s to be nonnegative, ensuring nonnegativity and (C3). So the synthesized $h$ is a PUCS.
	Steps (1)--(3) are polynomial-time symbolic computations. Step (4) solves an LP of polynomial size. Hence, the runtime is polynomial wrt the length of the program.
	The reasoning for PLCS synthesis is similar.
\end{proof}
\lstset{language=prog}
\lstset{tabsize=3}
\newsavebox{\runningexampleaa}
\begin{lrbox}{\runningexampleaa}
\begin{lstlisting}[mathescape]
$1$: $[x\ge 0]$ while $x\ge 1$ do
$2$: $[x\ge 1]$   $x:=x+r$;
$3$: $[x\ge 0]$   $y:=r'$;
$4$: $[x\ge 0]$   tick$(x * y)$
      od
$5$: $[x\ge 0]$
\end{lstlisting}
\end{lrbox}

\begin{example}\label{al:simple}
Consider the program in Figure~\ref{fig:ex2} (Page~\pageref{fig:ex2}).
Let the initial valuation be $x_0 = 100, y_0 = 0$. To obtain a quadratic PUCS, i.e.~$d = 2$, our algorithm proceeds as follows:
\begin{compactenum}[(1)]
	\item A quadratic template is created for $h$, by setting $h(\loc_n,x,y):=a_{n1}\cdot x^2+a_{n2}\cdot xy+a_{n3}\cdot x+a_{n4}\cdot y^2+a_{n5}\cdot y+a_{n6}$. This template contains all monomials of degree $2$ or less.
	\item A template for the function $\pre_h$ is computed in the same manner as in Example~\ref{ex:simple}, except for that the computation is now symbolic and contains the unknown variables $a_{ij}$. The resulting template is  presented
in Table~\ref{ta:runningexample}.
	\item For each label $\loc_i$, the algorithm symbolically computes $g = h(\loc_i) - \pre_h(\loc_i)$. For example, for $\loc_3$, the algorithm computes $g(x, y) = h(\loc_3, x, y) - \pre_h(\loc_3, x, y) = (a_{31} - a_{41}) \cdot x^2 + a_{32} \cdot x y + (a_{33} - \frac{1}{3} a_{42} - a_{43}) \cdot x + a_{34} \cdot y^2 + a_{35} \cdot y + a_{36} - a_{44} - \frac{1}{3} a_{45} - a_{46}$. It then rewrites $g$ according to $\handelmanformat$ using \rd{$\Gamma = \{x\}$}, i.e.~$g(x, y) = \sum c_k \cdot f_k (x, y)$. This is because we need to ensure $g \geq 0$ to fulfill condition (C3). This leads to the polynomial equation $\sum c_k \cdot f_k(x, y) = (a_{31} - a_{41}) \cdot x^2 + a_{32} \cdot x y + (a_{33} - \frac{1}{3} a_{42} - a_{43}) \cdot x + a_{34} \cdot y^2 + a_{35} \cdot y + a_{36} - a_{44} - \frac{1}{3} a_{45} - a_{46}$. \rd{Note that both sides of this equation are polynomials. The coefficients of the polynomial on the LHS are combinations of $c_k$'s and those of the RHS are combinations of $a_{ij}$'s. Given that two polynomials are equal if and only if all of their corresponding terms have the same coefficients, the equality above} can be translated to several linear equations in terms of the $c_k$'s and $a_{ij}$'s by equating the coefficient of each term on both sides. These linear equations are generated at every label.
	\item The algorithm calls an LP-solver to solve the system consisting of all linear equations obtained in step (3). Given that we are looking for an optimal upperbound on the expected accumulated cost with the initial valuation $x_0 = 100, y_0 = 0$, the algorithm minimizes $h(\loc_1, 100, 0) = 10000 \cdot a_{11} + 100 \cdot a_{13} + a_{16}$ subject to these linear equations.
\end{compactenum}
In this case, the resulting values for $a_{ij}$'s lead to the same PUCS $h$ as in Figure~\ref{fig:ex2}. So the upper bound on the expected accumulated cost is $\frac{1}{3} x_0^2 + \frac{1}{3} x_0 = 3366.\overline{6}$. The algorithm can similarly synthesize a PLCS. In this case, the same function $h$ is reported as a PLCS. Therefore, the exact  expected accumulated cost of this program is $3366.\overline{6}$ and our algorithm is able to compute it precisely. See Appendix~\ref{app:detail1} for more details. 
\end{example}

\begin{table}
\caption{Template for $\pre_h$ of the program in Figure~\ref{fig:ex2}}
	\begin{minipage}{\columnwidth}
		\begin{center}
			\footnotesize{
			\begin{tabular}{|l|p{8cm}<{\centering}|}
				\hline
				$n$                &$pre_h(\loc_n,x,y)$      \\
				\hline
				$1$
				&$\mathbf{1}_{x \geq 1}\cdot (a_{21}\cdot x^2 +a_{22} \cdot xy +a_{23}\cdot x+a_{24}\cdot y^2+a_{25}\cdot y+a_{26})+\mathbf{1}_{x < 1}\cdot 0$ \\
				\hline
				$2$   &$a_{31}\cdot x^2+a_{32}\cdot xy+(a_{33}-a_{31})\cdot x
				+a_{34}\cdot y^2+(a_{35}-\frac{1}{2}a_{32})\cdot y +a_{31}-\frac{1}{2}a_{33}+a_{36}$ \\
				\hline
				$3$   &$ a_{41}\cdot x^2+(\frac{1}{3} a_{42}+c_{43})\cdot x+a_{44}+\frac{1}{3} a_{45}+a_{46} $   \\
				\hline
				$4$   &$a_{11}\cdot x^2 +(a_{12}+1) \cdot xy +a_{13}\cdot x+a_{14}\cdot y^2+a_{15}\cdot y+a_{16} $   \\
				\hline
				$5$  & 0\\
				\hline
			\end{tabular}
		}
		\end{center}
	\end{minipage}
	\label{ta:runningexample}

\end{table}

\section{Experimental Results}\label{sec:experiment}

\begin{table*}
	\caption{Comparison of our approach with \cite{pldi18}.}
		\resizebox{0.98\textwidth}{!}{
		\footnotesize{
		\begin{tabular}{c|p{5cm}<{\centering}|p{5cm}<{\centering}|p{5cm}<{\centering}}
			\toprule
			
			\textbf{Program}  & Upper bound of \cite{pldi18}& \textbf{$h(\lin, \pv)$ in PUCS} & \textbf{$h(\lin, \pv)$ in PLCS} \\
			\hline
			
			\texttt{ber} &$2\cdot n-2\cdot x$  &$2\cdot n-2\cdot x$    & $2\cdot n-2\cdot x -2 $    \\
			\hline
			
			\texttt{bin} & $ 0.2\cdot n +1.8 $ &$ 0.2\cdot n +1.8 $   & $ 0.2\cdot n  -0.2$ \\
			\hline

			\texttt{linear01} & $0.6\cdot x$    &$0.6\cdot x$  & $0.6\cdot x-1.2$\\
			\hline
			
			\texttt{prdwalk} & $1.14286\cdot n - 1.14286\cdot x+4.5714$ &$1.14286\cdot n - 1.14286\cdot x+4.5714$ & $1.14286\cdot n - 1.14286\cdot x-1.1429$ \\
			\hline

			\texttt{race} & $0.666667\cdot t-0.666667\cdot h +6$ &$\frac{2}{3}\cdot t-\frac{2}{3}\cdot h+6$ & $\frac{2}{3}\cdot t-\frac{2}{3}\cdot h$ \\
			\hline

			\texttt{rdseql} &$2.25\cdot x+y $ &$2.25\cdot x+y +2.25$ & $2\cdot x$  \\
			\hline

			\texttt{rdwalk} & $ 2\cdot n- 2\cdot x+2$  & $ 2\cdot n- 2\cdot x+2$    & $2\cdot n - 2\cdot x - 2$  \\
			\hline

			\texttt{sprdwalk} &$2\cdot n - 2\cdot x$  &$2\cdot n - 2\cdot x$  &$2\cdot n- 2\cdot x-2$ \\
			\hline

			\texttt{C4B\_t13}& $1.25\cdot x +y$ &$1.25\cdot x +y$  & $x-1$ \\
			\hline

			\texttt{prnes} & $0.052631\cdot y -68.4795\cdot n $  &$0.05263\cdot y -68.4795\cdot n $ &$- 10\cdot n-10 $  \\
			\hline
			
			\texttt{condand} & $m+n$ &$m+n-1$ & $0$  \\
			\hline
			
			\texttt{pol04} & $4.5\cdot x^2 + 7.5\cdot x $ &$4.5\cdot x^2 + 10.5\cdot x $ &$0$  \\
			\hline
			
			\texttt{pol05} & $x^2+x$ &$0.5\cdot x^2 + 2.5\cdot x $ &$0$  \\
			\hline
			
			\texttt{rdbub} & $3\cdot n^2$ &$3\cdot n^2  $ &$0$  \\
			\hline
			
			\texttt{trader} &$- 5\cdot s_{min}^2 -5\cdot s_{min}+5\cdot s^2+5\cdot s $  &$- 5\cdot s_{min}^2 -5\cdot s_{min}+5\cdot s^2+5\cdot s $ &$0$  \\
			
			\bottomrule
		\end{tabular}
	}}
	\label{tab:expr3}
	
\end{table*}

\begin{table*}
	\caption{Symbolic upper and lower bounds, i.e.~$h(\lin,\pv)$, obtained through PUCS and PLCS.}
	\resizebox{.98\textwidth}{!}{
	\footnotesize{
	\begin{tabular}{c|c|C{5cm}|C{5cm}|c}
		\toprule
		
		\textbf{Program} & $\mathbf{\pv_0}$ & \textbf{$h(\lin, \pv)$ in PUCS} & \textbf{$h(\lin, \pv)$ in PLCS } & Runtime (s) \\
		\hline
		
		\makecell{Bitcoin Mining\\(Figure~\ref{fig:mining})} & $x_0 = 100$ &$1.475 - 1.475\cdot x$    & $-1.5\cdot x $  &  $9.24$ \\
		\hline

		\makecell{Bitcoin Mining Pool\\(Figure~\ref{fig:pool})} & $y_0 = 100$ &$ -  7.375\cdot y^2 - 41.62\cdot y + 49.0 $   & $ - 7.5\cdot y^2 - 67.5\cdot y$ & $27.81$ \\
		\hline

		\makecell{Queuing Network\\(Figure~\ref{fig:FJcode})} & 				$n_0 = 320$ &$0.0492\cdot n- 0.0492\cdot i +0.0103\cdot l_1^2+0.00342\cdot l_2^3+0.00726\cdot l_2^2 + 0.0492 $  & $0.0384\cdot n- 0.0384\cdot i- (1.76\times 10^{-4})\cdot l_1^2 - 0.00854\cdot l_1\cdot l_2^2 -(8.16\times 10^{-5})\cdot l_2^3- 0.00173\cdot l_2^2 + 0.0384 $ & $282.44$\\
		\hline

		\makecell{Species Fight\\(Figure~\ref{fig:species})}	& $a_0 = 16, b_0 = 10$ & $40\cdot a\cdot b - 180\cdot b - 180\cdot a + 810$ &--  & $16.30$\\
		\hline

		Figure~\ref{fig:example} & 	$x_0 = 200$ & $\frac{1}{3}\cdot x^2+\frac{1}{3}\cdot x$ & $\frac{1}{3}\cdot x^2+\frac{1}{3}\cdot x-\frac{2}{3}$ & $6.00$ \\
		\hline

		Nested Loop  &  $i_0 = 150$ & $\frac{1}{3}\cdot i^2+i$ & $\frac{1}{3}\cdot i^2-\frac{1}{3}\cdot i$  & $31.73$\\
		\hline

		Random Walk & $x_0 = 12, n_0 = 20$ & $ 2.5\cdot x- 2.5\cdot n$    & $2.5\cdot x - 2.5\cdot n - 2.5$  & $15.02$ \\
		\hline

		2D Robot & 	$x_0 = 100, y_0 = 80$ & $1.728\cdot x^2 - 3.456\cdot x\cdot y + 31.45\cdot x + 1.728\cdot y^2 - 31.45\cdot y + 126.5$  &$1.728\cdot x^2 - 3.456\cdot x\cdot y + 31.45\cdot x + 1.728\cdot y^2 - 31.45\cdot y $ & $40.28$\\
		\hline

		Goods Discount & 	$n_0 = 200, d_0 = 1$ & $0.00667\cdot d\cdot n - 0.7\cdot n - 3.803\cdot d + 0.00222\cdot d^2 + 119.4$  & $0.00667\cdot d\cdot n - 0.7133\cdot n - 3.812\cdot d + 0.00222\cdot d^2 + 112.4 $ & $16.89$\\
		\hline

		Pollutant Disposal &	$n_0 = 200$ &  $- 0.2\cdot n^2 + 50.2\cdot n $ &$- 0.2\cdot n^2 + 50.2\cdot n - 482.0 $  & $19.53$\\
		\bottomrule
	\end{tabular}
}
}
	\label{tab:expr2}
\end{table*}

We now report on an implementation of our approach and present experimental results. First, we compare our approach with~\cite{pldi18}.
Then, we show that our approach is able to handle programs that no previous approach could.

\paragraph{Implementation and Environment.} We implemented our approach in Matlab R2018b. We used the Stanford Invariant Generator~\cite{sting1} to find linear invariants and the tool in~\cite{ChatterjeeFG16} to ensure the concentrated termination property for the input programs. The results were obtained on a Windows machine with an Intel Core i7 3.6GHz processor and 8GB of RAM.

\paragraph{Comparison with~\cite{pldi18}.} \rd{ We ran our approach on several benchmarks from~\cite{pldi18}. The results are reported in Table~\ref{tab:expr3}. In general, the upper bounds obtained by our approach and~\cite{pldi18} are very similar. We obtain identical results on most benchmarks. Specifically, our leading coefficient is never worse than~\cite{pldi18}. The only cases where we get different bounds are \texttt{rdseql} (our bound is worse by a small additive constant), \texttt{condand} (our bound is better by a small additive constant), \texttt{pol04} (our bound is worse only in a non-leading coefficient), \texttt{pol05} (our bound is better in the leading coefficient).
	Moreover, note that we provide lower bounds through PLCS, while~\cite{pldi18} cannot obtain any lower bounds. 
}

\paragraph{Experimental Results on New Benchmarks.}
Table~\ref{tab:expr2} provides our experimental results over ten new benchmarks. These include the four motivating examples of Section~\ref{sec:motivation}, our running example, and five other classical programs. In each case, we optimized the bounds wrt a fixed initial valuation $\pv_0$. We report the upper (resp.~lower) bound obtained through PUCS (resp.~PLCS). Note that we do not have lower bounds for the Species Fight example as its updates are unbounded. \rd{See Appendix~\ref{app:results} for plots and details about the benchmarks.}  In all cases of Table~\ref{tab:expr2}, the obtained lower and upper bounds are very close, and in many cases they meet. Hence, our approach can obtain tight bounds on the expected accumulated cost of a variety of programs that could not be handled by any previous approach. Moreover, the reported runtimes show the efficiency of our algorithms in practice.

\section{Related Work} \label{sec:rel}

\paragraph{Termination and cost analysis.}
Program termination has been
studied extensively~\cite{DBLP:conf/pldi/CookPR06,DBLP:conf/esop/KuwaharaTU014,DBLP:conf/tacas/CookSZ13,DBLP:conf/sas/Urban13,DBLP:conf/popl/CousotC12,DBLP:conf/popl/LeeJB01,DBLP:journals/toplas/Lee09,DBLP:journals/fmsd/CookPR09,DBLP:conf/vmcai/AlurC10}.
Automated amortized cost analysis has also been widely studied~\cite{DBLP:conf/aplas/HoffmannH10,DBLP:conf/esop/HoffmannH10,DBLP:conf/popl/HofmannJ03,DBLP:conf/esop/HofmannJ06,DBLP:conf/csl/HofmannR09,DBLP:conf/fm/JostLHSH09,DBLP:conf/popl/JostHLH10,Hoffman1,DBLP:conf/popl/GimenezM16,DBLP:conf/cav/SinnZV14,DBLP:conf/sas/AliasDFG10,ChatterjeeFG16,SPEED1,SPEED2,DBLP:conf/cav/GulavaniG08}.
Other resource analysis approaches include: (a)~recurrence relations for worst-case analysis~\cite{DBLP:conf/icfp/Grobauer01,DBLP:journals/tcs/FlajoletSZ91,DBLP:journals/entcs/AlbertAGGPRRZ09,DBLP:conf/sas/AlbertAGP08,DBLP:conf/esop/AlbertAGPZ07}; (b)~average-case analysis by recurrence relations~\cite{DBLP:conf/cav/ChatterjeeFM17} and (c)~using theorem-proving~\cite{DBLP:conf/popl/SrikanthSH17}.
These approaches do not consider probabilistic programs.

\paragraph{Ranking functions.}
Ranking functions have been widely studied for intraprocedural analysis~\cite{BG05,DBLP:conf/cav/BradleyMS05,DBLP:conf/tacas/ColonS01,DBLP:conf/vmcai/PodelskiR04,DBLP:conf/pods/SohnG91,DBLP:conf/vmcai/Cousot05,DBLP:journals/fcsc/YangZZX10,DBLP:journals/jossac/ShenWYZ13,DBLP:conf/cav/ChatterjeeFG17}.
Most works focus on linear/polynomial ranking functions and target non-probabilistic programs~\cite{DBLP:conf/tacas/ColonS01,DBLP:conf/vmcai/PodelskiR04,DBLP:conf/pods/SohnG91,DBLP:conf/vmcai/Cousot05,DBLP:journals/fcsc/YangZZX10,DBLP:journals/jossac/ShenWYZ13}.
They have been extended in various directions, such as: symbolic approaches~\cite{DBLP:journals/toplas/BrockschmidtE0F16},
proof rules for deterministic programs~\cite{DBLP:journals/fac/Hesselink94},
sized types~\cite{DBLP:journals/lisp/ChinK01,DBLP:conf/icfp/HughesP99,DBLP:conf/popl/HughesPS96},
and polynomial resource bounds~\cite{DBLP:conf/tlca/ShkaravskaKE07}. Moreover,
\cite{DBLP:conf/cav/GulavaniG08} generates bounds through abstract interpretation using inference systems.
However all of these methods are also for non-probabilistic programs only.

\paragraph{Ranking supermartingales.}
Ranking functions have been extended to ranking supermartingales and studied in~\cite{SriramCAV,HolgerPOPL,ChatterjeeFNH16,ChatterjeeFG16,DBLP:journals/corr/ChatterjeeF17,ChatterjeeNZ2017,AgrawalC018}. Proof rules for probabilistic programs are provided in~\cite{JonesPhdThesis,OLKMLICS2016}.
However, these works consider qualitative termination problems. They do not consider precise cost analysis, which is the focus of our work.

\paragraph{Cost analysis for probabilistic programs.}
The most closely-related work is~\cite{pldi18}.
A detailed comparison has been already provided in Section~\ref{sec:novelty}.
In particular, we handle positive and negative costs, as well as unbounded costs, whereas~\cite{pldi18} can handle
only positive bounded costs.
Another related work is~\cite{ijcai18} which considers succinct Markov decision processes (MDPs) and
bounds for such MDPs. However, these MDPs only have single while loops and only linear bounds are obtained.
Our approach considers polynomial bounds for general probabilistic programs.

\newpage
\section{Conclusion}\label{sect:conclusion}
We considered the problem of cost analysis of probabilistic programs. 
While previous approaches only handled positive bounded costs,
our approach can derive polynomial bounds for programs with both positive
and negative costs. 
It is sound for general costs and bounded updates, and general updates
with nonnegative costs. However, finding sound approaches that can handle general costs and general updates remains 
an interesting direction for future work.
Another interesting direction is finding 
non-polynomial bounds.

\newpage
\begin{acks}
	The research was partially supported by National Natural Science Foundation of China (NSFC) Grants No. 61802254, 61772336, 61872142, 61672229, 61832015, Shanghai Municipal Natural Science Foundation (16ZR1409100), Open Project of Shanghai Key Laboratory of Trustworthy Computing, Vienna Science and
	Technology Fund (WWTF) Project ICT15-003, Austrian Science
	Fund (FWF) NFN Grant No S11407-N23 (RiSE/SHiNE), ERC
	Starting Grant (279307: Graph Games), an IBM PhD Fellowship,
	and a DOC Fellowship of the Austrian Academy of Sciences
	(\"{O}AW).
\end{acks}

\clearpage

\clearpage
\appendix

\section{Conditional Expectation}\label{app:condexpv}

Let $X$ be any random variable from a probability space $(\Omega, \mathcal{F},\probm)$ such that $\expv(|X|)<\infty$.
Then given any $\sigma$-algebra $\mathcal{G}\subseteq\mathcal{F}$, there exists a random variable (from $(\Omega, \mathcal{F},\probm)$), conventionally denoted by $\condexpv{X}{\mathcal{G}}$, such that
\begin{compactitem}
\item[(E1)] $\condexpv{X}{\mathcal{G}}$ is $\mathcal{G}$-measurable, and
\item[(E2)] $\expv\left(\left|\condexpv{X}{\mathcal{G}}\right|\right)<\infty$, and
\item[(E3)] for all $A\in\mathcal{G}$, we have $\int_A \condexpv{X}{\mathcal{G}}\,\mathrm{d}\probm=\int_A {X}\,\mathrm{d}\probm$.
\end{compactitem}
The random variable $\condexpv{X}{\mathcal{G}}$ is called the \emph{conditional expectation} of $X$ given $\mathcal{G}$.
The random variable $\condexpv{X}{\mathcal{G}}$ is a.s. unique in the sense that if $Y$ is another random variable satisfying (E1)--(E3), then $\probm(Y=\condexpv{X}{\mathcal{G}})=1$.

Conditional expectation has the following properties for any random variables $X,Y$ and $\{X_n\}_{n\in\Nset_0}$ (from a same probability space) satisfying $\expv(|X|)<\infty,\expv(|Y|)<\infty, \expv(|X_n|)<\infty$ ($n\ge 0$) and any suitable sub-$\sigma$-algebras $\mathcal{G},\mathcal{H}$:
\begin{compactitem}
\item[(E4)] $\expv\left(\condexpv{X}{\mathcal{G}}\right)=\expv(X)$ ;
\item[(E5)] if $X$ is $\mathcal{G}$-measurable, then $\condexpv{X}{\mathcal{G}}=X$ a.s.;
\item[(E6)] for any real constants $b,d$,
\[
\condexpv{b\cdot X+d\cdot Y}{G}=b\cdot\condexpv{X}{G}+d\cdot \condexpv{Y}{G}\mbox{ a.s.;}
\]
\item[(E7)] if $\mathcal{H}\subseteq\mathcal{G}$, then $\condexpv{\condexpv{X}{\mathcal{G}}}{\mathcal{H}}=\condexpv{X}{\mathcal{H}}$ a.s.;
\item[(E8)] if $Y$ is $\mathcal{G}$-measurable and $\expv(|Y|)<\infty$, $\expv(|Y\cdot X|)<\infty$,
then
\[
\condexpv{Y\cdot X}{\mathcal{G}}=Y\cdot\condexpv{X}{\mathcal{G}}\mbox{ a.s.;}
\]
\item[(E9)] if $X$ is independent of $\mathcal{H}$, then $\condexpv{X}{\mathcal{H}}=\expv(X)$ a.s., where $\expv(X)$ here is deemed as the random variable with constant value $\expv(X)$;
\item[(E10)] if it holds a.s that $X\ge 0$, then $\condexpv{X}{\mathcal{G}}\ge 0$ a.s.;
\item[(E11)] if it holds a.s. that (i) $X_n\ge 0$ and $X_n\le X_{n+1}$ for all $n$ and (ii) $\lim\limits_{n\rightarrow\infty}X_n=X$, then
\[
\lim\limits_{n\rightarrow\infty}\condexpv{X_n}{\mathcal{G}}=\condexpv{X}{\mathcal{G}}\mbox{ a.s.}
\]
\item[(E12)] if (i) $|X_n|\le Y$ for all $n$ and (ii) $\lim\limits_{n\rightarrow\infty} X_n=X$, then
\[
\lim\limits_{n\rightarrow\infty}\condexpv{X_n}{\mathcal{G}}=\condexpv{X}{\mathcal{G}}\mbox{ a.s.}
\]
\item[(E13)] if $g:\Rset\rightarrow\Rset$ is a convex function and $\expv(|g(X)|)<\infty$, then $g(\condexpv{X}{\mathcal{G}})\le \condexpv{g(X)}{\mathcal{G}}$ a.s.
\end{compactitem}
We refer to~\cite[Chapter~9]{williams1991probability} for more details.

\section{Detailed Syntax}\label{app:syntax}

In the sequel, we fix two countable sets of \emph{program variables} and \emph{sampling variables}.
W.l.o.g, these three sets are pairwise disjoint.

Informally, program variables are variables that are directly related to the control-flow of a program, while sampling variables reflect randomized inputs to the program.
Every program variable holds an integer upon instantiation,
while every sampling variable is bound to a discrete probability distribution. 

\noindent{\bf The Syntax.} 
Below we explain the grammar in in Figure~\ref{fig:syntax}.
\begin{compactitem}
\item \emph{Variables.} Expressions $\langle\mathit{pvar}\rangle$ (resp. $\langle\mathit{rvar}\rangle$) range over program (resp. sampling) variables.
\item \emph{Constants.} Expressions $\langle\mathit{const}\rangle$ range over decimals.
\item \emph{Arithmetic Expressions.} Expressions $\langle\mathit{expr}\rangle$ (resp. $\langle\mathit{pexpr}\rangle$) range over arithmetic expressions over both program and sampling variables (resp. program variables). As a theoretical paper, we do not fix the syntax for $\langle\mathit{expr}\rangle$ and $\langle\mathit{pexpr}\rangle$.

\item \emph{Boolean Expressions.} Expressions $\langle\mathit{bexpr}\rangle$ range over propositional arithmetic predicates over program variables.
\item \emph{Nondeterminism.} The symbol `$\star$' indicates a nondeterministic choice to be resolved in a demonic way.
\item \emph{Statements $\langle \mathit{stmt}\rangle$.} Assignment statements are indicated by `$:=$';
`\textbf{skip}' is the statement that does nothing;
conditional branches and nondeterminism are both indicated by the keyword `\textbf{if}';
while-loops are indicated by the keyword `\textbf{while}';
sequential compositions are indicated by semicolon;
finally, tick statements are indicated by `\textbf{tick}'.  
\end{compactitem}

\section{Detailed Semantics}\label{app:semantic}
Informally, a control-flow graph specifies how values for program variables and the program counter change along an execution of a program.
We refer to the status of the program counter as a \emph{label}, and assign an initial label $\lin{}$ and a terminal label $\lout{}$ to the start and the end of the program.
Moreover, we have five types of labels, namely \emph{assignment}, \emph{branching},  \emph{probabilistic}, \emph{nondeterministic} and \emph{tick} labels.
\begin{compactitem}
\item An \emph{assignment} label corresponds to an assignment statement indicated by `$:=$', and leads to the next label right after the statement with change of values specified by the update function determined at the right-hand-side of `$:=$'. The update function gives the next valuation on program variables, based on the current values of program variables and the sampled values for this statement.
\item A \emph{branching} label corresponds to a conditional-branching statement indicated by the keyword `\textbf{if}' or `\textbf{while}' together with a propositional arithmetic predicate $\phi$ over program variables (as the condition or the loop guard), and leads to the next label determined by $\phi$ without change on values.

\item A \emph{probabilistic} label corresponds to a probabilistic-branching statement indicated by the keywords `\textbf{if}' and `\textbf{prob}($p$)' with $p\in [0,1]$, and leads to the labels  of the \textbf{then}-branch with probability $p$ and the \textbf{else}-branches with probability $1-p$, without change on values.
\item A \emph{nondeterministic} label corresponds to a nondeterministic-branching statement indicated by the keywords `\textbf{if}' and `$\star$', and leads to the labels of the \textbf{then}- and \textbf{else}-branches without change on values.
\item A \emph{tick} label corresponds to a tick statement `\textbf{tick}($q$)' that triggers a cost/reward, and leads to the next label without change on values. The arithmetic expression
$q$ determines a \emph{cost function} that outputs a real number (as the amount of cost/reward) upon the current values of program variables for this statement. 
\end{compactitem}

It is intuitively clear that any probabilistic program can be transformed into a CFG.
We refer to existing results~\cite{ChatterjeeFG16,ChatterjeeFNH16} for a detailed transformation
from programs to CFGs.

Based on CFGs, the semantics of the program is given by general state space Markov chains (GSSMCs) as follows.
Below we fix a probabilistic program $W$ with its CFG in the form (\ref{eq:cfg}).
To illustrate the semantics, we need the notions of \emph{configurations}, \emph{sampling functions}, \emph{runs} and \emph{schedulers} as follows.

\noindent{\em Configurations.}  A \emph{configuration} is a
triple $(\loc,\nu)$ where $\loc\in\locs{}$ and
$\nu\in\val{\pvars{}}$.
We say that a configuration $(\loc,\nu)$ is \emph{terminal} if $\loc=\lout{}$;
moreover, it is \emph{nondeterministic} if $\loc\in\Dlocs{}$.
Informally, a configuration $(\loc,\nu)$ specifies that the next statement to be executed is the one labelled with $\loc$ and the current values
of program variables is specified by the valuation $\nu$.

\noindent{\em Sampling functions.}
A \emph{sampling function} $\Upsilon$ is a function assigning to every sampling variable $r\in\rvars$ a (possibly continuous) probability distribution over $\Rset$.
Informally, a sampling function $\Upsilon$ specifies the probability distributions for the sampling of all sampling variables, i.e.,
for each $r\in\rvars$, its sampled value is drawn from the probability distribution $\Upsilon(r)$.

\noindent{\em Finite and infinite runs.}
A \emph{finite run} $\rho$ is a finite sequence $(\loc_0,\nu_0),\dots,(\loc_n,\nu_n)$ of configurations.
An \emph{infinite run} is an infinite sequence $\{(\loc_n,\nu_n)\}_{n\in\Nset_0}$ of configurations.
The intuition is that each $\loc_n$ and $\nu_n$ are the current program counter and respectively the current valuation for program variables
at the $n$th step of a program execution.

\noindent{\em Schedulers.}
A scheduler $\sigma$ is a function that assigns to every finite run ending in a nondeterministic configuration $(\loc,\nu)$ a transition  with source label  $\loc$ (in the CFG)
that leads to the target label as the next label.
Thus, based on the whole history of configuration visited so far, a scheduler resolves the choice between the \textbf{then}- and \textbf{else}-branch at a nondeterministic branch.

Based on these notions, we can have an intuitive description on an execution of a probabilistic program.
Given a scheduler $\sigma$, the execution starts in an initial configuration $(\loc_0,\nu_0)$.
Then in every step $n\in\Nset_0$, assuming that the current configuration is $c_n=(\loc_n,\nu_n)$, the following happens.
\begin{compactitem}
\item If $\loc_n=\lout{}$ (i.e., the program terminates), then $(\loc_{n+1},\nu_{n+1})=(\loc_n,\nu_n)$. Otherwise, proceed as follows.
\item A valuation $\mathbf{r}$ on the sampling variables is sampled w.r.t the probability distributions in the sampling function $\Upsilon$.
\item A transition $\tau=(\loc_n,\alpha^*,\loc^*)$ enabled at the current configuration $(\loc_n,\nu_n)$ is chosen, and then the next configuration is determined by the chosen transition. In detail, we have the following.
      \begin{compactitem}
       \item If $\loc_n\in\alocs{}$, then $\tau$ is chosen as the unique transition from $\loc_n$ such that $\alpha^*$ is an update function, and the next configuration $(\loc_{n+1},\nu_{n+1})$ is set to be $(\loc^*, \alpha^*(\nu_n,\mathbf{r}))$.
       \item If $\loc_n\in\blocs{}$, then $\tau$ is chosen as the unique transition such that $\nu_n$ satisfies the propositional arithmetic predicate $\alpha^*$,
       and the next configuration $(\loc_{n+1},\nu_{n+1})$ is set to be $(\loc^*, \nu_n)$.
       \item If $\loc_n\in\plocs{}$ with the probability $p$ specified in its corresponding statement, then $\tau$ is chosen to be the \textbf{then}-branch with probability $p$ and the \textbf{else}-branch with probability $1-p$, and the next configuration $(\loc_{n+1},\nu_{n+1})$ is set to be $(\loc^*, \nu_n)$.
      \item If $\loc_n\in\Dlocs{}$, then $\tau$ is chosen by the scheduler $\sigma$. That is, if $\rho=c_0 c_1 \cdots c_n$ is the finite path of configurations traversed so far, then $\tau$ equals $\sigma(c_0 c_1 \cdots c_n)$, and the next configuration $(\loc_{n+1},\nu_{n+1})$ is set to be $(\loc^*, \nu_n)$.
      \item If $\loc_n\in\tlocs{}$, then $\tau$ is chosen as the unique transition from $\loc_n$ such that $\alpha^*$ is a cost function, then the next configuration $(\loc_{n+1},\nu_{n+1})$ is set to be $(\loc^*, \nu_n)$ and the statement triggers a cost of amount $\alpha^*(\nu_n)$.
      \end{compactitem}
\end{compactitem}
In this way, the scheduler and random choices eventually produce a random infinite run in a probabilistic program.
Then given any scheduler that resolves nondeterminism, the semantics of a probabilistic program is a GSSMC, where the kernel functions
can be directly defined over configurations and based on the transitions in the CFG so that
they specify the probabilities of the next configuration given the current configuration.

Given a scheduler $\sigma$ and an initial configuration $c$, the GSSMC of a probabilistic program induces a probability space where the sample space is the set of all infinite runs, the sigma-algebra is generated from cylinder sets of infinite runs, and the probability measure is determined by the scheduler and the random sampling in the program.

\section{Proofs for Martingale Results}\label{app:martingale}

\subsection{The Extended OST}\label{app:OST}

In the proof of the extended OST, for a stopping time $U$ and a nonnegative integer $n\in\Nset_0$, we denote by $U\wedge n$ the random variable $\min\{U,n\}$.

~\\

{\bf Theorem~\ref{eost}.}~(The Extended OST)
Consider any stopping time $U$ wrt a filtration $\{\mathcal{F}_n\}_{n=0}^\infty$ and any martingale (resp. supermartingale) $\{X_n\}_{n=0}^\infty$ adapted to $\{\mathcal{F}_n\}_{n=0}^\infty$ and let $Y = X_U$.
Then the following condition is sufficient to ensure that $\expv\left(|Y|\right)<\infty$ and $\expv\left(Y\right)=\expv(X_0)$ (resp. $\expv\left(Y\right)\le\expv(X_0)$):
\begin{compactitem}
\item There exist real numbers $M, c_1, c_2, d > 0$ such that (i) for sufficiently large $n \in \setN$, it holds that $\probm(U>n) \leq c_1 \cdot e^{-c_2 \cdot n}$ and (ii) for all $n \in \setN$, $\vert X_{n+1} - X_n \vert \leq M \cdot n^d$ almost surely.
\end{compactitem}

\begin{proof}

We only prove the ``$\le$'' case, the ``$=$'' case is similar. For every $n\in\Nset_0$,
\begin{eqnarray*}
\left|X_{U\wedge n}\right|&=& \left|X_0+\sum_{k=0}^{U\wedge n-1} \left(X_{k+1}-X_k\right)\right| \\
&=& \left|X_0+\sum_{k=0}^\infty \left(X_{k+1}-X_k\right)\cdot \mathbf{1}_{U>k\wedge n>k}\right|\\
&\le& \left|X_0\right|+\sum_{k=0}^\infty \left|\left(X_{k+1}-X_k\right)\cdot \mathbf{1}_{U>k\wedge n>k}\right| \\
&\le& \left|X_0\right|+\sum_{k=0}^\infty \left|\left(X_{k+1}-X_k\right)\cdot \mathbf{1}_{U>k}\right|\enskip. \\
\end{eqnarray*}

Then
\begin{eqnarray*}
& &   \expv\left(\left|X_0\right|+\sum_{k=0}^\infty \left|\left(X_{k+1}-X_k\right)\cdot \mathbf{1}_{U>k}\right|\right) \\
&=& \mbox{(By Monotone Convergence Theorem)} \\
& & \expv\left(\left|X_0\right|\right)+\sum_{k=0}^\infty \expv\left(\left|\left(X_{k+1}-X_k\right)\cdot \mathbf{1}_{U>k}\right|\right) \\
&=& \expv\left(\left|X_0\right|\right)+\sum_{k=0}^\infty \expv\left(\left|X_{k+1}-X_k\right|\cdot \mathbf{1}_{U>k}\right) \\
&\le & \expv\left(\left|X_0\right|\right)+\sum_{k=0}^\infty \expv\left( \lambda\cdot k^d\cdot \mathbf{1}_{U>k}\right) \\
&=& \expv\left(\left|X_0\right|\right)+\sum_{k=0}^\infty M\cdot k^d\cdot \probm\left(U>k\right) \\
&\le & \expv\left(\left|X_0\right|\right)+\sum_{k=0}^\infty M\cdot k^d\cdot c_{1}\cdot e^{-c_{2}\cdot k} \\
&=& \expv\left(\left|X_0\right|\right)+M\cdot c_1 \cdot \sum_{k=0}^\infty k^d\cdot e^{-c_{2}\cdot k} \\
&<& \infty\enskip.
\end{eqnarray*}
Thus, by Dominated Convergence Theorem
and the fact that $X_U=\lim\limits_{n\rightarrow\infty} X_{R\wedge n}$ a.s.,
\[
\expv\left(X_U\right)=\expv\left(\lim\limits_{n\rightarrow\infty} X_{U\wedge n}\right)=\lim\limits_{n\rightarrow\infty}\expv\left(X_{U\wedge n}\right)\enskip.
\]
Finally the result follows from properties for the stopped process $\{X_{U\wedge n}\}_{n\in\Nset_0}$ that 
\[
\expv\left(X_U\right) \le \expv\left(X_0\right)\enskip.
\]
\end{proof}

\subsection{An Important Lemma}

In this part, we prove an important lemma.
Below we define the following sequences of (vectors of) random variables:
\begin{compactitem}
\item $\overline{\pv}_0,\overline{\pv}_1,\dots$ where each $\overline{\pv}_n$ represents the valuation to program variables at the $n$th execution step of a probabilistic program; 
\item $\overline{\rv}_0,\overline{\rv}_1,\dots$ where each $\overline{\rv}_n$ represents the sampled valuation to sampling variables at the $n$th execution step of a probabilistic program;

\item $\overline{\ell}_0,\overline{\ell}_1,\dots$ where each $\overline{\ell}_n$ represents the label at the $n$th execution step of a probabilistic program.

\end{compactitem}

\begin{lemma}\label{app:lemm:important}

Let $h$ be a PUCS and $\sigma$ be any scheduler. Let the stochastic process $\{X_n\}_{n\in\Nset_0}$ be defined such that $X_n:=h(\overline{\ell}_n,\overline{\pv}_n)$.
Then for all $n\in\Nset_0$, we have $\condexpv{X_{n+1}+ C_n}{\mathcal{F}_n} \le pre_h (\overline{\ell}_n,\overline{\pv}_n)$.
\end{lemma}

\begin{proof}
For all $n\in\Nset_0$, from the program syntax we have
\[
X_{n+1}=\mathbf{1}_{\overline{\ell}_n =\lout{}}\cdot X_n +Y_p+Y_a+Y_{nd}+Y_t+Y_b
\]
where the terms are described below:
\[
Y_p:=\sum_{\loc\in\plocs{}}\left[\mathbf{1}_{\overline{\ell}_n=\loc} \cdot \sum_{i\in\{0,1\}} \mathbf{1}_{B_{\loc}=i} \cdot h(\loc_{B_{\loc}=i},\overline{\pv}_n)\right]
\]
where each random variable $B_{\loc}$ is the Bernoulli random variable for the decision of the probabilistic branch and $\loc_{B_{\loc}=0}, \loc_{B_{\loc}=1}$ are the corresponding successor locations of $\loc$. Note that all $B_{\loc}$ 's and $\rv$ 's are independent of $\mathcal{F}_n$ . In other words, $Y_p$ describes the semantics of probabilistic locations.
\[
Y_a:=\sum_{\loc\in\alocs{}}  \mathbf{1}_{\overline{\ell}_n =\loc} \cdot h(\loc',F_\ell(\overline{\pv}_n,\rv))
\]
describes the semantics of assignment locations where $\loc'$ is its successor label.

\[
Y_{nd}:=\sum_{\loc\in\Dlocs{}} \mathbf{1}_{\overline{\ell}_n =\loc} \cdot h(\sigma\left(\loc,\overline{\pv}_n \right),\overline{\pv}_n)
\]
describes the semantics of nondeterministic locations, where $\sigma\left(-,-\right)$ here denotes the target location of the transition chosen by the scheduler $\sigma$.
\[
Y_t:=\sum_{\loc\in\tlocs{}} \mathbf{1}_{\overline{\ell}_n =\loc} \cdot h(\loc',\overline{\pv}_n)
\]
describes the semantics of tick locations.
\[
Y_b:=\sum_{\loc\in\blocs{}} \left[\mathbf{1}_{\overline{\ell}_n = \loc} \cdot \sum_{i\in\{1,2\}} \mathbf{1}_{\overline{\pv}_n \models\phi_i} \cdot h(\loc_{i},\overline{\pv}_n)\right]
\]
describes the semantics of branching locations, where $\phi_1 = \phi, \phi_2 =\neg\phi$ and $\loc_1,\loc_2$ are the corresponding successor locations.
Then from properties of conditional expectation, one obtains:
\begin{eqnarray*}
& & \condexpv{X_{n+1}+ C_n}{\mathcal{F}_n} \\
&=& \condexpv{X_{n+1}}{\mathcal{F}_n} + \condexpv{C_n}{\mathcal{F}_n} \\
&=& \mathbf{1}_{\overline{\ell}_n =\lout{}}\cdot X_n +Y'_p+Y'_a+Y_{nd}+Y_t+Y_b\\
&& +\mathbf{1}_{\overline{\ell}_n \in\tlocs{}} \cdot C_n \\
\end{eqnarray*}
where
\[
Y'_p:=\sum_{\loc\in\plocs{}}\left[\mathbf{1}_{\overline{\ell}_n=\loc} \cdot \sum_{i\in\{0,1\}} \probm(B_{\loc}=i) \cdot h(\loc_{B_{\loc}=i},\overline{\pv}_n)\right]
\]
and
\[
Y'_a:=\sum_{\loc\in\alocs{}} \mathbf{1}_{\overline{\ell}_n =\loc} \cdot \expv_{\rv}(h(\loc',F_\ell(\overline{\pv}_n,\rv)))
\]
This follows from the facts that (i) $\mathbf{1}_{\overline{\ell}_n =\lout{}} \cdot X_n$, $Y_{nd}$, $Y_t$, $Y_b$ are measurable in $\mathcal{F}_n$; (ii) $\condexpv{C_n}{\mathcal{F}_n}=\mathbf{1}_{\overline{\ell}_n =\tlocs{}} \cdot C_n$;  (iii) for $Y_p$ and $Y_a$, their conditional expectations are resp. $Y'_p$, $Y'_a$.

From (C3), when $\overline{\ell}_n\in \plocs{} \cup \alocs{} \cup \blocs{}$, we have
$pre_h(\overline{\ell}_n,\overline{\pv}_n) = \mathbf{1}_{\overline{\ell}_n =\lout{}} \cdot X_n +Y'_p + Y'_a + Y_b$.
When $\overline{\ell}_n\in \tlocs{}$, we have $pre_h(\overline{\ell}_n,\overline{\pv}_n) = Y_t +C_n$.
Then we get:
\[
\condexpv{X_{n+1}+ C_n}{\mathcal{F}_n} = pre_h(\overline{\ell}_n,\overline{\pv}_n)\enskip.
\]
When $\overline{\ell}_n\in \Dlocs{}$, we have $\condexpv{X_{n+1}+ C_n}{\mathcal{F}_n} =Y_{nd} \le  pre_h(\overline{\ell}_n,\overline{\pv}_{n})$.
Hence the result follows.
\end{proof}

\subsection{Polynomial Upper Cost Supermartingales (PUCSs)}\label{app:PUPFs}

{\bf Theorem~\ref{thm:lmf1}.}~(Soundness of PUCS)
Consider a nondeterministic probabilistic program $P$, with a linear invariant $I$ and a PUCS $h$. If $P$ satisfies the concentration property and the bounded update property, then $\supval(\pv)\le h(\lin,\pv)$ for all initial valuations $\pv\in I(\lin)$.

\begin{proof}[Proof of Theorem~\ref{thm:lmf1}]
Fix any scheduler $\sigma$ and initial valuation $\pv$ for a nondeterministic probabilistic program $P$.
Let $T=\min\{n\mid \overline{\ell}_n= \lout{}\}$.
By our assumption, $\expv(T)<\infty$ under $\sigma$.
We recall the random variables $C_0,C_1,\dots$ where each $C_n$ represents the cost/reward accumulated during the $n$th execution step of $P$.

We define the stochastic process $\{X_n\}_{n\in\Nset_0}$ by $X_n=h(\overline{\loc}_n,\overline{\nu}_n)$.
Then we define the stochastic process $Y_0,Y_1,\dots$ by:
\[
\textstyle Y_n:=h(\overline{\ell}_{n},\overline{\pv}_n)+\sum_{m=0}^{n-1} C_m\enskip.
\]
Furthermore, we accompany $Y_0,Y_1,\dots$ with the filtration $\mathcal{F}_0,\mathcal{F}_1,\dots$ such that each $\mathcal{F}_n$ is the smallest sigma-algebra that makes all random variables from $\{\overline{\pv}_0,\dots,\overline{\pv}_n \},\{\overline{\rv}_0,\dots,\overline{\rv}_{n} \}$ and $\{\overline{\ell}_0,\dots,\overline{\ell}_{n-1} \}$ measurable.
Then by Lemma~\ref{app:lemm:important}, we have $\condexpv{X_{n+1}+ C_n}{\mathcal{F}_n} \le X_n$.\\
Thus we get:
\begin{eqnarray*}
& &\condexpv{Y_{n+1}}{\mathcal{F}_n}\\
&=& \condexpv{Y_{n}+h(\overline{\ell}_{n+1},\overline{\pv}_{n+1})-h(\overline{\ell}_{n},\overline{\pv}_n)+C_{n}}{\mathcal{F}_n}\\
&=& Y_n+\left(\condexpv{h(\overline{\ell}_{n+1},\overline{\pv}_{n+1})+C_{n}}{\mathcal{F}_n}-h(\overline{\ell}_n,\overline{\pv}_n)\right)\\
&\le& Y_n
\end{eqnarray*}
Hence, $\{Y_n\}_{n\in\Nset_0}$ is a supermartingale. Moreover, we have from the bounded-update property that

\begin{eqnarray*}
|Y_{n+1}-Y_n|&=&|h_{n+1}+\sum_{m=1}^{n} C_m-h_n-\sum_{m=1}^{n-1} C_m| \\
&=& |h_{n+1}-h_n+C_n| \\
&\le& |h_{n+1}-h_n|+|C_n| \\
&\le& M\cdot n^d + c''\cdot n \\
&\le& \overline{M}^d \cdot n
\end{eqnarray*}
for some $\overline{M} >0$. \\
Thus, by applying Optional Stopping Theorem, we obtain immediately that $\expv(Y_T)\le \expv(Y_0)$.
By definition,
\[
\textstyle Y_T=h(\overline{\ell}_{T},\overline{\pv}_T)+\sum_{m=1}^{T-1} C_m=\sum_{m=1}^{T-1} C_m\enskip.
\]
It follows from (C2) that $\expv(C_\infty)=\expv(\sum_{m=1}^{T-1} C_m)\le \expv(Y_0)=h(\lin{}, \pv)$.
Since the scheduler $\sigma$ is chosen arbitrarily, we obtain that $\mbox{\sl supval}(\pv)\le h(\lin{}, \pv)$.
\end{proof}

\subsection{Polynomial Lower Cost Submartingales (PLCSs)}\label{app:PLPFs}

{\bf Theorem~\ref{thm:llmf1}.}~(Soundness of PLCS)
Consider a nondeterministic probabilistic program $P$, with a linear invariant $I$ and a PLCS $h$. If $P$ satisfies the concentration property and the bounded update property, then $\supval(\pv)\ge h(\lin,\pv)$ for all initial valuations $\pv\in  I(\lin)$.

\begin{proof}[Proof of Theorem~\ref{thm:llmf1}]
We follow most definitions above.
Fix any scheduler $\sigma$ and initial valuation $\pv$ for a nondeterministic probabilistic program $P$.
Let $T=\min\{n\mid \overline{\ell}_n= \lout{}\}$.
By our assumption, $\expv(T)<\infty$ under $\sigma$.
We recall the random variables $C_0,C_1,\dots$ where each $C_n$ represents the cost/reward accumulated during the $n$th execution step of $P$.
We define the stochastic process $\{X_n\}_{n\in\Nset_0}$ by $X_n=h(\overline{\loc}_n,\overline{\nu}_n)$.
Then we define the stochastic process $Y_0,Y_1,\dots$ by:
\[
\textstyle Y_n:=h(\overline{\ell}_{n},\overline{\pv}_n)+\sum_{m=0}^{n-1} C_m\enskip.
\]
Furthermore, we accompany $Y_0,Y_1,\dots$ with the filtration $\mathcal{F}_0,\mathcal{F}_1,\dots$ such that each $\mathcal{F}_n$ is the smallest sigma-algebra that makes all random variables from $\{\overline{\pv}_0,\dots,\overline{\pv}_n \}, \{\overline{\rv}_0,\dots,\overline{\rv}_{n-1} \}$ and $\{\overline{\ell}_0,\dots,\overline{\ell}_{n} \}$ measurable.
Then by (C3'), we have $\condexpv{X_{n+1}+ C_n}{\mathcal{F}_n} \ge h(\overline{\loc}_n,\pv_n)$.\\
Thus we get:
\begin{eqnarray*}
& &\condexpv{Y_{n+1}}{\mathcal{F}_n}\\
&=& \condexpv{Y_{n}+h(\overline{\ell}_{n+1},\overline{\pv}_{n+1})-h(\overline{\ell}_{n},\overline{\pv}_n)+C_{n}}{\mathcal{F}_n}\\
&=& Y_n+\left(\condexpv{h(\overline{\ell}_{n+1},\overline{\pv}_{n+1})+C_{n}}{\mathcal{F}_n}-h(\overline{\ell}_n,\overline{\pv}_n)\right)\\
&\ge& Y_n
\end{eqnarray*}
Hence, $\{Y_n\}_{n\in\Nset_0}$ is a submartingale, so $\{-Y_n\}_{n\in\Nset_0}$ is a supermartingale. Moreover, we have from the bounded update property that 
\begin{eqnarray*}
|-Y_{n+1}-(-Y_n)|&=& |Y_{n+1}-Y_n| \\
&=&|h_{n+1}+\sum_{m=1}^{n} C_m-h_n-\sum_{m=1}^{n-1} C_m| \\
&=& |h_{n+1}-h_n+C_n| \\
&\le& |h_{n+1}-h_n|+|C_n| \\
&\le& M\cdot n^d + c''\cdot n \\
&\le& \overline{M}^d \cdot n
\end{eqnarray*}
for some $\overline{M} >0$. \\
Thus, by applying Optional Stopping Theorem, we obtain immediately that $\expv(-Y_T)\le \expv(-Y_0)$, so $\expv(Y_T)\ge \expv(Y_0)$.
By definition,
\[
\textstyle -Y_T= -h(\overline{\ell}_{T},\overline{\pv}_T)-\sum_{m=1}^{T-1} C_m=-\sum_{m=1}^{T-1} C_m\enskip.
\]
It follows from (C2) that $\expv(C_\infty)=\expv(\sum_{m=1}^{T-1} C_m)\ge \expv(Y_0)=h(\pv)$.
Since the scheduler $\sigma$ is chosen arbitrarily, we obtain that $\mbox{\sl supval}(\pv)\ge h(\pv)$.

\end{proof}

\subsection{Unbounded Nonnegative Costs and General Updates}\label{app:PUPFs2}

{\bf Theorem~\ref{thm:lmf2}.}~(Soundness of nonnegative PUCS)
Consider a nondeterministic probabilistic program $P$, with a linear invariant $I$ and a nonnegative PUCS $h$. If all the step-wise costs in $P$ are always nonnegative, then $\supval(\pv)\le h(\lin,\pv)$ for all initial valuations $\pv\in I(\lin)$.

\begin{proof}[Proof of Theorem~\ref{thm:lmf2}]
We also follow most definitions above.
Fix any scheduler $\sigma$ and initial valuation $\pv$ for a nondeterministic probabilistic program $P$.
Let $T=\min\{n\mid \overline{\ell}_n= \lout{}\}$.
By our assumption, $\expv(T)<\infty$ under $\sigma$.
We recall the random variables $C_0,C_1,\dots$ where each $C_n$ represents the cost/reward accumulated during  the $n$th execution step of $P$.
Then we define the stochastic process $X_0,X_1,\dots$ by:
\[
\textstyle X_n:=h(\overline{\ell}_{n},\overline{\pv}_n)\enskip.
\]
Furthermore, we accompany $X_0,X_1,\dots$ with the filtration $\mathcal{F}_0,\mathcal{F}_1,\dots$ such that each $\mathcal{F}_n$ is the smallest sigma-algebra that makes all random variables from $\{\overline{\pv}_0,\dots,\overline{\pv}_n \},\{\overline{\rv}_0,\dots,\overline{\rv}_{n-1} \}$ and $\{\overline{\ell}_0,\dots,\overline{\ell}_{n-1} \}$ measurable.
Then by C3, we have $$\condexpv{X_{n+1}+ C_n}{\mathcal{F}_n} \le X_{n}.$$
Thus we get:
\begin{eqnarray*}
& &\condexpv{X_{n+1}+C_n}{\mathcal{F}_n}\le X_n\\
\Leftrightarrow& & \expv(\condexpv{X_{n+1}+C_n}{\mathcal{F}_n})\le \expv(X_n)\\
\Leftrightarrow& & \expv(X_{n+1}+C_n)\le \expv(X_n) \\
\Leftrightarrow& & \expv(X_{n+1})+\expv(C_n)\le \expv(X_n) \\
& &\mbox{(By Induction)} \\
\Leftrightarrow& & \expv(X_{n+1})+\sum\limits_{m=0}^{n}\expv(C_m)\le \expv(X_0) \\
& &\mbox{(By C2)} \\
\Leftrightarrow& & \sum\limits_{m=0}^{n}\expv(C_m)\le \expv(X_0) \\
& &\mbox{(By Monotone Convergence Theorem)} \\
\Leftrightarrow& & \expv(\sum\limits_{m=0}^{n} C_m)\le \expv(X_0) \\
\end{eqnarray*}
The Induction is:
\begin{eqnarray*}
& &\expv(X_n)+\expv(C_{n-1})\le \expv(X_{n-1}) \\
\Leftrightarrow& & \expv(X_n)\le \expv(X_{n-1})-\expv(C_{n-1}) \\
& &\expv(X_{n-1})+\expv(C_{n-2})\le \expv(X_{n-2}) \\
\Leftrightarrow& & \expv(X_{n-1})\le \expv(X_{n-2})-\expv(C_{n-2}) \\
\end{eqnarray*}
Then
\begin{eqnarray*}
& & \expv(X_n)\le \expv(X_{n-1})-\expv(C_{n-1}) \\
\Leftrightarrow& & \expv(X_n)\le \expv(X_{0})-\sum\limits_{m=0}^{n-1}\expv(C_{m})
\end{eqnarray*}
Because all the PUCSs are nonnegative, we can get $\expv(X_{n+1})\ge 0$.
When $n\to \infty$, we obtain
\[
\textstyle \expv(\sum\limits_{m=0}^{\infty} C_m)\le \expv(X_0)\enskip.
\]
Since the scheduler $\sigma$ is chosen arbitrarily, we obtain that $\mbox{\sl supval}(\pv)\le h(\pv)$.

\end{proof}

\section{Details of Example~\ref{al:simple}}\label{app:detail1}

Since our algorithm is technical, we will illustrate the computational steps of the our algorithms on
the example in Figure~\ref{fig:example}.

\begin{example}[Illustration of our algorithms]\label{de:simple}
 We consider the example in Figure~\ref{fig:example}, and assign the invariant $I$ as in Figure~\ref{fig:ex2}.

Firstly, the algorithm sets up a quadratic template $h$ for a PUCS by setting $h(\loc_n,x,y):=a_{n1}\cdot x^2+a_{n2}\cdot xy+a_{n3}\cdot x+a_{n4}\cdot y^2+a_{n5}\cdot y+a_{n6}$ for each $\loc_n(n=1,\dots,4)$ and $h(\loc_5,x,y)=0$ because $\loc_5=\lout{}$, where $a_{np}$ are scalar variables for $n=1,\dots,4$ and $p=1,\dots,6$.

Next we compute the pre-expectations of this example.
\begin{itemize}
\item $\loc_1\in\blocs{}$,\\
      \begin{eqnarray*}
        pre_h(\loc_1,x,y)&=& \mathbf{1}_{\loc'=\loc_2}\cdot h(\loc_2,x,y)+\mathbf{1}_{\loc'=\loc_5}\cdot h(\loc_5,x,y) \\
            &=&\mathbf{1}_{\loc'=\loc_2}\cdot (a_{21}\cdot x^2 +a_{22} \cdot xy +a_{23}\cdot x\\
            &&+a_{24}\cdot y^2+a_{25}\cdot y+a_{26})
            +\mathbf{1}_{\loc'=\loc_5}\cdot 0\\
      \end{eqnarray*}
\item $\loc_2\in\alocs{}$,\\
       \begin{eqnarray*}
          pre_h(\loc_2,x,y) &=& \expv_{R}[h(\loc_3,x+r,y)] \\
          &=& \expv_{R}[a_{31}\cdot (x+r)^2 +a_{32} \cdot (x+r)y \\&&+a_{33}\cdot (x+r)+a_{34}\cdot y^2+a_{35}\cdot y+a_{36}] \\
          &=& a_{31}\cdot x^2+a_{32}\cdot xy+(a_{33}-a_{31})\cdot x\\
          &&+a_{34}\cdot y^2+(a_{35}-\frac{1}{2}a_{32})\cdot y +a_{31}\\&&-\frac{1}{2}a_{33}+a_{36}\\
      \end{eqnarray*}
\item $\loc_3\in\alocs{}$,\\
      \begin{eqnarray*}
        pre_h(\loc_3,x,y) &=& \expv_{R}[h(\loc_4,x,r')]\\
          &=& \expv_{R}[a_{41}\cdot x^2 +a_{42} \cdot x\cdot r'+a_{43}\cdot x\\&& +a_{44}\cdot r'^2+a_{45}\cdot r'+a_{46} ] \\
          &=& a_{41}\cdot x^2+(\frac{1}{3} a_{42}+a_{43})\cdot x+a_{44}\\&&+\frac{1}{3} a_{45}+a_{46} \\
      \end{eqnarray*}
\item $\loc_4\in\tlocs{}$, \\
      \begin{eqnarray*}
        pre_h(\loc_4,x,y) &=& h(\loc_1,x,y)+\expv_{R}(x\cdot y) \\
          &=& a_{11}\cdot x^2 +a_{12} \cdot xy +a_{13}\cdot x+a_{14}\cdot y^2\\&&+a_{15}\cdot y+a_{16}+ xy\\
          &=& a_{11}\cdot x^2 +(a_{12}+1) \cdot xy +a_{13}\cdot x\\&&+a_{14}\cdot y^2+a_{15}\cdot y+a_{16} \\
      \end{eqnarray*}
\end{itemize}

Let the maximal number of multiplicands $t$ in $Monoid(\Gamma)$ be 2, the form of Eq. ($\sharp$) is as following:
\begin{itemize}
\item(label 1)
   (1)$\loc'=\loc_2$ \\
   \begin{eqnarray*}
   \Gamma&=&\{x,x-1\}\\
    u_1&=&1,u_2=x,u_3=x-1,u_4=x^2-x,\\
    u_5&=&x^2,u_6=x^2-2x+1;\\
    g(\mathbf{x})&=&b_1+b_2x+b_3(x-1)+b_4(x^2-x)+b_5x^2\\
    & &+b_6(x^2-2x+1) \\
    &=&(b_4+b_5+b_6)x^2+(b_2+b_3-b_4-2b_6)x\\
    & &+b_1-b_3+b_6 \\
   \end{eqnarray*}
   (2)$\loc'=\loc_5$ \\
   \begin{eqnarray*}
   \Gamma&=&\{x,1-x\}\\
   u_1&=&1,u_2=x,u_3=1-x,u_4=x-x^2,\\
   u_5&=&x^2,u_6=1-2x+x^2;\\
   g(\mathbf{x})&=&b_7+b_8x+b_9(1-x)+b_{10}(x-x^2)+b_{11}x^2\\&&+b_{12}(1-2x+x^2) \\
   &=& (b_{11}+b_{12}-b_{10})x^2+(b_8-b_9+b_{10}-2b_{12})x\\&&+b_7+b_{12} \\
   \end{eqnarray*}
   for $b_i\ge 0, i=1,\dots,12$.
\item(label 2)
   \begin{eqnarray*}
   \Gamma&=&\{x-1\} \\
   u_1&=&1,u_2=x-1,u_3=x^2-2x+1;\\
   g(\mathbf{x})&=&c_1+c_2(x-1)+c_3(x^2-2x+1) \\
   &=& c_3x^2+(c_2-2c_3)x+c_1-c_2+c_3 \\
   \end{eqnarray*}
   for $c_j\ge 0, j=1,2,3$.
\item(label 3)
   \begin{eqnarray*}
   \Gamma&=&\{x\} \\
   u_1&=&1,u_2=x,u_3=x^2;\\
   g(\mathbf{x})&=&d_1+d_2x+d_3x^2 \\
   &=& d_3x^2+d_2x+d_1 \\
   \end{eqnarray*}
   for $d_l\ge 0, l=1,2,3$.
\item(label 4)
   \begin{eqnarray*}
   \Gamma&=&\{x,1-y,1+y\} \\
   u_1&=&1,u_2=x,u_3=1-y,u_4=1+y,\\
   u_5&=&x(1-y),u_6=x(1+y),u_7=(1-y)(1+y),\\
   u_8&=&x^2,u_9=(1-y)^2,u_{10}=(1+y)^2;\\
   g(\mathbf{x})&=&e_1+e_2x+e_3(1-y)+e_4(1+y)+e_5x(1-y)\\&&+e_6x(1+y)+e_7(1-y)(1+y)+e_8x^2\\&&+e_9(1-y)^2+e_{10}(1+y)^2 \\
   &=& e_8x^2+(e_6-e_5)xy+(e_2+e_5+e_6)x\\&&+(e_9+e_{10}-e_7)y^2+(e_4-e_3-2e_9+2e_{10})y\\&&+e_1+e_3+e_4+e_7+e_9+e_{10} \\
   \end{eqnarray*}
   for $e_m\ge 0, m=1,\dots,10$.
\end{itemize}

Then we extract instances conforming to pattern $g=h(\loc,\nu)-\pre_{h}(\loc,\nu)$ from C3.
\begin{itemize}
\item(C3, label 1) \\
     (1)$\loc'=\loc_2$ \\
     \begin{eqnarray*}
     g(\mathbf{x})&=&h(\loc_1,x,y)-pre_h(\loc_1,x,y) \\
     &=&h(\loc_1,x,y)-h(\loc_2,x,y) \\
     &=&(a_{11}-a_{21})x^2+(a_{12}-a_{22})xy+(a_{13}-a_{23})x\\
     &&+(a_{14}-a_{24})y^2+(a_{15}-a_{25})y+a_{16}-a_{26} \\
     \end{eqnarray*}
     (2)$\loc'=\loc_5$ \\
     \begin{eqnarray*}
     g(\mathbf{x})&=&h(\loc_1,x,y)-pre_h(\loc_1,x,y) \\
     &=&h(\loc_1,x,y)-h(\loc_5,x,y) \\
     &=&a_{11}\cdot x^2 +a_{12} \cdot xy +a_{13}\cdot x+a_{14}\cdot y^2\\&&+a_{15}\cdot y+a_{16} \\
     \end{eqnarray*}
\item(C3, label 2)
     \begin{eqnarray*}
     g(\mathbf{x})&=&h(\loc_2,x,y)- pre_h(\loc_2,x,y) \\
     &=&(a_{21}-a_{31})x^2+(a_{22}-a_{32})xy \\&&+(a_{23}-a_{33}+a_{31})x+(a_{24}-a_{34})y^2\\&&+(a_{25}-a_{35}+\frac{1}{2}a_{32})y+a_{26}-a_{31}+\frac{1}{2}a_{33}-a_{36} \\
     \end{eqnarray*}
\item(C3, label 3)
     \begin{eqnarray*}
     g(\mathbf{x})&=&h(\loc_3,x,y)- pre_h(\loc_3,x,y) \\
     &=&(a_{31}-a_{41})x^2+a_{32}xy+(a_{33}-\frac{1}{3}a_{42}-a_{43})x\\&&+a_{34}y^2+a_{35}y+a_{36}-a_{44}-\frac{1}{3}a_{45}-a_{46}\\
     \end{eqnarray*}
\item(C3, label 4)
     \begin{eqnarray*}
     g(\mathbf{x})&=&h(\loc_4,x,y)- pre_h(\loc_4,x,y) \\
     &=&(a_{41}-a_{11})x^2+(a_{42}-a_{12}-1)xy+(a_{43}-a_{13})x\\&&+(a_{44}-a_{14})y^2+(a_{45}-a_{15})y+a_{46}-a_{16}\\
     \end{eqnarray*}
\end{itemize}

So we can translate them into systems of linear equalities.\\
(I) For label 1,
\[	
\begin{cases}
a_{11}-a_{21}=b_4+b_5+b_6 \\
a_{12}-a_{22}=0     \\
a_{13}-a_{23}=b_2+b_3-b_4-2b_6  \\
a_{14}-a_{24}=0 \\
a_{15}-a_{25}=0 \\
a_{16}-a_{26}=b_1-b_3+b_6 \\
\end{cases}
\]
and
\[
\begin{cases}
a_{11}=b_{11}+b_{12}-b_{10} \\
a_{12}=0\\
a_{13}=b_8-b_9+b_{10}-2b_{12} \\
a_{14}=0 \\
a_{15}=0 \\
a_{16}=b_7+b_{12} \\
\end{cases}
\]

(II) For label 2,
\[	
\begin{cases}
a_{21}-a_{31}=c_3 \\
a_{22}-a_{32}=0 \\
a_{23}-a_{33}+a_{31}=c_2-2c_3 \\
a_{24}-a_{34}=0 \\
a_{25}-a_{35}+\frac{1}{2}a_{32}=0 \\
a_{26}-a_{31}+\frac{1}{2}a_{33}-a_{36}=c_1-c_2+c_3 \\
\end{cases}
\]

(III) For label 3,
\[
\begin{cases}
a_{31}-a_{41}=d_3 \\
a_{32}=0 \\
a_{33}-\frac{1}{3}a_{42}-a_{43}=d_2   \\
a_{34}=0 \\
a_{35}=0 \\
a_{36}-a_{44}-\frac{1}{3}a_{45}-a_{46}=d_1    \\
\end{cases}
\]

(IV) For label 4,
\[
\begin{cases}
a_{41}-a_{11}=e_8 \\
a_{42}-a_{12}-1=e_6-e_5 \\
a_{43}-a_{13}=e_2+e_5+e_6  \\
a_{44}-a_{14}=e_9+e_{10}-e_7 \\
a_{45}-a_{15}=e_4-e_3-2e_9+2e_{10} \\
a_{46}-a_{16}=e_1+e_3+e_4+e_7+e_9+e_{10}    \\
\end{cases}
\]

Our target function is $h(\loc_1,x_0,y_0)$, where $x_0,y_0$ are the initial inputs and we fix $x_0$ to be a proper large integer, i.e. $x_0=100$, and $y_0$ to be $0$.
\begin{eqnarray*}
min\  a_{11}x_0^2+a_{13}x_0+a_{16} \\
subject\  to\  (I),(II),(III),(IV) \\
b_i,c_j,d_l,e_m\ge 0, \forall i,j,l,m
\end{eqnarray*}

Finally, the algorithm gives the optimal solutions through linear programming such that:
\begin{eqnarray*}
h(\loc_1,x,y)&=&\frac{1}{3}\cdot x^2 +\frac{1}{3} \cdot x \\
h(\loc_2,x,y)&=&\frac{1}{3}\cdot x^2 +\frac{1}{3} \cdot x \\
h(\loc_3,x,y)&=&\frac{1}{3}\cdot x^2 +\frac{2}{3} \cdot x \\
h(\loc_4,x,y)&=& \frac{1}{3}\cdot x^2 +xy+\frac{1}{3} \cdot x \\
\end{eqnarray*}

To find a PLCS for this example, the steps are similar.
The algorithm sets up a quadratic template $h'$ for a PLCS with the similar form of the above PUCS $h$.
By the same way, we get the optimal solutions of the template $h'$ and find they are the same as the PUCS's.

By the definition of PUCS and PLCS (see Section~\ref{sec:approach}), we can conclude that this template $h$ is both PUCS and PLCS, and we can get the accurate value of expected resource consumption that $\expv(C_\infty)= h(\loc_1,x_0,y_0)=\frac{1}{3}x_0^2+\frac{1}{3}x_0$.
\end{example}

\section{Experimental Results}\label{app:results}

\subsection{Benchmarks}
We use ten example programs for our experimental results, including (1)~Bitcoin Mining (see Figure~\ref{fig:mining}); (2)~Bitcoin Mining Pool (see Figure~\ref{fig:pool}); (3)~Queuing Network (see Figure~\ref{fig:FJcode}); (4)~Species Fight (see Figure~\ref{fig:species}); (5)~Simple Loop (see Figure~\ref{fig:example}); (6)~Nested Loop (see Figure~\ref{fig:complexdisplay}); (7)~Random Walk (see Figure~\ref{fig:rdwalk}); (8)~2D Robot (see Figure~\ref{fig:robot}); (9)~Goods Discount (see Figure~\ref{fig:goods}); and (10)~Pollutant Disposal (see Figure~\ref{fig:pollutant}).

We now provide a brief introduction about 2D Robot, Goods discount, and Pollutant Disposal.

\smallskip\noindent{\bf{2D Robot.}}
We consider a robot walking in a plane. Suppose that the robot is initially located below the line $y=x$ and we want it to cross this line, i.e.~the program continues until the robot crosses the line. At each iteration, the robot probabilistically chooses one of the following 9 directions and moves in that direction: \{0: North, 1: South, 2: East, 3: West, 4: Northeast, 5: Southeast, 6: Northwest, 7: Southwest, 8: Stay\}. Moreover, the robot's step size is a uniformly random variable between 1 and 3. At the end of each iteration, a cost is incurred, which is dependent on the distance between the robot and the line $y=x$.

\smallskip\noindent{\bf{Goods Discount.}}
Consider a shop that sells a specific type of perishable goods with an expiration date. After a certain number of days ($30$ days in our example), when the expiration date is close, the goods have to be sold at a discount, which will cause losses. Moreover, stocking goods takes up space, which also incurs costs. On the other hand, selling goods leads to a reward. In this example, we model this scenario as follows: $n$ is the number of goods which are on sale, $d$ is the number of days after the goods are manufactured and each time one piece of goods is sold, $d$ will be incremented by a random variable $r$ which has a uniform distribution $[1,2]$ (This models the time it takes to sell the next piece). The program starts with the initial value $n=a$ , $d=b$ and terminates if $d$ exceeds $30$ days, which will eventually happen with probability $1$.

\smallskip\noindent{\bf{Pollutant Disposal.}}
We consider a pollutant disposal factory that has two machines $A$ and $B$. At first, the factory is given an initial amount of pollutants to dispose of. At each iteration, the factory uses machine $A$ with probability $0.6$ and machine $B$ with probability $0.4$. Machine $A$ can dispose of $r_1$ units of pollutants, while creating $r'_1$ new units of pollutants in the process. Similarly, machine $B$ can dispose of $r_2$ units by creating $r'_2$ new units of pollutants. The sampling variables $r_1, r_2$ are \emph{integer-valued} random variables which have an equivalent sampling rate between $1$ and $10$. Similarly, $r'_1, r'_2$ are \emph{integer-valued} random variables which have an equivalent sampling rate between $2$ and $8$. There is a reward associated with disposing of each unit of pollutants. On the other hand, at the end of each iteration, a cost is incurred which is proportional to the amount of remaining pollutants.

\lstset{language=prog}
\lstset{tabsize=3}
\newsavebox{\complexdisplay}
\begin{lrbox}{\complexdisplay}
\begin{lstlisting}[mathescape]
$1$: while $i\ge 1$ do
$2$:    $x:=i$;
$3$:    while $x\ge 1$ do
$4$:        $x:=x+r$;
$5$:         $y:=r'$;
$6$:        $\mathbf{tick}(y)$
        od
$7$:   $i:=i+r''$;
$8$    $z:=r'''$;
$9$:   $\mathbf{tick}(-z * i)$
     od
$10$:
\end{lstlisting}
\end{lrbox}

\begin{figure}[H]
\centering
  \begin{minipage}[c]{0.45\textwidth}
      \centering
      \usebox{\complexdisplay}
      \caption{A Nested Loop Example}
      \label{fig:complexdisplay}
  \end{minipage}

\end{figure}

\lstset{language=prog}
\lstset{tabsize=3}
\newsavebox{\rdwalk}
\begin{lrbox}{\rdwalk}
\begin{lstlisting}[mathescape]
$\probm(r=1)=0.25,\probm(r=-1)=0.75$
while $x\le n$ do
   if prob(0.6) then
       x:=x+1
   else
       x:=x-1
   fi;
   $y=r$;
   tick$(y)$
od
\end{lstlisting}
\end{lrbox}

\begin{figure}[H]
  \centering
  \usebox{\rdwalk}
  \caption{rdwalk}
  \label{fig:rdwalk}
\end{figure}

\lstset{language=prog}
\lstset{tabsize=3}
\newsavebox{\robot}
\begin{lrbox}{\robot}
\begin{lstlisting}[mathescape]
 x:=a;y:=b;
 while $y\le x$ do
    if prob$(0.2)$ then
       $y:=y+r_0$
    else if prob$(0.125)$ then
            $y:=y-r_1$
    else if prob$(0.143)$ then
            $x:=x+r_2$
    else if prob$(0.167)$ then
            $x:=x-r_3$
    else if prob$(0.2)$ then
            $x:=x+r_4$;
            $y:=y+r'_4$
    else if prob$(0.25)$then
            $x:=x+r_5$;
            $y:=y-r'_5$
    else if prob$(0.333)$ then
            $x:=x-r_6$;
            $y:=y+r'_6$
    else if prob$(0.5)$ then
            $x:=x-r_7$;
            $y:=y-r'_7$
         else skip
    fi fi fi fi fi fi fi fi;
    tick$(0.707*(x-y))$
 od
\end{lstlisting}
\end{lrbox}

\begin{figure}[H]
\centering
\usebox{\robot}
\caption{2D Robot}
\label{fig:robot}
\end{figure}

\lstset{language=prog}
\lstset{tabsize=3}
\newsavebox{\goods}
\begin{lrbox}{\goods}
\begin{lstlisting}[mathescape]
 n:=a;d:=1;
 while $d\le 30$ and $n>=1$ do
     $n:=n-1$;
     tick$(5)$;
     $d:=d+r$;
     tick$(-0.01* n)$
 od;
 tick$(-0.5* n)$
\end{lstlisting}
\end{lrbox}

\begin{figure}[H]
\centering
\usebox{\goods}
\caption{Goods Discount}
\label{fig:goods}
\end{figure}

\lstset{language=prog}
\lstset{tabsize=3}
\newsavebox{\pollutant}
\begin{lrbox}{\pollutant}
\begin{lstlisting}[mathescape]
 n:=a;
 while $n\ge 10$ do
    if prob$(0.6)$ then
       $x:=r_1$;
       $n:=n-x+r'_1$;
       tick$(5* x)$
    else
       $y:=r_2$;
       $n:=n-y+r'_2$;
       tick$(5* y)$
    fi;
    tick$(-0.2* n)$
 od
\end{lstlisting}
\end{lrbox}

\begin{figure}[H]
\centering
\usebox{\pollutant}
\caption{Pollutant Disposal}
\label{fig:pollutant}
\end{figure}

\subsection{Numeric Bounds and Plots}

Table~\ref{tab:expr} shows the numeric upper and lower bounds obtained for each benchmark over several initial valuations. In each case, we report the upper bound obtained through PUCS, the runtime of our PUCS synthesis algorithm, the lower bound obtained through PLCS, and the runtime of our PLCS synthesis algorithm. Moreover, we simulated $1000$ runs of each program with each initial value, computed the resulting costs, and reported the mean $\mu$ and standard deviation $\sigma$ of the costs.
Note that we do not have simulation results for Bitcoin mining examples as they involve nondeterminism.
Also we do not have lower bounds for species fight example as its updates are unbounded. 

In another experiment, we set the time limit for simulations to the time it takes for our approach to synthesize a PUCS/PLCS. We also replaced the nondeterministic $\textbf{if}~\star$ statements with probabilistic $\textbf{if prob(}0.5\textbf{)}$. The results are reported in Table~\ref{tab:exprnew}.

\begin{table*}[hbtp]
	\vspace{5cm}
	\caption{Experimental Results. All times are reported in seconds.}
	\resizebox{15cm}{!}{
		\begin{tabular}{c|c|cc|cc|cc}
			\toprule
			
			\textbf{Benchmark} & $\mathbf{\pv_0}$ & \multicolumn{2}{c|}{\textbf{PUCS}} & \multicolumn{2}{c|}{\textbf{PLCS}} & \multicolumn{2}{c}{\textbf{Simulation}}\\
			\textbf{Program} & & $h(\lin, \pv_0)$ & $T$ & $h(\lin, \pv_0)$ & $T$ & $\mu$ & $\sigma$\\
			\hline
			
			Bitcoin Mining & $x_0 = 20$ & $-28.03$ & $4.69$ & $-30.00$ & $4.73$ &--  & -- \\
			(Figure~\ref{fig:mining}) & $x_0 = 50$ & $-72.28$ & $4.66$ & $-75.00$ & $4.63$ & --  & --  \\
			& $x_0 = 100$ & $-146.03$ & $4.62$ & $-150.00$ & $4.62$ & -- & --\\
			\hline

			Bitcoin Mining Pool & $y_0 = 20$ & $-3.73\times 10^3$ & $14.03$ & $-4.35\times 10^3$ & $13.73$ & -- & --\\
			(Figure~\ref{fig:pool}) & $y_0 = 50$ & $-2.05 \times 10^4$ & $13.78$ & $-2.21 \times 10^4$ & $13.76$ & --& --\\
			& $y_0 = 100$ & $-7.79 \times 10^4$ & $13.96$ & $-8.18 \times 10^4$ & $13.85$ & -- &-- \\
			\hline

			Queuing Network & 				$n_0 = 240$ & $11.82$ & $141.28$ & $9.23$ & $141.32$ & $9.90$ & $4.43$\\
			(Figure~\ref{fig:FJcode})&	  	$n_0 = 280$ & $13.79$ & $142.16$ & $10.76$ & $140.70$ & $11.15$ & $4.66$ \\
			&								$n_0 = 320$ & $15.76$ & $141.02$ & $12.30$ & $141.42$ & $12.99$ & $5.29$ \\
			\hline

			Species Fight 					& $a_0 = 12, b_0 = 10$ & $1.65 \times 10^3$ & $16.43$ & -- & -- & $817.40$ & $379.28$\\
			(Figure~\ref{fig:species})		& $a_0 = 14, b_0 = 10$ & $2.09 \times 10^3$ & $16.47$ & -- & -- & $971.86$ & $453.89$\\
			& $a_0 = 16, b_0 = 10$ & $2.53 \times 10^3$ & $16.30$ & -- & -- & $1.13 \times 10^3$ & $0.55 \times 10^3$\\
			\hline

			Figure~\ref{fig:example} & 	$x_0 = 100$ & $3.37 \times 10^3$ & $3.05$ & $3.37 \times 10^3$ & $3.03$ & $3.41 \times 10^3$ & $0.90 \times 10^3$\\
			&							$x_0 = 160$ & $8.59 \times 10^3$ & $3.00$ & $8.59 \times 10^3$ & $3.02$ & $8.62 \times 10^3$ & $1.76 \times 10^3$\\
			&							$x_0 = 200$ & $1.34 \times 10^4$ & $3.00$ & $1.34 \times 10^4$ & $3.00$ & $1.35 \times 10^4$ & $0.25 \times 10^4$\\
			\hline

			Nested Loop  &  $i_0 = 50$ & $883.33$ & $15.82$ & $816.67$ & $15.91$ & $872.78$ & $344.29$\\
			& 	            $i_0 = 100$ & $3.43 \times 10^3$ & $16.13$ & $3.30 \times 10^3$ & $15.89$ & $3.43 \times 10^3$ & $0.90 \times 10^3$\\
			&				$i_0 = 150$ & $7.65 \times 10^3$ & $15.80$ & $7.45 \times 10^3$ & $15.93$ & $7.66 \times 10^3$ & $1.68 \times 10^3$\\
			\hline

			Random Walk & $x_0 = 4, n_0 = 20$ & $-40.00$ & $7.00$ & $-42.50$ & $7.07$ & $-42.77$ & $23.46$\\
			& $x_0 = 8, n_0 = 20$ & $-30.00$ & $6.96$ & $-32.50$ & $6.96$ & $-32.32$ & $21.27$\\
			& $x_0 = 12, n_0 = 20$ & $-20.00$ & $7.09$ & $-22.50$ & $7.93$ & $-23.23$ & $18.47$\\
			\hline

			2D Robot & 	$x_0 = 100, y_0 = 40$ & $8.23 \times 10^3$ & $20.11$ & $8.11 \times 10^3$ & $20.03$ & $7.96 \times 10^3$ & $5.83 \times 10^3$\\
			&			$x_0 = 100, y_0 = 60$ & $4.15 \times 10^3$ & $20.16$ & $4.02 \times 10^3$ & $20.13$ & $4.01 \times 10^3$ & $3.64 \times 10^3$\\
			&			$x_0 = 100, y_0 = 80$ & $1.45 \times 10^3$ & $20.15$ & $1.32 \times 10^3$ & $20.13$ & $1.36 \times 10^3$ & $2.00 \times 10^3$\\
			\hline

			Goods Discount & 	$n_0 = 100, d_0 = 1$ & $46.30$ & $8.42$ & $37.89$ & $8.45$ & $41.45$ & $4.22$\\
			&					$n_0 = 150, d_0 = 1$ & $11.63$ & $8.43$ & $2.56$ & $8.43$ & $6.33$ & $3.84$\\
			&					$n_0 = 200, d_0 = 1$ & $-23.02$ & $8.46$ & $-32.77$ & $8.43$ & $-28.26$ & $3.34$\\
			\hline

			Pollutant Disposal &	$n_0 = 50$ & $2.01 \times 10^3$ & $10.04$ & $1.53 \times 10^3$ & $9.85$ & $1.66 \times 10^3$ & $1.02 \times 10^3$\\
			&						$n_0 = 80$ & $2.74 \times 10^3$ & $9.78$ & $2.25 \times 10^3$ & $9.88$ & $2.42 \times 10^3$ & $1.13 \times 10^3$\\
			&						$n_0 = 200$ & $2.04 \times 10^3$ & $9.75$ & $1.56 \times 10^3$ & $9.78$ & $1.66 \times 10^3$ & $1.56 \times 10^3$\\

			\bottomrule
		\end{tabular}
	}
	\vspace{5cm}
	\label{tab:expr}
\end{table*}

\begin{table*}[hbtp]
	\vspace{5cm}
	\caption{Experimental Results on Programs in which Nondeterminism is Replaced with Probability.}
	\resizebox{15cm}{!}{
		\begin{tabular}{c|c|cc|cc|cc}
			\toprule
			
			\textbf{Benchmark} & $\mathbf{\pv_0}$ & \multicolumn{2}{c|}{\textbf{PUCS}} & \multicolumn{2}{c|}{\textbf{PLCS}} & \multicolumn{2}{c}{\textbf{Simulation}}\\
			\textbf{Program} & & $h(\lin, \pv_0)$ & $T$ & $h(\lin, \pv_0)$ & $T$ & $\mu$ & $\sigma$\\
			\hline
			
			Modified Bitcoin Mining & $x_0 = 20$ & $-28.26$ & $4.41$ & $-29.75$ & $4.24$ &$-30.04$  &500.75 \\
			& $x_0 = 50$ & $-72.89$ & $4.52$ & $-74.38$ & $4.33$ &$-73.90$ &$788.04$  \\
			& $x_0 = 100$ & $-147.26$ & $6.26$ & $-148.75$ & $6.12$ &$-145.40$ & $1.11\times 10^3$\\
			\hline

			Modified Bitcoin Mining Pool & $y_0 = 20$ & $-3.77\times 10^3$ & $14.62$ & $-4.31\times 10^3$ & $14.73$ & $-4.26\times 10^3$ &$7.00\times 10^3$\\
			& $y_0 = 50$ & $-2.06 \times 10^4$ & $14.88$ & $-2.19 \times 10^4$ & $14.76$ & $-2.27\times 10^4$& $1.86\times 10^4$\\
			& $y_0 = 100$ & $-7.85 \times 10^4$ & $14.93$ & $-8.11 \times 10^4$ & $14.81$ & $-7.92\times 10^4$ &$3.93\times 10^4$ \\
			\hline

			Queuing Network & 				$n_0 = 240$ & $11.82$ & $141.28$ & $9.23$ & $141.32$ & $10.64$ & $4.94$\\
			(Figure 6)&	  	$n_0 = 280$ & $13.79$ & $142.16$ & $10.76$ & $140.70$ & $12.41$ & $5.34$ \\
			&								$n_0 = 320$ & $15.76$ & $141.02$ & $12.30$ & $141.42$ & $14.17$ & $5.69$ \\
			\hline

			Species Fight 					& $a_0 = 12, b_0 = 10$ & $1.65 \times 10^3$ & $16.43$ & -- & -- & $808.20$ & $379.91$\\
			(Figure 8)		& $a_0 = 14, b_0 = 10$ & $2.09 \times 10^3$ & $16.47$ & -- & -- & $978.48$ & $458.93$\\
			& $a_0 = 16, b_0 = 10$ & $2.53 \times 10^3$ & $16.30$ & -- & -- & $1.13 \times 10^3$ & $0.54 \times 10^3$\\
			\hline

			Figure 2 & 	$x_0 = 100$ & $3.37 \times 10^3$ & $3.05$ & $3.37 \times 10^3$ & $3.03$ & $3.36 \times 10^3$ & $0.92 \times 10^3$\\
			&							$x_0 = 160$ & $8.59 \times 10^3$ & $3.00$ & $8.59 \times 10^3$ & $3.02$ & $8.58 \times 10^3$ & $1.84 \times 10^3$\\
			&							$x_0 = 200$ & $1.34 \times 10^4$ & $3.00$ & $1.34 \times 10^4$ & $3.00$ & $1.34 \times 10^4$ & $0.26 \times 10^4$\\
			\hline

			Nested Loop  &  $i_0 = 50$ & $883.33$ & $15.82$ & $816.67$ & $15.91$ & $880.12$ & $329.54$\\
			& 	            $i_0 = 100$ & $3.43 \times 10^3$ & $16.13$ & $3.30 \times 10^3$ & $15.89$ & $3.40 \times 10^3$ & $0.92 \times 10^3$\\
			&				$i_0 = 150$ & $7.65 \times 10^3$ & $15.80$ & $7.45 \times 10^3$ & $15.93$ & $7.67 \times 10^3$ & $1.61 \times 10^3$\\
			\hline

			Random Walk & $x_0 = 4, n_0 = 20$ & $-40.00$ & $7.00$ & $-42.50$ & $7.07$ & $-42.73$ & $24.12$\\
			& $x_0 = 8, n_0 = 20$ & $-30.00$ & $6.96$ & $-32.50$ & $6.96$ & $-32.43$ & $20.86$\\
			& $x_0 = 12, n_0 = 20$ & $-20.00$ & $7.09$ & $-22.50$ & $7.93$ & $-22.45$ & $17.37$\\
			\hline

			2D Robot & 	$x_0 = 100, y_0 = 40$ & $8.23 \times 10^3$ & $20.11$ & $8.11 \times 10^3$ & $20.03$ & $8.17 \times 10^3$ & $6.19 \times 10^3$\\
			&			$x_0 = 100, y_0 = 60$ & $4.15 \times 10^3$ & $20.16$ & $4.02 \times 10^3$ & $20.13$ & $4.08 \times 10^3$ & $3.85 \times 10^3$\\
			&			$x_0 = 100, y_0 = 80$ & $1.45 \times 10^3$ & $20.15$ & $1.32 \times 10^3$ & $20.13$ & $1.36 \times 10^3$ & $2.07 \times 10^3$\\
			\hline

			Goods Discount & 	$n_0 = 100, d_0 = 1$ & $46.30$ & $8.42$ & $37.89$ & $8.45$ & $41.43$ & $4.23$\\
			&					$n_0 = 150, d_0 = 1$ & $11.63$ & $8.43$ & $2.56$ & $8.43$ & $6.48$ & $3.76$\\
			&					$n_0 = 200, d_0 = 1$ & $-23.02$ & $8.46$ & $-32.77$ & $8.43$ & $-28.43$ & $3.33$\\
			\hline

			Pollutant Disposal &	$n_0 = 50$ & $2.01 \times 10^3$ & $10.04$ & $1.53 \times 10^3$ & $9.85$ & $1.66 \times 10^3$ & $1.00 \times 10^3$\\
			&						$n_0 = 80$ & $2.74 \times 10^3$ & $9.78$ & $2.25 \times 10^3$ & $9.88$ & $2.38 \times 10^3$ & $1.11 \times 10^3$\\
			&						$n_0 = 200$ & $2.04 \times 10^3$ & $9.75$ & $1.56 \times 10^3$ & $9.78$ & $1.69 \times 10^3$ & $1.57 \times 10^3$\\

			\bottomrule
		\end{tabular}
	}
	\vspace{5cm}
	\label{tab:exprnew}
\end{table*}

In a third experiment, we simulated each example program over $20$ different initial valuations ($1000$ simulated executions for each valuation) to experimentally compare upper and lower bounds obtained using our approach with simulations (Figures~\ref{com:btcmining}--\ref{com:pollutant}).

\begin{figure}[H]
	\centering
	\includegraphics[width=.75\linewidth]{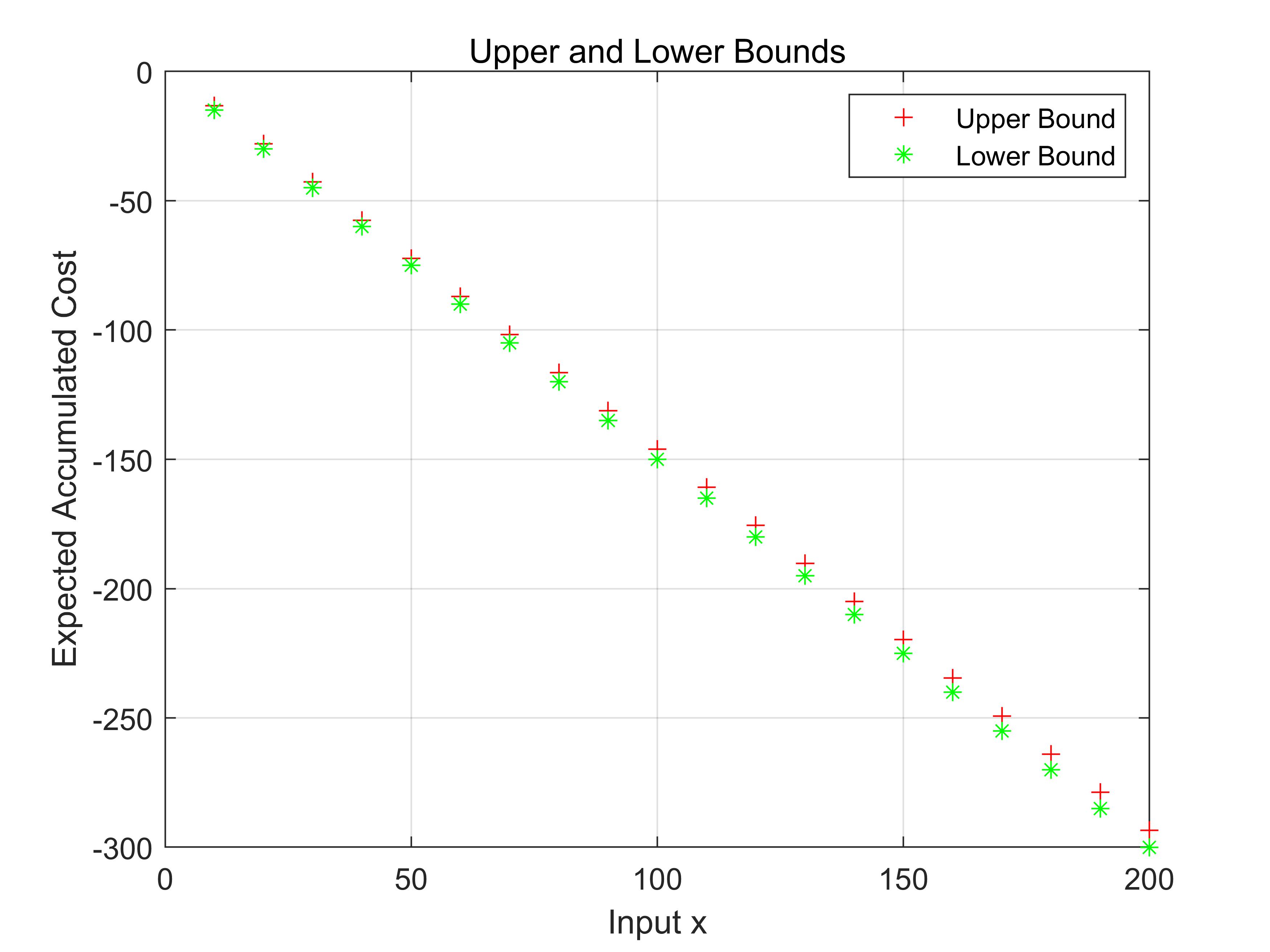}
	\caption{Bitcoin Mining}
	\label{com:btcmining}
\end{figure}

\begin{figure}[H]
	\centering
	\includegraphics[width=.75\linewidth]{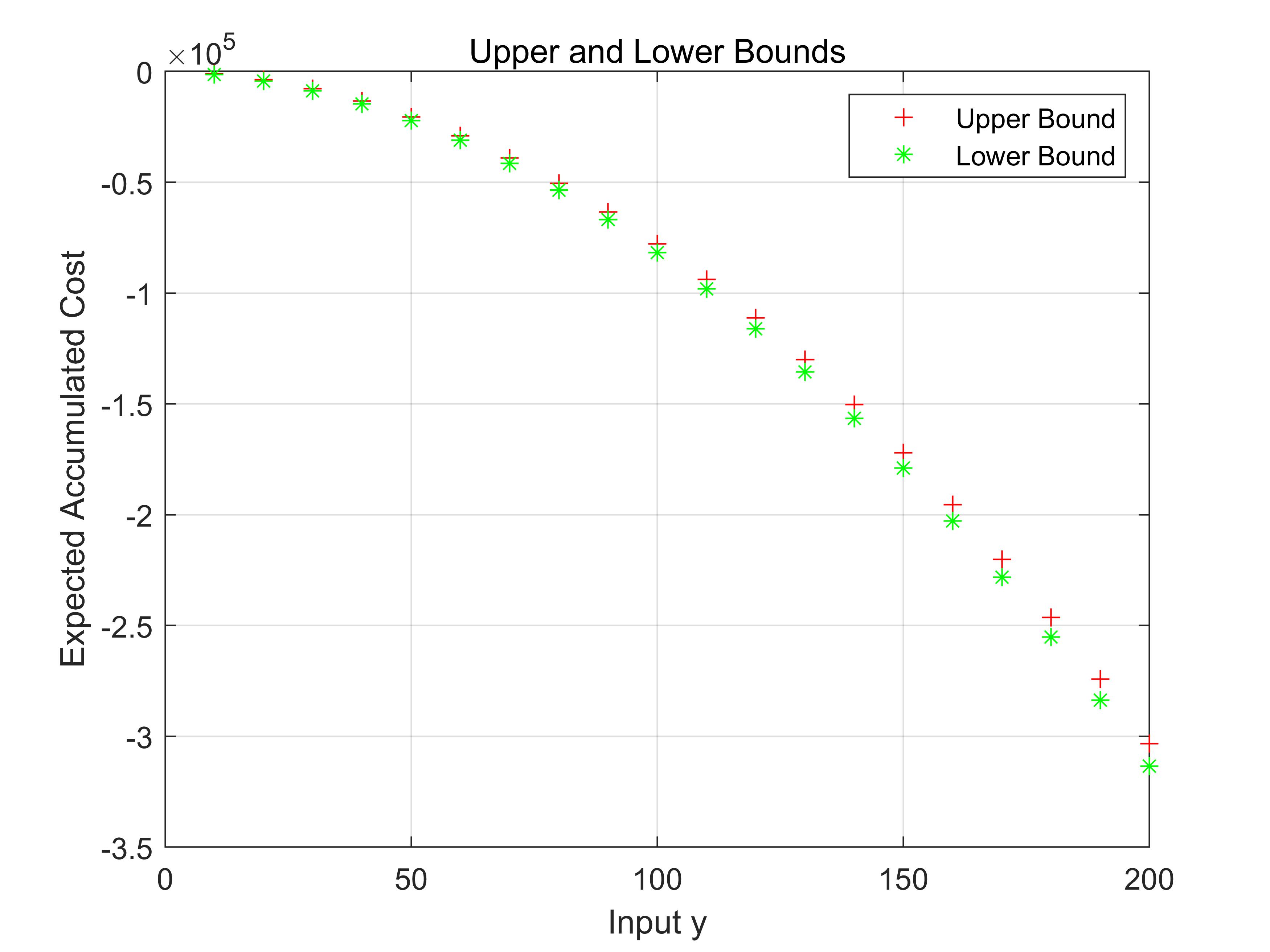}
	\caption{Bitcoin Mining Pool}
	\label{com:btcminingp}
\end{figure}

\begin{figure}[H]
	\centering
	\includegraphics[width=.75\linewidth]{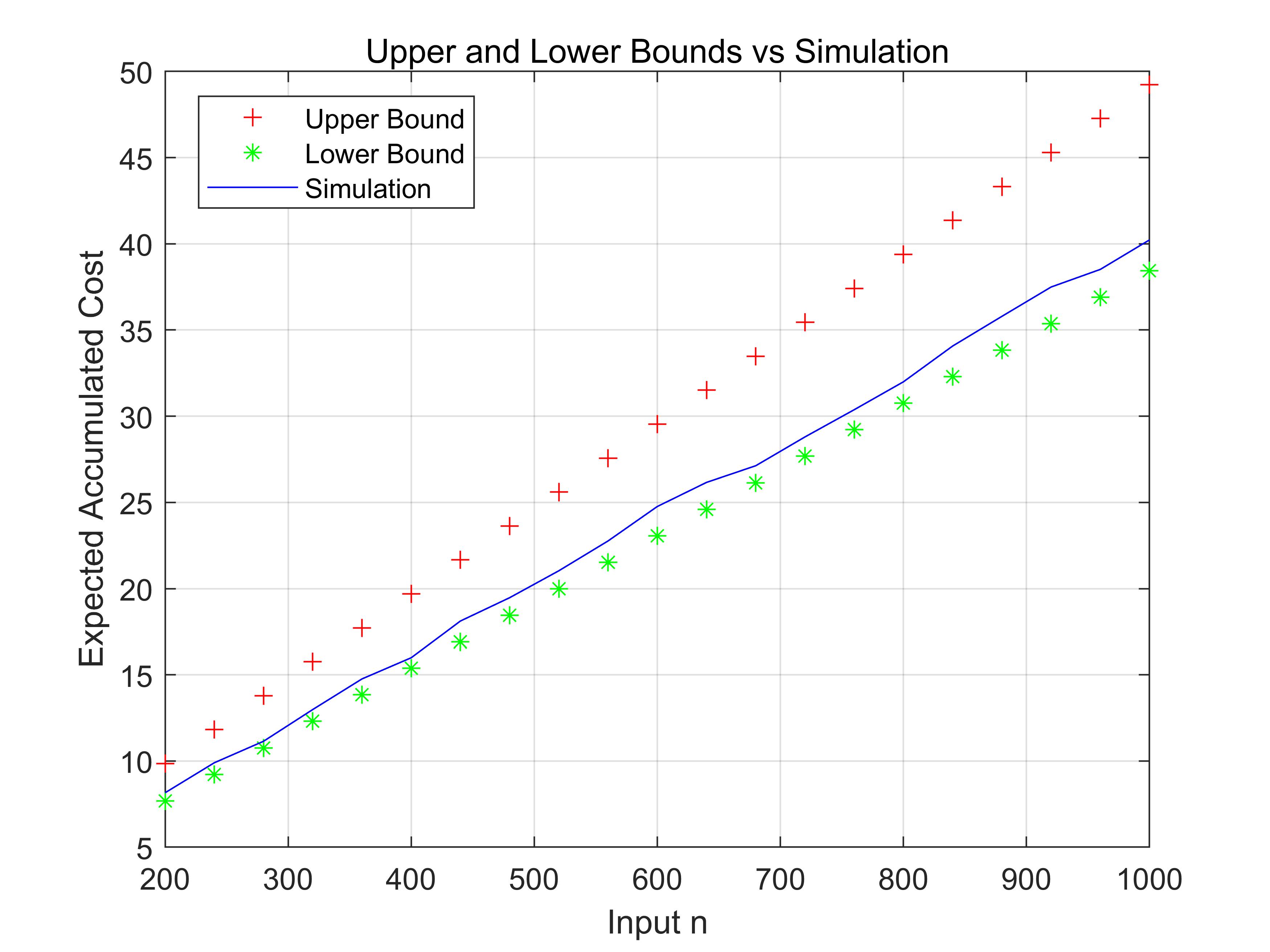}
	\caption{Queuing Network}
	\label{com:queuing}
\end{figure}

\begin{figure}[H]
	\centering
	\includegraphics[width=.75\linewidth]{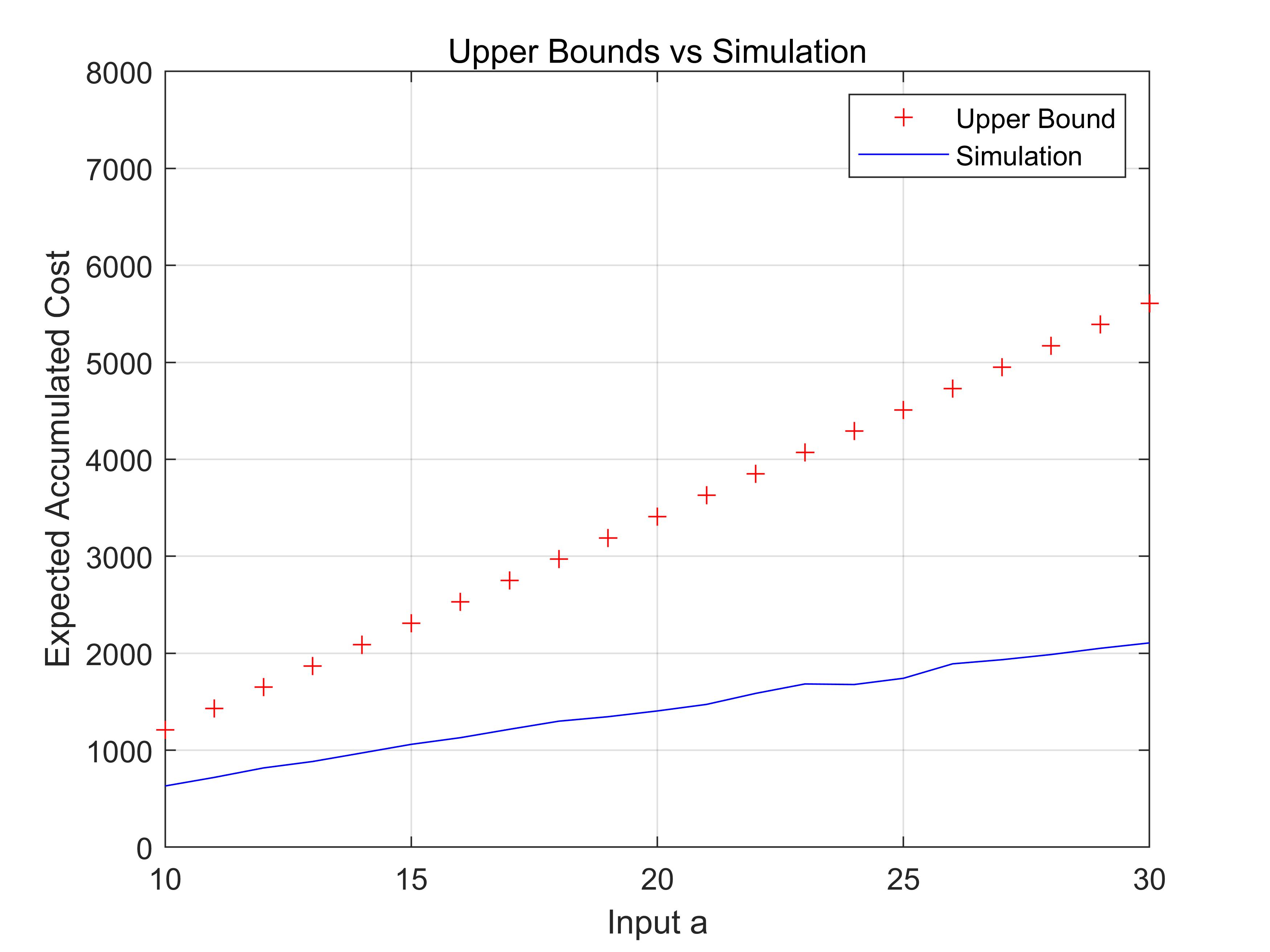}
	\caption{Species Fight}
	\label{com:fight}
\end{figure}

\begin{figure}[H]
	\centering
	\includegraphics[width=.75\linewidth]{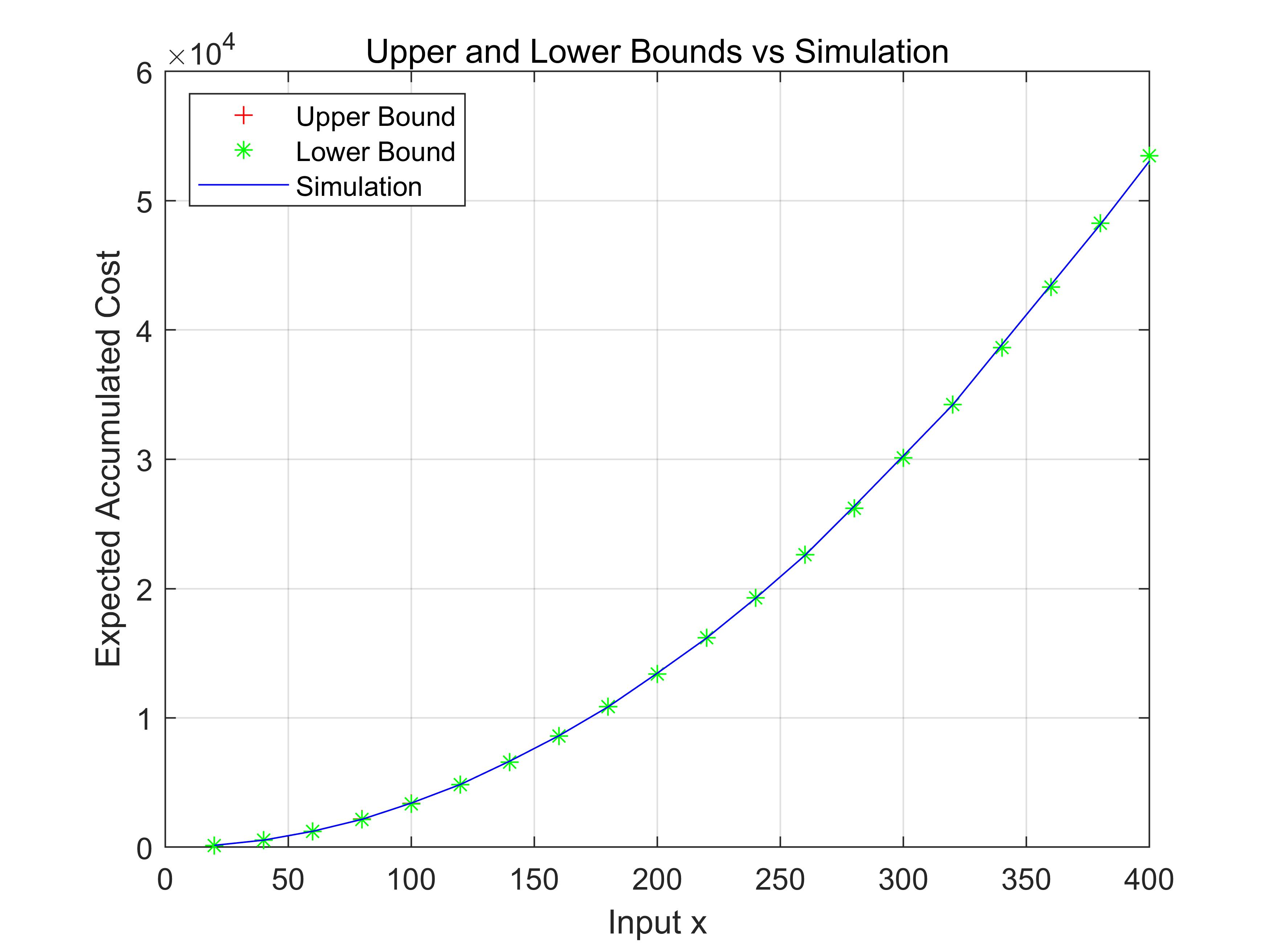}
	\caption{Simple Loop}
	\label{com:simple}
\end{figure}

\begin{figure}[H]
	\centering
	\includegraphics[width=.75\linewidth]{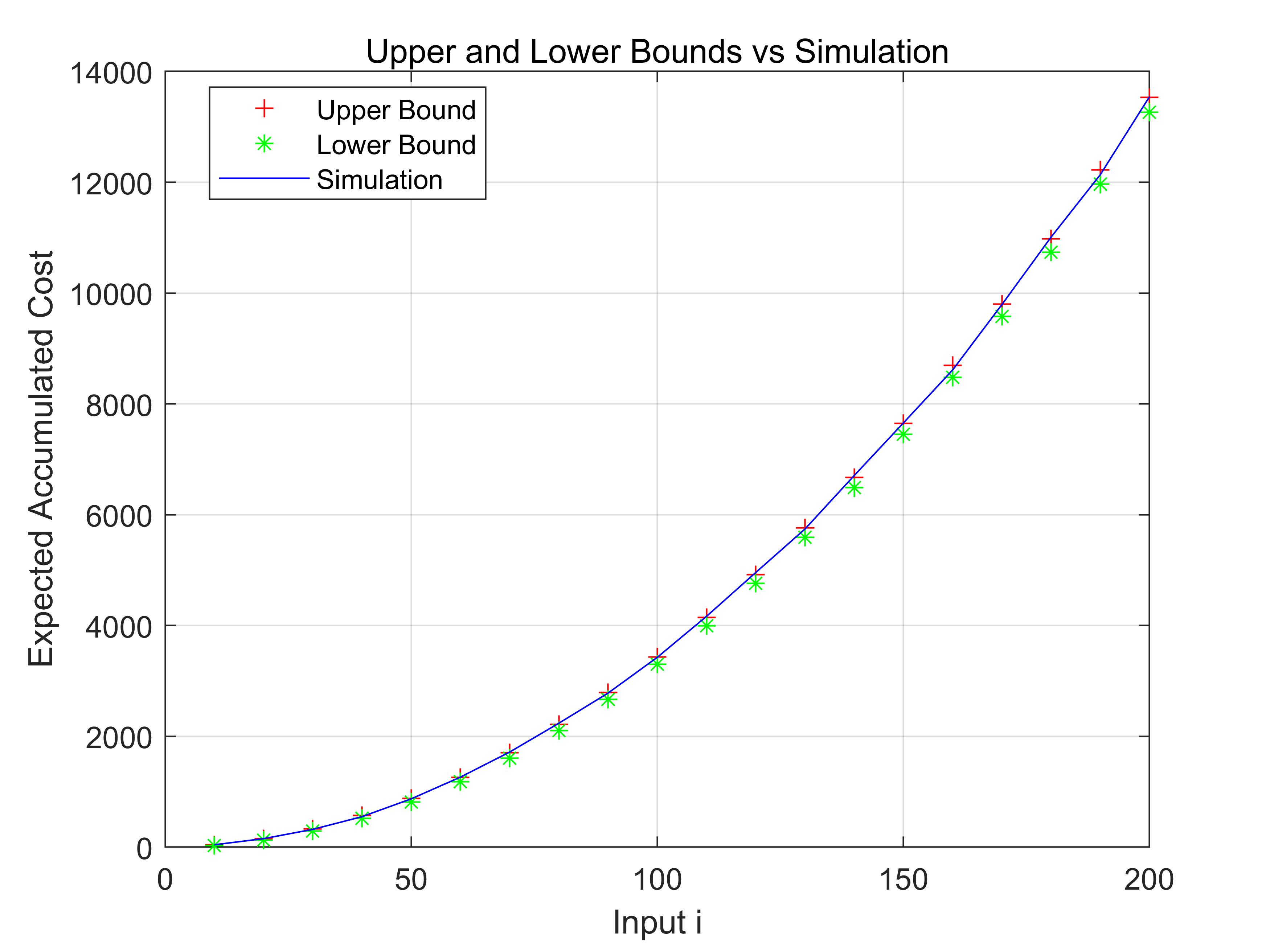}
	\caption{Nested Loop}
	\label{com:complex}
\end{figure}

\begin{figure}[H]
	\centering
	\includegraphics[width=.75\linewidth]{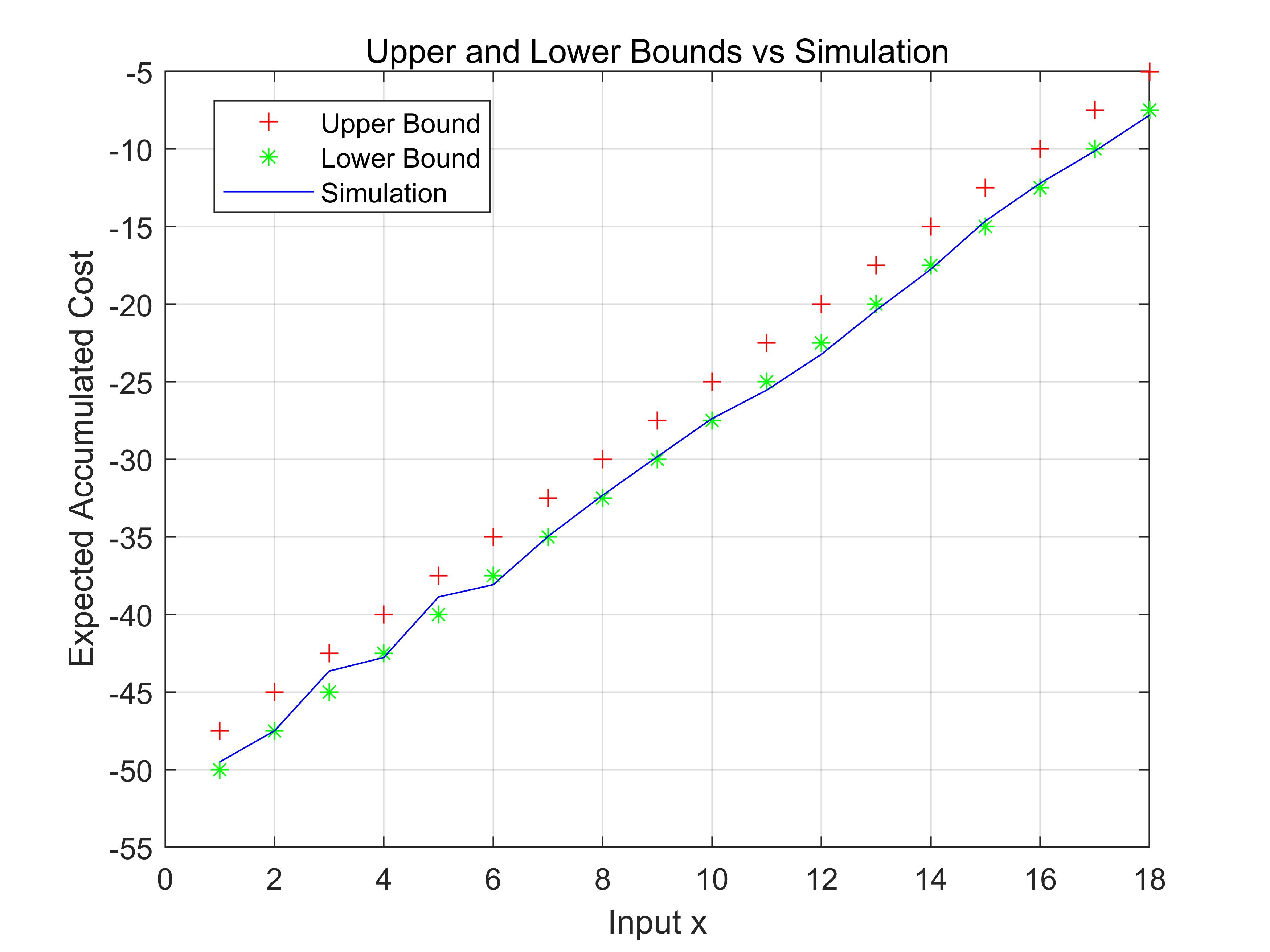}
	\caption{Random Walk}
	\label{com:rdwalk}
\end{figure}

\begin{figure}[H]
	\centering
	\includegraphics[width=.75\linewidth]{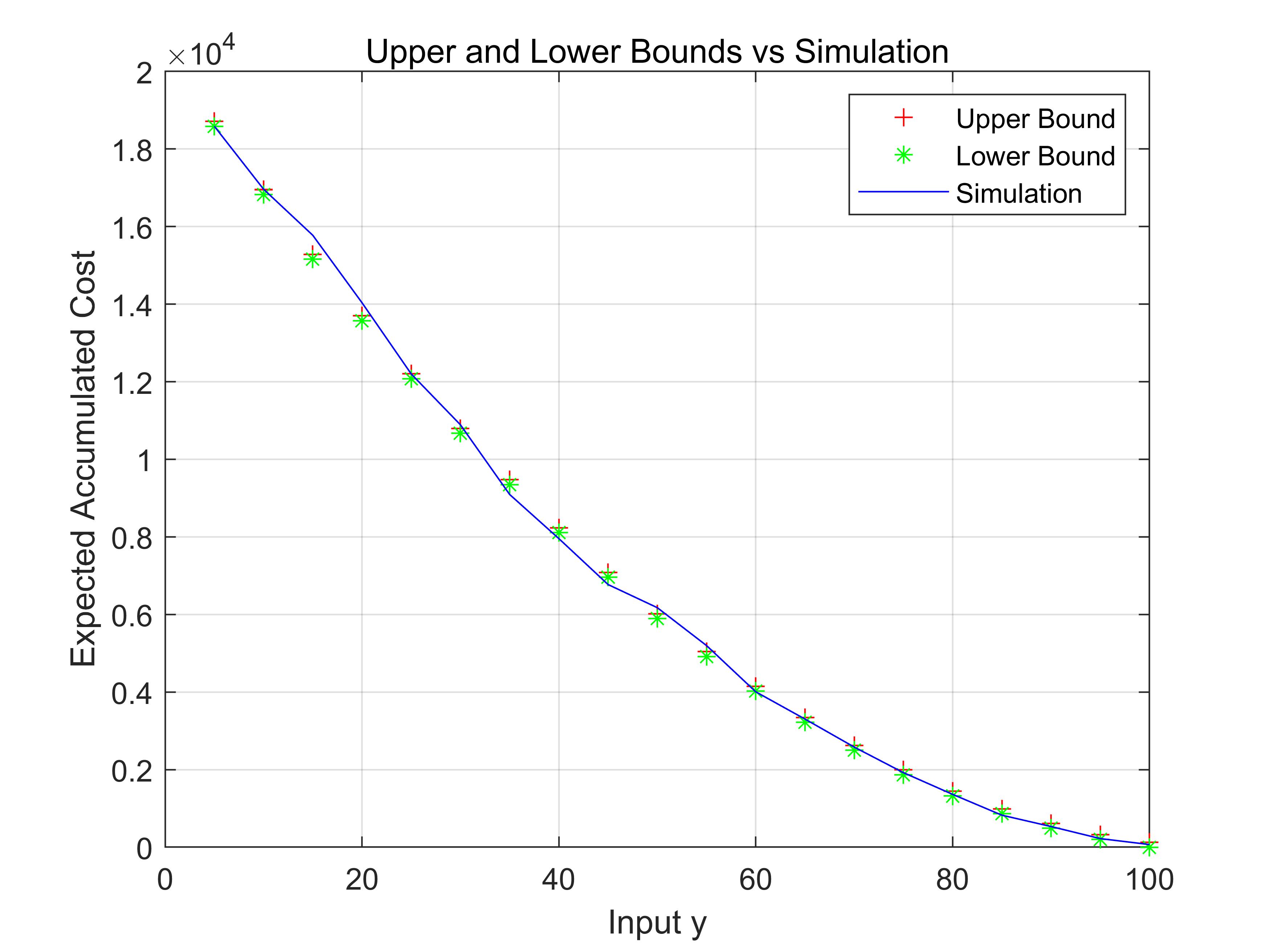}
	\caption{2D Robot}
	\label{com:robot}
\end{figure}

\newpage

\begin{figure}[H]
	\centering
	\includegraphics[width=.75\linewidth]{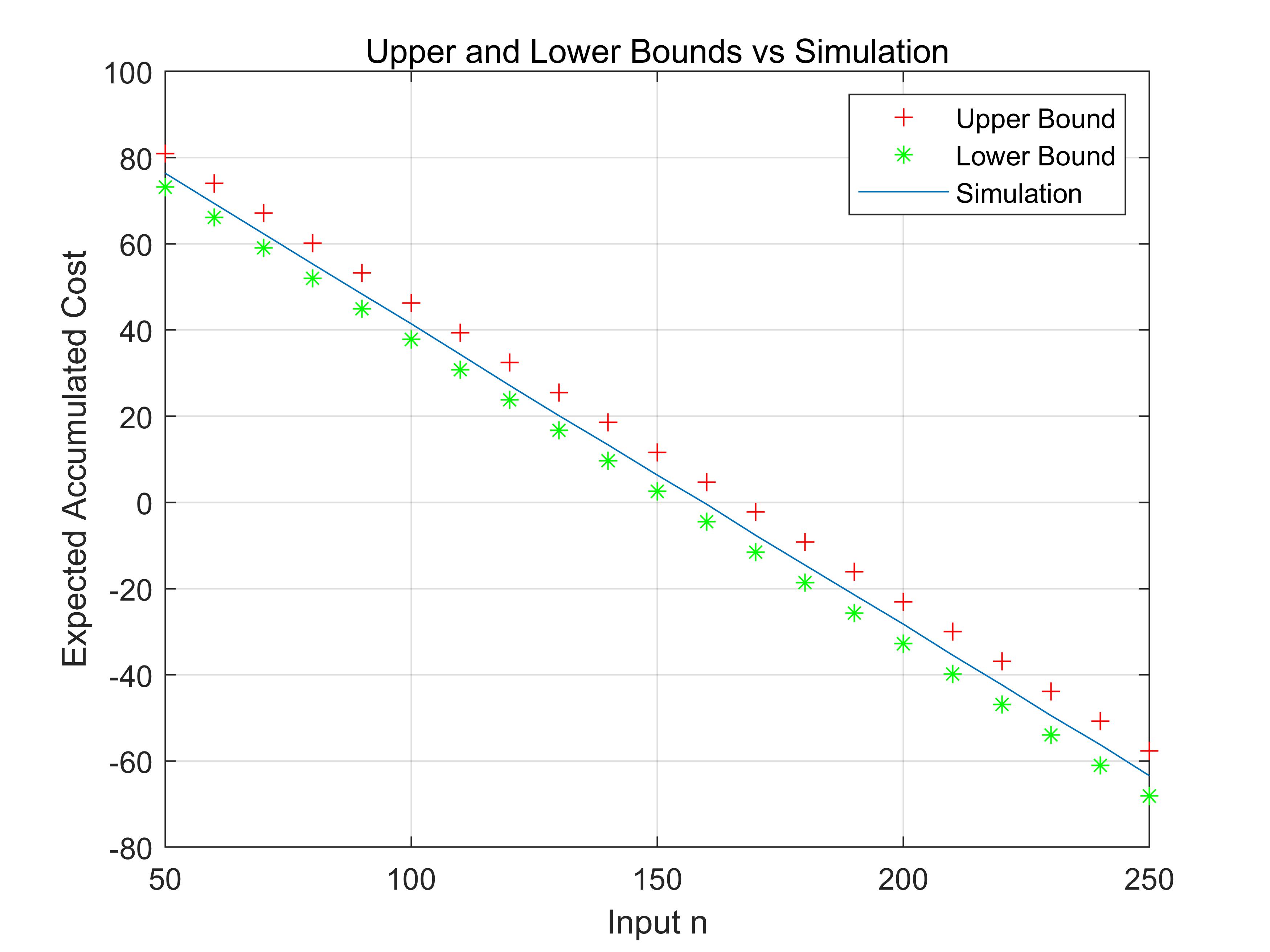}
	\caption{Goods Discount}
	\label{com:goods}
\end{figure}

\begin{figure}[H]
	\centering
	\includegraphics[width=.75\linewidth]{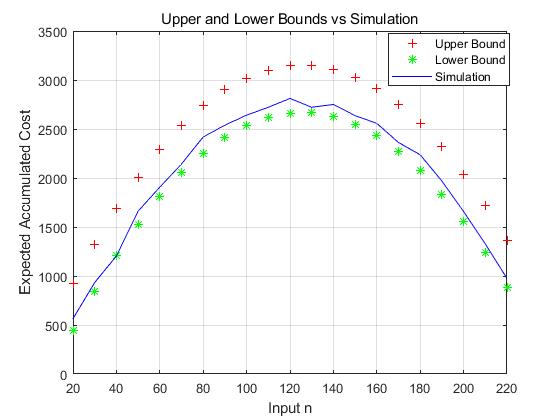}
	\caption{Pollutant Disposal}
	\label{com:pollutant}
\end{figure}

\end{document}